%% file: main_arxiv.tex
\documentclass{article}
\usepackage[utf8]{inputenc}
\usepackage{natbib}

\input{header.tex}

\allowdisplaybreaks

\usepackage[margin=3cm]{geometry}
\usepackage{times}

\pgfplotsset{compat=1.16}
\begin{document}

\title{Collective Outlier Detection and Enumeration with Conformalized Closed Testing}

\author{Chiara G. Magnani\footnote{Department of Economics, Management and Statistics, University of Milano-Bicocca, Milano, Italy.} \footnote{Email: c.magnani9@campus.unimib.it} \and Matteo Sesia\footnote{Department of Data Sciences and Operations, University of Southern California, Los Angeles, California, USA.} \and Aldo Solari\footnote{Department of Economics, Ca' Foscari University of Venice, Venice, Italy; Department of Economics, Management and Statistics, University of Milano-Bicocca, Milano, Italy.}}
\date{\today}

\maketitle

\input{abstract.tex}


\input{body}

\newpage

\noindent
\textbf{Software Availability.} \url{https://github.com/msesia/conformal_nout}.

\medskip

\noindent
\textbf{Data Availability.} The LHCO data that support the findings of this study are openly available at \url{https://lhco2020.github.io/homepage/}; see \citet{kasieczka2021lhc}.

\medskip

\noindent
\textbf{AI Usage.} ChatGPT-5.2 for language editing.

\medskip

\noindent
\textbf{Funding.}
Co-funded by the European Union (ERC, BigBayesUQ, project number: 101041064).
Views and opinions expressed are however those of the author(s) only and do not necessarily reflect those of the European Union or the European Research Council. Neither the European Union nor the granting authority can be held responsible for them.

\bibliographystyle{chicago}
\bibliography{biblio}

\clearpage

\appendix

\input{appendix}

\end{document}

%% file: header.tex
\usepackage{amsthm}
\usepackage{tikz}
\usepackage{pgfplots}
\usepackage{dsfont}
\usepackage{comment}
\usepackage[linesnumbered,ruled,vlined]{algorithm2e}
\usepackage{multirow}
\usepackage{amsfonts}
\usepackage{amsmath}
\usepackage{enumitem} 
\usepackage{bm}
\usepackage{pgfplots}
\usepackage{subcaption}
\usepackage{float}
\usepackage{tabularx,booktabs,ragged2e}
\usepackage{xcolor}
\usepackage{placeins}
\usepackage[hidelinks]{hyperref}
\usepackage{cancel}

\definecolor{darkgreen}{rgb}{0.13, 0.55, 0.13}

\newtheorem{theorem}{Theorem}
\newtheorem{proposition}{Proposition}

\newtheorem{lemma}{Lemma}
\newtheorem{corollary}{Corollary}

\theoremstyle{definition}

\newcommand{\I}[1]{\mathds{1}[#1]}
\newcommand{\E}[1]{\mathbb{E}[#1]}

%% file: abstract.tex
\begin{abstract}
 This paper develops a flexible distribution-free method for collective outlier detection and enumeration, designed for situations in which the presence of outliers can be detected powerfully even though their precise identification may be challenging due to the sparsity, weakness, or elusiveness of their signals. This method builds upon recent developments in conformal inference and integrates classical ideas from other areas, including multiple testing, locally most powerful and adaptive rank tests, and non-parametric large-sample asymptotics. The key innovation lies in developing a principled and effective approach for automatically choosing the most appropriate machine learning classifier and two-sample testing procedure for a given data set. The performance of our method is investigated through extensive empirical demonstrations, including an analysis of the LHCO high-energy particle collision data set.
\end{abstract}


%% file: body.tex
\section{Introduction}  \label{sec:intro}

\subsection{Background and Motivation}

Outlier detection is a fundamental statistical problem with numerous applications, ranging from many areas of scientific research to fraud detection and security monitoring.
In the age of AI, its relevance is further growing due to concerns over the rise of potentially unrealistic synthetic data. 
The complexity of high-dimensional data has encouraged the use of sophisticated machine learning models for outlier detection, but these models often lack transparency and are prone to errors, complicating their reliability. 
This has spurred substantial interest in conformal inference \citep{vovk2005algorithmic}, which can provide principled statistical guarantees for any outlier detection algorithm under relatively mild assumptions.

While conformal inference is gaining momentum in both statistics and machine learning, its focus in the context of outlier detection has so far primarily been on individual-level outlier {\em identification}, where each data point is separately evaluated as a potential outlier. However, individual-level identification is not always feasible. In many applications, outliers may be too rare or weak to reach statistical significance \citep{donoho2004higher}, or they may exhibit unremarkable behavior in isolation but reveal anomalous patterns when analyzed collectively \citep{feroze2021group}, such as showing under-dispersion relative to inliers.

To extend the applicability of conformal inference to these particularly challenging settings, this paper introduces a novel approach for {\em collective outlier detection} that addresses two main challenges: (1) testing the global null hypothesis that a dataset---or a subset of it---contains no outliers, and (2) estimating the number of outliers present. Our method solves these interrelated problems by leveraging in an innovative way powerful ``black-box'' machine learning algorithms and highly flexible, data-driven test statistics.
This integrative approach is tailored to maximize power for the data at hand and offers reliable type-I error control under mild assumptions. 
Additionally, our method is scalable to large data sets. It builds on recent advances in conformal inference, while also revisiting and integrating several classical concepts from other areas of statistics. As we will demonstrate, this collective approach succeeds even in scenarios where individual outlier detection fails.

Our method is broadly applicable. Financial institutions, for instance, could use it to identify complex fraud schemes involving seemingly legitimate transactions. Similarly, cybersecurity systems could apply it to detect coordinated denial-of-service attacks \citep{ahmed2014network}. 
In high-energy physics, where collective outlier detection generally plays a central role  \citep{vatanen2012semi}, researchers could use our method to sift through massive datasets from particle decay sensors in the search for new particles. These signals are typically rare and weak, making individual-level detection impractical, but they present a promising use case for our approach, as previewed in Figure~\ref{fig:exp-lhco} and Section~\ref{sec:numerical-lhco}.

\subsection{Contributions and Outline} \label{sec:contributions}

We present ACODE, an {\em Automatic Conformal Outlier Detection and Enumeration} method. 
ACODE leverages two-sample rank tests applied to univariate {\em conformity scores}, which can be generated by any model trained to distinguish outliers from inliers. By employing split-conformal inference, ACODE converts these scores into statistical tests controlling the type-I error rate.
In addition, by utilizing the closed testing principle \citep{marcus1976closedtesting}, ACODE can construct simultaneous lower confidence bounds for the number of outliers within any subset of the test sample \citep{goemansolari2011}, while also offering a global test for outlier detection as a byproduct of outlier enumeration. 
While these high-level ideas are intuitive, the flexibility of the conformal inference framework introduces substantial complexity and a great deal of implementation freedom. This brings us to the two key methodological questions addressed in this paper.

The first question involves selecting the most effective classification algorithm for computing conformity scores. 
We specifically examine the trade-offs between existing approaches based on one-class classifiers \citep{bates2021testing} and positive-unlabeled learning via binary classification \citep{marandon2022machine}. 
The second question concerns the choice of rank test to determine statistical significance based on these scores.
As these issues are intertwined and dependent on specific data, we present a principled, data-driven solution. The effectiveness of this approach is showcased in Figure~\ref{fig:exp-lhco}, analyzing data from the 2020 Large Hadron Collider Olympics (LHCO) \citep{kasieczka2021lhc}.

\begin{figure}[!htb]
\centering
\includegraphics[width=\linewidth]{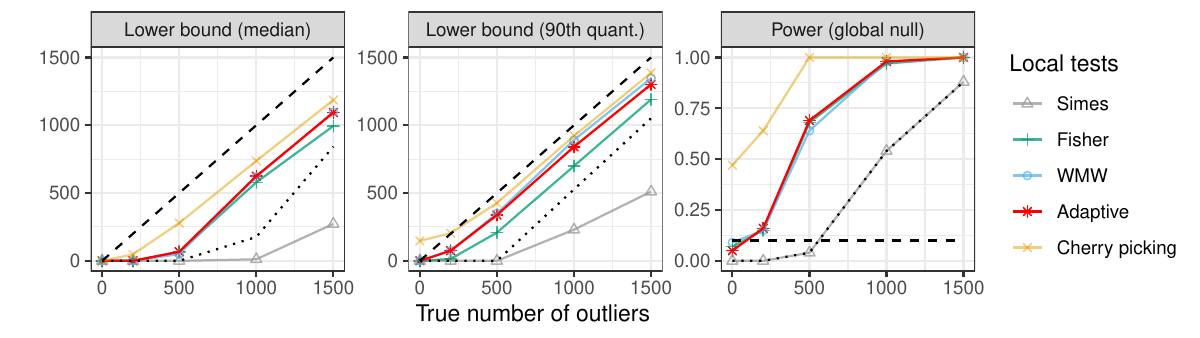}
\caption{Preview of performance of ACODE on the LHCO data.
Left: Median and 90th-quantile over repeated experiments for a 90\% lower confidence bound for the number of outliers.
Right: power against the global null hypothesis of no outliers at the 10\% level (horizontal line).
The results are shown as a function of the true number of outliers in a test set of cardinality 10,000.
ACODE utilizes a testing procedure that may be adaptively selected (red curve) or fixed (other solid curves).
Dotted curve: number of {\em individual} discoveries or power obtained by applying the Benjamini-Hochberg procedure (BH) to conformal $p$-values, controlling the false discovery rate below 10\%.
}
\label{fig:exp-lhco}
\end{figure}

As detailed in Section~\ref{sec:numerical-lhco}, our analysis of the LHCO data aims to (1) detect outliers, or unusual collision events, within a test sample, and (2) establish a 90\% lower confidence bound for the number of these outliers. Our method not only controls errors for both tasks but also provides insightful inferences. This illustrates ACODE's adaptability, as it selects the most suitable testing procedures in a data-driven manner. 

Additionally, Figure~\ref{fig:exp-lhco} highlights the challenge of conducting an adaptive data analysis without incurring selection bias. 
It compares ACODE's performance with that of a naive ``cherry-picking" approach that greedily tests various classification algorithms and testing procedures, selecting the one with the most appealing outcomes. Unsurprisingly, that can lead to inflated type-I errors because the same data are used both to select the classification algorithm and testing procedure and to assess statistical significance.
In contrast, ACODE is protected against selection bias while often achieving similarly high power.

Figure~\ref{fig:exp-lhco} also highlights the potential advantages of an approach specifically designed for collective outlier detection over standard alternatives for individual-level discovery.
The results show that ACODE leads to much more informative collective inferences compared to the Benjamini-Hochberg (BH) procedure for False Discovery Rate (FDR) control \citep{benjamini1995controlling} applied to individual conformal $p$-values \citep{bates2021testing, marandon2022machine}.
In fact, the BH procedure aims to identify outliers individually, and is generally not an effective solution for global testing or outlier enumeration.

The novel contributions of this paper are the following.
\begin{enumerate}
    \item We introduce ACODE, a flexible method for outlier enumeration with data-driven selection of the machine-learning classifier and the two-sample testing procedure. This method is built on a novel integration of existing theoretical results on closed testing (Proposition~\ref{prop:closed-testing}), locally most powerful rank tests (Theorem~\ref*{thm:LMPI}), and data-driven adaptivity via sample splitting (Proposition~\ref*{prop:adaptive_exchangeability}), and yields rigorous inferences.
    
    \item We present asymptotic approximations to the null distributions of the locally most powerful rank test and its adaptive version under exchangeable conformity scores (Theorems~\ref{thm:asymp_distr_LMPI-exch} and \ref{thm:asymp_distr_adaptive-exch}), extending earlier results that assume i.i.d.~scores.  
    
    \item We implement an adaptive rank test, which nonparametrically estimates the unknown outlier distribution, that attains the same local asymptotic power as the locally most powerful rank test with oracle knowledge of the outlier distribution (Theorem~\ref{thm:asymp_equivalence}).  

\end{enumerate}

This paper is structured as follows. Section~\ref{sec:related-work} reviews relevant prior literature, Section~\ref{sec:methods} describes our method, and Section~\ref{sec:numerical} demonstrates its performance on synthetic and real data.
Section~\ref{sec:discussion} concludes by suggesting some directions for future research.
Additional methodological aspects, implementation-specific details of the method, computational shortcuts, and further numerical results are provided in the Appendices.

\subsection{Related Work} \label{sec:related-work}

Conformal inference \citep{vovk2005algorithmic} is an active research topic, broadly seeking reliable uncertainty estimation for the predictions output by black-box machine learning models.
In the context of outlier detection \citep{laxhammar2015inductive,guan2022prediction}, prior works focused on individual-level identification, often aiming to control the FDR \citep{bates2021testing,marandon2022machine}.
By contrast, we focus on collective detection and enumeration; this tends to lead to more informative results when dealing with weak or sparse signals that are difficult to localize.
This approach also connects our work to a broader literature on multiple testing and large-scale inference.

Recent advances in large-scale inference have led to methods for signal detection and estimation of the proportion of non-null effects \citep{cai2017large}. 
In particular, the Higher Criticism test \citep{donoho2004higher} has been extended to provide simultaneous bounds on the false discovery proportion (FDP) via closed testing  \citep{ goeman2021onlyclosed}. However, the FDP bounds derived from Higher Criticism are valid only under the assumption of independent $p$-values, which does not hold in the conformal inference setting.

In this paper, we integrate conformal inference with closed testing to construct simultaneous lower confidence bounds for the number of outliers in any subset of the test set. A central challenge is selecting an effective local testing procedure suited to the data at hand. 

When a uniformly most powerful test is unavailable, a possible strategy is to maximize power locally around the null hypothesis. 
In a seminal paper, \citet{lehmann53power} showed that the Wilcoxon–Mann–Whitney (WMW) test \citep{wilcoxon1945individual, mannwhitney1947} is the \emph{locally most powerful rank} test under the alternative in which the outlier distribution is given by the maximum of two independent copies of the inlier distribution. This result was later extended by \citet{shiraishi1985local}, who characterized locally most powerful rank tests for arbitrary outlier distributions, but at the cost of requiring their specification.

In practice, however, the assumption that the outlier distribution is known is unrealistic. Rank tests are distribution-free under the null and therefore exact regardless of the underlying distribution, offering maximal protection against model misspecification. Their power, however, depends on the true outlier distribution and may be sensitive to misspecification, potentially yielding a locally optimal test for the wrong model.

This dependence can be mitigated through \emph{adaptation}.
The literature on \emph{adaptive rank tests} is extensive; here we focus on two basic paradigms, following \cite{sidak1999theory}.
The \emph{restrictive} paradigm \citep{hogg1975two} requires to fix a reasonable family of distributions for outliers and select one, based on the data, then perform the corresponding rank test. If the true outlier distribution is close to the defined family, the performed test tends to be nearly optimal. 
In the \emph{non-restrictive paradigm}, the outlier distribution is first estimated from the data and then used to construct a test that is optimal for the estimated model. Remarkably, this strategy can succeed: in Theorem \ref{thm:asymp_equivalence}, we show that the resulting adaptive rank test attains the same local asymptotic power as the locally most powerful rank test that has oracle knowledge of the outlier distribution.

ACODE combines these two approaches by considering both a prespecified family of rank tests and an adaptive rank test based on density estimation, and automatically selecting the procedure with the best empirical performance.

A typical concern with closed testing is the computational cost, generally exponential. However, efficient polynomial-time \emph{shortcuts} are often available \citep{goeman2019simultaneous, tian2023Large-scale}.
In Section~\ref{sec:closed-testing} and in Appendix~\ref*{app:shortcuts} we, respectively, present the closed testing method and the computational shortcuts that make our method feasible and scalable. As we will see, different local tests often require different shortcuts and have different costs.

\section{Methodology} \label{sec:methods}

Section~\ref{sec:setup} starts by formally stating the problem. Section~\ref{sec:conf-scores} introduces the learning component of ACODE, detailing how a classification model is trained and applied to rank the test points and a subset of the reference inliers. Section~\ref{sec:local-tests} discusses using these ranks to test local hypotheses about the presence of outliers within specific subsets of test points. Section~\ref{sec:closed-testing} describes the integration of these local tests into a closed testing framework, enabling the derivation of simultaneous confidence bounds for the number of outliers in any test data subset. Section~\ref{sec:conformal-closed-testing} summarizes the overall data analysis pipeline and introduces the enhancements described in Section~\ref{sec:tuning_new}, which make ACODE more adaptive.

\subsection{Setup and Problem Statement} \label{sec:setup}

Consider a {\em reference} data set $Z_1, \ldots, Z_{m} \in \mathbb{R}^d$ comprising $m$ observations each with dimensions $d \geq 1$, where $d$ may be large compared to $m$.
We assume these are {\em independent and identically distributed} (i.i.d.)~random samples from some unknown distribution $P_0$ on $\mathbb{R}^d$.
Consider also a {\em test} set $Z_{m+1}, \ldots, Z_{m+n}$ and, for each $j \in [n] := \{1,\ldots,n\}$, assume $Z_{m+j} \in \mathbb{R}^d$ is independently sampled from some unknown distribution $P_j$, which may or may not be equal to $P_0$. That is,
\begin{align} \label{eq:model}
  & Z_{1},\ldots,Z_{n} \overset{\text{i.i.d.}}{\sim} P_0,
  & Z_{m+j} \overset{\text{ind.}}{\sim} P_j, \; \forall j \in [m].
\end{align}

In the following, we will refer to the data points sampled from $P_0$ as {\em inliers} and those from any $P_j \neq P_0$ as {\em outliers}.
For each $j \in [m]$, define the null hypothesis $H_j: P_j = P_0$. Let $I_0 := \{j \in [n] : P_j = P_0 \}$ denote the subset of inliers in the test set.
Similarly, let $I_1 = [n] \setminus I_0$ denote the subset of outliers in the test set.

Prior work on conformal inference for outlier detection \citep{bates2021testing, marandon2022machine} has primarily focused on testing individual-level hypotheses $H_j$ for each $j \in [m]$, with the goal of controlling the FDR. However, as discussed in Section~\ref{sec:intro}, there are many scenarios where the power to reject these individual-level hypotheses $H_j$ is very low, highlighting the need for an alternative approach.

In this paper, we address two key tasks related to collective outlier detection and enumeration. The first task, {\em outlier detection}, involves testing the intersection hypothesis $H_{S} := \cap_{j \in S} H_j$, which posits that a subset $S \subseteq [n]$ of test points contains no outliers. The second task, {\em outlier enumeration}, focuses on estimating the number of outliers within a subset $S \subseteq [n]$, specifically by providing a lower confidence bound, without necessarily identifying them.
For both tasks, we consider cases where $S$ is either fixed or data-driven. In the latter scenario, we will utilize the closed testing principle 
to obtain {\em simultaneous} inferences that remain valid for all possible subsets $S$ \citep{goemansolari2011}. This approach gives practitioners flexible tools for interactive data analysis.

To deal with the possibly high-dimensional nature of the data and avoid parametric assumptions about their distribution, we will tackle the aforementioned tasks using a conformal inference framework. Essentially, we will extract one-dimensional {\em conformity} scores from the data by leveraging powerful machine learning algorithms trained to distinguish outliers from inliers. Then, we will obtain {\em distribution-free} inferences by applying appropriate {\em rank tests} to these scores.
With this overall strategy in view, the following section begins with a concise overview of the relevant conformal inference background.

\subsection{Computing Conformity Scores and Conformal $p$-values} \label{sec:conf-scores}

Let $(X_1,\ldots,X_m)$ and $(Y_1,\ldots,Y_n)$ indicate two vectors of conformity scores obtained by applying a function $\hat{s}: \mathbb{R}^d \rightarrow \mathbb{R}$ to each calibration and test point, respectively.
In general, $\hat{s}$ may depend on an independent training data set $\mathcal{D}_{\text{train}}$ and possibly also on the unordered collection of observations $\{Z_1, \ldots, Z_{n+m}\}$.
That is, for any $i \in [m]$ and $j \in [n]$,
\begin{align} \label{eq:scores}
& X_i = \hat{s}(Z_i; \mathcal{D}_{\text{train}}, \{Z_1, \ldots, Z_{n+m}\}),
& Y_j = \hat{s}(Z_{m+j}; \mathcal{D}_{\text{train}}, \{Z_1, \ldots, Z_{n+m}\}).
\end{align}
We refer to Section~\ref{sec:conformal-closed-testing} for further details on the computation of the conformity scores.
Here, it is important to point out that Equation~\eqref{eq:scores}, combined with the data generating model defined in~\eqref{eq:model}, implies that the scores $(X_1,\ldots,X_m, (Y_j)_{j \in I_0})$ are exchangeable \citep{marandon2022machine}.
Additionally, it is easy to see that $(X_1,\ldots,X_m, (Y_j)_{j \in I_0})$ are mutually independent if the function $\hat{s}$ does not depend on $\{Z_1, \ldots, Z_{n+m}\}$.

A typical approach is to leverage these scores to compute the following {\em conformal $p$-value} $p_j$ for the null hypothesis $H_j$, for all $j \in [n]$:
\begin{align}\label{eq-permp}
p_j &= \frac{1}{(m+1)}\left(1+ \sum_{i = 1}^{m} \mathds{1}\{ X_i \geq Y_{j} \}\right).
\end{align}
As long as $(X_1,\ldots,X_m, (Y_i)_{i \in I_0})$ are exchangeable, $p_j$ is a valid $p$-value for $H_j$, in the sense that $\mathbb{P}[p_j \leq \alpha \; ;\; H_j] \leq \alpha$ for all $\alpha \in [0,1]$.
Further, $p_j$ is not too conservative, in the sense that $\mathbb{P}[p_j \leq \alpha \; ;\; H_j]$ is not far from $\alpha$, as long as $m$ is large and the conformity scores are almost surely distinct; the latter would be a mild assumption (since one can always add a small amount of independent noise to $\hat{s}$) but is not required.

Testing the individual hypotheses $H_1, \ldots, H_n$ is a ``many-to-one comparisons to a control'' problem \citep{dunnett1955multiple}.
Since the $n$ comparisons of $(Y_j)$ to $(X_1,\ldots,X_m)$ all use the same control sample, the $p$-values $p_1,\ldots,p_n$ are mutually dependent, complicating the application of multiple testing procedures.
\citet{bates2021testing} and \citet{marandon2022machine} proved that $(p_1, \ldots, p_n)$ are \emph{positively regression dependent on a subset} (PRDS) \citep{benjamini2001control}, which implies BH controls the FDR.
Further, \citet{bates2021testing} and \citet{marandon2022machine} have also studied the implications of such positive dependence on the validity of $p$-value based methods designed to test intersection null hypotheses $H_S$, where $S \subseteq [n]$.
This is the starting point of the next section.

\subsection{Local Tests for Outlier Detection} \label{sec:local-tests}

\subsubsection{Existing Approaches}

Consider the null hypothesis, denoted as $H_S$, that a {\em fixed} subset $S \subseteq [n]$ of test points contains no outliers.
As outlined below, there are many tests available for this hypothesis, each with unique strengths and weaknesses. Given the limited scope of existing optimality results and the unfeasibility of predicting which method will perform best on a specific data set, ACODE is designed to integrate a {\em toolbox} of such testing procedures. It then selects the most effective approach in a principled data-driven manner (ref.~Section~\ref{sec:tuning_new}). To decide what to include in this toolbox, we review some relevant literature below.

One approach to testing $H_S$ is to apply standard aggregation methods to the conformal $p$-values.
The combination of $p$-values dates back to the 1930s with Fisher, Tippett, and Pearson, originally under an independence assumption \citep[see][for a review]{owen2009karl}.
Moving beyond independence makes the problem more challenging but is necessary in the context of conformal inference \citep{bates2021testing}.

\citet{gazin2024transductive} characterize the dependence structure of conformal $p$-values, which depends only on the ranks of the conformity scores. Under the null hypothesis, these ranks are exchangeable. Ignoring this structure by using aggregation methods that are valid under arbitrary dependence leads to \emph{inadmissible} procedures---this follows from \citet{vovk2022admissible}, which show that universally valid methods are dominated by the Simes procedure \citep{simes1986improved}, which is applicable to conformal $p$-values \citep{sarkar2008simes}.

While the Simes method tends to work well in scenarios featuring rare and strong signals, it is itself also inadmissible (see Appendix~\ref*{app:local-psimes}), and in practice it is often less powerful than alternative approaches because any $p$-value combination method is optimal against some alternative \citep{birnbaum1954combining}.
For instance, Fisher's method, recently refined by \cite{bates2021testing} to accommodate the positive dependencies of conformal $p$-values, is well-suited for scenarios with numerous weak signals \citep{heard2018choosing}.

A related optimality result is due to \citet{lehmann53power}, which establishes that the WMW test is the \emph{locally most powerful rank} (LMPR) test for the global null hypothesis $\theta=0$ under a nonparametric alternative in which the test scores $Y_j$ follow the mixture distribution 
$(1-\theta)F + \theta F^2$. Here, $\theta \in [0,1]$ is the proportion of outliers and $F$ denotes the cumulative distribution function of the calibration scores $X_i$.
In this model the outliers have \emph{weak} signals, as $F^2$ is the distribution of the maximum of two independent random variables with distribution $F$, where $\mathbb{P}(X_i < Y_j) = 2/3$ for any $F$.
As shown in Appendix~\ref*{app:local-wmw}, the Mann-Whitney statistic for testing
$H_S$ can be formulated as a combination of conformal $p$-values:
$T^{\text{MW}}_S  = (m+1)  \sum_{j \in S}( 1- p_j)$.
Since $T^{\text{MW}}_S$ is a monotonic function of $\bar{p}_S = |S|^{-1}\sum_{j \in S}p_j$, the test statistic can be simplified to the average of conformal $p$-values $\bar{p}_S$.

Finally, permutation tests offer significant flexibility, but in addition to their higher computational cost, which makes them less well-suited to be integrated within a closed-testing framework, they suffer from the same limitation as the aforementioned approaches that in practice it is unclear which test statistic should be used for a given data set.

In this paper, we focus on applying ACODE in conjunction with Simes test, Fisher's combination method, and WMW test, as previewed in Figure~\ref{fig:exp-lhco}. These procedures are chosen for their ease of use within our conformal inference framework and for their complementary strengths against a range of alternative hypotheses, including both sparse/strong and dense/weak effects. We detail these methods in Table~\ref{tab:standard_approaches} in Appendix~\ref*{app:local-testing}, discussing their underlying assumptions, precision, and the scenarios in which they are most effective. Additional, although non-exhaustive, information about existing local testing procedures can be found in Appendix~\ref*{app:local-testing}.

The empirical and theoretical strengths of the WMW test also motivate us to explore in the next section a flexible extension of this approach, inspired by \citet{shiraishi1985local}, that enjoys similar theoretical properties and sometimes achieves higher power in practice.

\subsubsection{Shirashi's Locally Most Powerful Rank (LMPR) Test} \label{sec:new_local-g-wmw}

In this section, let us imagine the calibration scores $X_1,\ldots,X_m$ are drawn independently from an unknown distribution $F$ and the test scores $Y_1,\ldots,Y_n$ are drawn independently from a mixture distribution $(1-\theta)F + \theta G(F)$, where $\theta \in [0,1]$ is the unknown proportion of outliers and $G$ is a distribution function on $[0,1]$ with density function $g$. The global null hypothesis $H_{[n]}: \cap_{j \in [n]} H_j$ of no outliers implies $H_0:\theta = 0$ under this mixture model.

As a starting point, let us assume $g$ is known.
In practice, however, the test discussed in this section can be implemented using a data-driven estimate $\hat{g}$ of $g$, as detailed in Appendix~\ref*{app:outlier-distr}.
Importantly, we will show in Proposition~\ref{prop:adaptive_exchangeability} that this test remains valid when using $\hat{g}$ in place of $g$, provided that $\hat{g}$ is estimated in a manner that preserves the exchangeability between the calibration and test scores.

Let $R_1,\ldots,R_N$ denote the ranks of the pooled sample $(X_1,\ldots,X_m,Y_1,\ldots,Y_n)$, where $N=m+n$.
For any distribution $G$ on $[0,1]$ with density function $g$, define the statistic
\begin{align} \label{eq:def-TG}
  T^{\text{g}} & = \sum_{j = m+1}^{N} \E{g(U_{N}^{(R_{j})})},
\end{align}
where $U_{N}^{(R_j)}$ is the $R_j$-th order statistic in a sample of uniform random variables of size $N$.  The test rejects $H_{0}:\theta = 0$ at level $\alpha \in (0,1)$  if $T^{\text{g}}$ is larger than a  suitable critical value $c^{\text{g}}_{\alpha}(m,n)$, discussed later.
The indicator of this rejection event is:
\begin{align}\label{eq:generalized-WMW-test}
    \phi^{\mathrm{g}} &= \mathds{1}\left\{ T^{\text{g}}  > c^{\text{g}}_{\alpha}(m,n) \right\}.
\end{align}

In the special case of $G(F)=F^2$, the statistic $T^{\text{g}}$ in (\ref{eq:def-TG}) reduces to $\frac{2}{N+1} \cdot T^{\text{WMW}}$, where
\begin{equation*}\label{eq:classic_WMW}
    T^{\text{WMW}} = \sum_{j=m+1}^{N} R_j
\end{equation*}
is the classical WMW statistic. The corresponding test is therefore equivalent to the WMW test, since the multiplicative constant can be ignored.
\cite{lehmann53power} proved that in the mixture model $Y_i \sim (1-\theta)F + \theta F^{2}$, the WMW test
is the LMPR test of $H_0:\theta = 0$, i.e.~it is most powerful within the class of all rank tests at level $\alpha$ uniformly in a small neighbourhood of 0.
\cite{shiraishi1985local} extended the result by \cite{lehmann53power} and proved that the 
LMPR test for $H_0$
given any continuous density function $g$ is the one in Equation~\eqref{eq:generalized-WMW-test}; see Theorem~\ref{thm:LMPI} in Appendix~\ref*{app:local-shirashi}.

In Appendix~\ref*{app:local-shirashi}, we study the power of Shiraishi's LMPR test and compare its relative efficiency with that of the Neyman-Pearson optimal test which assumes oracle knowledge of $F$, $G$ and $\theta$. 
This analysis reveals that Shiraishi's test is most powerful for balanced calibration and test sample sizes, which provides a potentially useful practical guideline. Moreover, it shows that the asymptotic relative efficiency (ARE), or Pitman efficiency, of Shiraishi's LMPR test relative to the LRT optimal test equals $1-\lambda$ if $n/N \rightarrow \lambda \in (0,1)$.

\subsubsection{Asymptotic null distribution of Shiraishi’s test without independence} \label{sec:asymptotic_rank}

The critical value $c^{\text{g}}_{\alpha}(m,n)$ in~\eqref{eq:generalized-WMW-test} may be obtained from the permutation distribution or from a large-sample approximation. The following theorem extends the asymptotic approximation of the null distribution of $T^{\text{g}}$ given by \citet{shiraishi1985local}, which assumed i.i.d.\ scores, to the weaker assumption of exchangeable scores. 
This extension is useful because it allows Shiraishi’s test to be applied in combination with the ACODE methodology presented in this paper, which involves a key step that breaks independence while maintaining exchangeability.
The proof, in Appendix~\ref*{app:proofs}, follows from a general formulation of the finite-population central limit theorem \citep{hajek1960limiting, li2017general}. 

\begin{theorem}
\label{thm:asymp_distr_LMPI-exch}
    As $N\rightarrow \infty$, suppose that
    \begin{eqnarray}\label{eq:regulation_condition}
        \frac{\max_{1\leq r \leq N}(\E{g(U_N^{(r)})} -\mu_N)^2}{\min(n,m) \sigma^2_N} \rightarrow 0,
    \end{eqnarray}
    where $$\mu_N = \frac1{N}\sum_{r\in[N]}\E{g(U_N^{(r)})},\qquad  
    \sigma^2_N = \frac1{N-1} \sum_{r \in [N]}( \E{g(U_N^{(r)})} -\mu_N )^2.$$ 
    Then, under the null hypothesis $H^*_0: \mathrm{the\,\,vector\,\,}(X_1,\ldots,X_m,Y_1,\ldots,Y_n)\,\,\mathrm{is\,\,exchangeable},$ the standardized Shiraishi statistic converges in distribution to a standard normal,
    \begin{eqnarray}\label{eq:rescaled-LMPI}
        \frac{T^{\mathrm{g}} - n\mu_N}{\sqrt{\frac{m n}{N} \sigma^2_N}} \xrightarrow{d} N(0,1).
    \end{eqnarray}
\end{theorem}

Note that as $N\rightarrow \infty$, condition \eqref{eq:regulation_condition} implies  $m\rightarrow \infty$ and $n\rightarrow \infty$ \citep{li2017general}. 
Moreover, the asymptotic approximation of the null distribution of $T^{\text{g}}$ given in \citet[Theorem 1]{shiraishi1985local} requires that $g$ be bounded, which in turn implies~\eqref{eq:regulation_condition}.

As a consequence of Theorem~\ref{thm:asymp_distr_LMPI-exch}, when considering any $S\subseteq [n]$ and denoting by $s$ its cardinality, under the null hypothesis $H_{S}: \cap_{j \in S} H_j$, the test statistic
\begin{align} \label{eq:def-TG-S}
  T^{\text{g}}_S & = \sum_{j = m+1}^{m+s} \E{g(U_{m+s}^{(R_j^S)})},
\end{align}
where $R^S_1,\ldots,R^{S}_{m+s}$ denote the ranks of $(X_1,\ldots,X_m,(Y_j)_{j\in S})$ and $U_{m+s}^{(R_j^S)}$ is the $R^S_j$-th order statistic in a sample of uniform random variables of size $m+s$, is asymptotically normal with mean 
\begin{equation*}
    s\mu_{m+s} = \frac{s}{m+s}\sum_{r\in[m+s]} \E{g(U^{(r)}_{m+s})}
\end{equation*}
and variance
\begin{equation*}
    \frac{ms}{m+s}\sigma^2_{m+s} = \frac{m s}{(m+s)(m+s-1)}\sum_{r\in[m+s]} \left(\E{g(U^{(r)}_{m+s})} - \mu_{m+s}\right)^2.
\end{equation*}

\subsubsection{Shirashi's Adaptive Rank Test} \label{sec:adaptive_rank}

The oracle procedure studied in Theorem~\ref{thm:asymp_distr_LMPI-exch} can be translated into a practical test by replacing the unknown density $g$ with a suitable empirical estimate $\hat g$. In Appendix~\ref*{app:outlier-distr}, we establish finite-sample validity of the resulting adaptive rank test (Proposition~\ref*{prop:adaptive_exchangeability}) and derive an asymptotic approximation to its null distribution (Theorem~\ref*{thm:asymp_distr_adaptive-exch}).

The following theorem, proved in Appendix~\ref*{app:proofs}, implies that Shiraishi's adaptive rank test is asymptotically equivalent to LMPR test in terms of local asymptotic power. Specifically, we first show that the difference between the oracle and adaptive rank statistics is negligible under the null hypothesis.
By contiguity, this negligibility extends to local alternatives, and it follows that the adaptive rank test attains the same local asymptotic power function as the LMPR test with oracle knowledge of $g$.
To simplify notation, we replace  $\mathbb{E}[ g( U^{(r)}_N ) ]$ by the approximation $g\!\left(\frac{r}{N+1}\right)$. The difference  $\mathbb{E}[ g( U^{(r)}_N ) ]-g\!\left(\frac{r}{N+1}\right)$  is asymptotically negligible and does not affect the asymptotic results. 

\begin{theorem}
\label{thm:asymp_equivalence}
Let \(X_{1}, \ldots, X_{m} \stackrel{\mathrm{i.i.d.}}{\sim} F\) and 
\(Y_{1}, \ldots, Y_{n} \stackrel{\mathrm{i.i.d.}}{\sim} (1-\theta)F + \theta G(F)\),
where \(F\) is continuous and \(G\) admits a continuous and bounded density \(g\) on \([0,1]\).
Let \(N = m + n\) and assume that \(n/N \to \lambda \in (0,1)\) as \(N \to \infty\).

Let $\hat g_N$ be an estimate of \(g\) based on the unordered pooled sample
\(\{W_1,\ldots,W_N\}\). If
\begin{equation}\label{eq-L2consistency}
\frac{1}{N}\sum_{r=1}^N
\Bigl[
\hat g_N\!\left(\frac{r}{N+1}\right)
-
g\!\left(\frac{r}{N+1}\right)
\Bigr]^2
\;\xrightarrow{\mathbb{P}}\; 0 ,
\end{equation}
then, under the null hypothesis \(H_0:\theta = 0\),
\[
\frac{T^{\hat{\text{g}}} - n\hat{\mu}_N}{\sqrt{N}} -  \frac{T^{\text{g}} - n\mu_N}{\sqrt{N}}\;\xrightarrow{ \mathbb{P}_0 } 0,
\]
where $T^{\hat{\text{g}}} = \sum_{j = m+1}^{N} \hat g_N\!\left(\frac{R_j}{N+1}\right)$ and $\hat{\mu}_N = N^{-1} \sum_{r \in [N]} \hat g_N\!\left(\frac{r}{N+1}\right)$. 
\end{theorem}

\begin{corollary} \label{cor:asymp_equivalence-power}
As a consequence of Theorem~\ref{thm:asymp_equivalence}, the Shiraishi's adaptive rank statistic attains the same local asymptotic power function as the oracle LMPR statistic.
\end{corollary}

In Appendix~\ref*{app:adaptive_localpower}, we show that our data-driven estimate of $\hat{g}$ can be implemented using the approach of \citet{patra2016estimation}, in which case condition \eqref{eq-L2consistency} is satisfied. Thus, we can leverage an adaptive rank test with the same local asymptotic power as the oracle LMPR test. Moreover, this test can be practically applied within our closed testing methodology described below thanks to the closed-testing shortcut described in Appendix~\ref*{app:shortcuts-G}.

The desirable theoretical properties of this adaptive test are corroborated by the empirical results presented in Section~\ref{sec:numerical}, where we show that in some scenarios of interest this adaptive test can achieve much higher power compared to other local testing methods, such as Simes', Fisher's, and the WMW test.
In any case, the true strength of our ACODE method lies in its flexibility, as it does not depend on any single testing procedure but can instead dynamically select the most effective approach based on the data at hand.

\subsection{Estimating the Number of Outliers via Closed Testing} \label{sec:closed-testing}

Closed testing \citep{marcus1976closedtesting} was initially proposed within the scope of family-wise error rate control, but it is also useful to obtain simultaneous bounds for the false discovery proportion \citep{goemansolari2011}. In this paper, we apply closed testing to construct a $(1-\alpha)$ {\em simultaneous} lower confidence bound $d(S)$ for the number of outliers in any subset $S \subseteq [n]$ of the test set.
This approach is outlined by Algorithm~\ref{alg:closed-testing}.
Although this procedure may seem computationally unfeasible, since it generally requires evaluating $\mathcal{O}(\exp(n))$ tests, it can be implemented efficiently in many useful cases, as discussed later.

\begin{algorithm}[!htb]
\SetKwInOut{Input}{Input}
\Input{Individual hypotheses $H_j$, for all $j \in [n]$; a method $\phi$ for carrying out {\em local} tests $\phi_J$ of $H_J$, for any $J \subseteq [n]$; significance level $\alpha \in (0,1)$.}

For each $J \subseteq [n]$, test $H_J$ at level $\alpha$; let $\phi_J \in \{0,1\}$ denote the rejection indicator.

For each $J \subseteq [n]$, adjust the local test by setting $\bar{\phi}_J := \min\{\phi_K: K\supseteq J\}$.

\textbf{Query \(S\subseteq[n]\):} compute the lower bound for the number of true discoveries in $S$,
\begin{align}\label{eq-dS}
d(S) &:= \min_{K \subseteq S} \{|S \setminus K|:\bar{\phi}_K = 0\}.
\end{align}

\SetKwInOut{Output}{Output}
\Output{A $(1-\alpha)$ lower bound $d(S)$ for any queried subset $S \subseteq [n]$.}
\caption{Closed testing blueprint for simultaneous outlier enumeration} \label{alg:closed-testing}
\end{algorithm}

The first step of Algorithm~\ref{alg:closed-testing} tests all possible intersection hypotheses $H_J$. We refer to these tests as the \emph{local tests}.
In the second step, a {\em multiplicity adjustment} is applied to all local tests. The latter results in $H_J$
being rejected, setting
$\bar{\phi}_J=1$, if and only if $H_{K}$ was locally rejected for all supersets $K$ of $J$.
Finally, $d(S)$ is given by~\eqref{eq-dS}.
This is easiest to explain in the special case of $S=[n]$, when the lower confidence bound for the total number of outliers is given by the difference between the total number of test points, $n$, and the cardinality of the largest subset for which the multiplicity-adjusted local test failed to reject the null hypothesis. For example, if the global null $H_{[n]}$ is rejected but there exists a subset $K \subseteq [n]$ with cardinality $|K|=n-1$ for which $\phi_K=0$, then $d([n]) = 1$.

The following result, due to \citet{goemansolari2011}, guarantees that closed testing leads to simultaneously valid confidence bounds as long as the local tests are valid. 
\begin{proposition}[\citet{goemansolari2011}] \label{prop:closed-testing}
If Algorithm~\ref{alg:closed-testing} is applied using a valid local test $\phi$ satisfying $\mathbb{P}[\phi_{I_0} = 1 \; ;\; H_{I_0}] \leq \alpha$, then the output lower bounds $d(S)$ satisfy:
\begin{align}\label{eq-TDguarantee}
  \mathbb{P}\left[ d(S) \leq |I_1 \cap S| \text{ for all } S \subseteq [n] \right] \geq 1-\alpha.
\end{align}
\end{proposition}

The availability of computational shortcuts for Algorithm~\ref{alg:closed-testing} hinges on the choice of local tests \citep{goeman2019simultaneous, goeman2021onlyclosed, tian2023Large-scale}.
See Table~\ref*{tab:comp-costs} in Appendix~\ref*{app:shortcuts} for a summary of relevant shortcuts and their computational complexity.
Appendix~\ref*{app:shortcuts-simes} reviews the shortcut of \citet{goeman2019simultaneous}, which reduces the cost of Algorithm~\ref{alg:closed-testing} to $\mathcal{O}(n \log n)$ in the case of  Simes local tests. Appendix~\ref*{app:shortcuts-adaptive-simes} presents a shortcut for an adaptive version of the Simes local test. Appendix~\ref*{app:shortcuts-G} introduces a shortcut for the Shirashi test described in Section~\ref{sec:new_local-g-wmw}.
Appendix~\ref*{app:shortcuts-tian} reviews the strategy of \citet{tian2023Large-scale}, which provides similar shortcuts for a family of local tests including Fisher's method and the WMW test.

\subsection{Conformalized Closed Testing} \label{sec:conformal-closed-testing}

We now present Algorithm~\ref{alg:conformalized-closed-testing}, which summarizes the analysis pipeline described above, starting from the training of a classifier.
In short, this applies the closed testing method defined in Algorithm~\ref{alg:closed-testing} based on local tests that take as input the calibration and test scores, assuming a practical computational shortcut is available. Therefore, Algorithm~\ref{alg:conformalized-closed-testing} outputs a simultaneous lower bound $d(S)$ guaranteed to satisfy~\eqref{eq-TDguarantee}.

Valid conformity scores can be obtained through various methods. A common approach involves using a one-class classifier trained on an independent set of inliers \citep{bates2021testing}. Alternatively, one can employ a binary classifier via positive-unlabeled learning \citep{marandon2022machine}.
Although the optimal scoring function can, in principle, be approximated through binary classification \citep[Lemma 4.3; see also see Appendix~A4]{marandon2022machine},  in practice there are settings in which one-class classification approaches may perform better \citep{liang2022integrative}. 
Another setting in which binary classifiers can be used is when labeled outliers are available \citep{liang2022integrative}; however, we focus for simplicity on scenarios where all labeled data consist of inliers.

 \begin{algorithm}[!htb]
          \SetKwInOut{Input}{Input}
          \Input{Inlier data $\mathcal{D}^{\text{train}} = \{Z^{\text{train}}_1, \ldots, Z^{\text{train}}_{m_{\text{train}}}\}$ and $D^{\text{cal}} = \{Z_1, \ldots, Z_{m}\}$.\\
            Test data $D^{\text{test}} = \{Z_{m+1}, \ldots, Z_{m+n}\}$. Significance level $\alpha \in (0,1)$. \\
            Machine learning algorithm $\mathcal{A}$ for one-class or binary classification.\\
            Chosen local testing method $\phi$; e.g., Simes, WMW, etc.
          }
        
          \uIf{$\mathcal{A}$ is a one-class classification algorithm}{
            Train $\mathcal{A}$ using the data in $(Z^{\text{train}}_1, \ldots, Z^{\text{train}}_{m_{\text{train}}})$.\\
          }
          \uElseIf{$\mathcal{A}$ is a binary classification algorithm}{
            Train $\mathcal{A}$ using $(Z^{\text{train}}_1, \ldots, Z^{\text{train}}_{m_{\text{train}}})$ and $\{Z_1, \ldots, Z_{m}, Z_{m+1}, \ldots, Z_{m+n}\}$.\\
          }
          Apply $\mathcal{A}$ to evaluate the calibration and test scores $(X_1,\ldots,X_m)$ and $(Y_1,\ldots,Y_n)$.
          
          Compute the ranks $R_1,\ldots,R_{n+m}$ of the pooled sample and the conformal $p$-values.

          \textbf{Query  \(S\subseteq[n]\):} compute $d(S)$ using Algorithm~\ref{alg:closed-testing}
          (ref.~Appendix~\ref*{app:shortcuts}) and the closed-testing adjusted $p$-value $\hat{p}(S)$ for the local test $\bar{\phi}_S$ of $H_S$.\\

        \SetKwInOut{Output}{Output}
        \Output{A $(1-\alpha)$ lower bound $d(S)$ for the number of outliers in $S$, and a $p$-value $\hat p(S)$ for testing $H_S$, for any queried subset $S \subseteq [n]$.}          
        \caption{Conformal Outlier Detection and Enumeration} \label{alg:conformalized-closed-testing}
    \end{algorithm}

Proposition~\ref{prop:closed-testing} implies Algorithm~\ref{alg:conformalized-closed-testing} produces valid inferences as long as the local test $\phi_{I_0}$ is valid, and our method is precisely designed to achieve this (ref.~Section~\ref{sec:local-tests}).
While the lower bound $d(S)$ output by Algorithm~\ref{alg:conformalized-closed-testing} is {\em simultaneously} valid for all possible subsets $S \subseteq [n]$, in the sense of~\eqref{eq-TDguarantee}, this does not mean Algorithm~\ref{alg:conformalized-closed-testing} needs to explicitly output $d(S)$ for all $S \subseteq [n]$.
On the contrary, one would typically apply Algorithm~\ref{alg:conformalized-closed-testing} focusing on a particular (but possibly data-driven) choice of $S$, as demonstrated in Section~\ref{sec:numerical}.
In any case, the components of Algorithm~\ref{alg:conformalized-closed-testing} involving the model training and the computation of the conformity scores only need to be applied once, irrespective of $S$.

We conclude this section by emphasizing that the flexibility of Algorithm~\ref{alg:conformalized-closed-testing}, which can accommodate a variety of classifiers and testing procedures, introduces some challenges.
In particular, different classification algorithms and testing procedure may result in significant performance variation across different data sets, and it is unclear a priori how to maximize power.
Unfortunately, leaving too much latitude to practitioners is not always desirable, as it may inadvertently encourage ``cherry picking" behaviors that, as previously illustrated in Figure~\ref{fig:exp-lhco}, can result in invalid inferences.
This issue motivates the extension introduced in the next section. This extension will channel the flexibility of Algorithm~\ref{alg:conformalized-closed-testing} into a principled, automatic method that often achieves competitively high power in practice while adding a layer of protection against the risks of human-driven selection bias.

\subsection{Data-Driven Tuning} \label{sec:tuning_new}

We extend Algorithm~\ref{alg:conformalized-closed-testing} to leverage a potentially diverse suite of classifiers and local testing procedures, with the goal of approximately maximizing power while controlling type-I error. Our solution, outlined in Algorithm~\ref{alg:acode_oracle}, is inspired by approaches proposed by \citet{liang2022integrative} and \citet{marandon2022machine} in the context of individual outlier identification under FDR control. The key idea is to embed a data-driven tuning step within an additional layer of sample splitting. As detailed below, this extra sample splitting is essential to preserve the validity of local tests after selecting the most effective classifier and local testing procedure; at the same time, provided the available sample size is sufficiently large, it typically does not lead to a substantial loss of power relative to data-driven selection without splitting (an intuitive but theoretically invalid benchmark).

\begin{algorithm}[!htb]
  \SetKwInOut{Input}{Input}
 \Input{Inlier data $\mathcal{D}^{\text{train}} = \{Z^{\text{train}}_1, \ldots, Z^{\text{train}}_{m_{\text{train}}}\}$, $D^{\text{cal}} = \{Z_1, \ldots, Z_{m}\}$, $\mathcal{D}^{\text{tune}} = \{Z^{\text{tune}}_1, \ldots, Z^{\text{tune}}_{m_{\text{tune}}}\}$. \\
   Test data $D^{\text{test}} = \{Z_{m+1}, \ldots, Z_{m+n}\}$,
   Significance level $\alpha \in (0,1)$. \\
   A list of algorithms $\mathcal{A}^{(1)}, \ldots, \mathcal{A}^{(K)}$ for one-class or binary classification.\\
   A list of local testing methods $\phi^{(1)}, \ldots, \phi^{(L)}$; e.g., Simes, WMW, etc.\\
   A subset of interest $R$, with default choice $R=[n]$.
 }
  
  Define $\tilde{\mathcal{D}}^{\text{test}} = \{Z_1, \ldots, Z_{m}\} \cup \{ Z_{m+1}, \ldots, Z_{m+n} \}$.\\
    \For{each $k \in [K]$}{

        \For{each $l \in [L]$}{
        Compute $d_{k,l}(R)$ and $\hat{p}_{k,l}(R)$ by applying Algorithm~\ref{alg:conformalized-closed-testing} based on $\mathcal{A}^{(k)}$ and $\phi^{(l)}$, using $\mathcal{D}^{\text{train}}$ for training, $\mathcal{D}^{\text{tune}}$ for calibration, and $\tilde{\mathcal{D}}^{\text{test}}$ as the test set.\\
        }
    }
    Find $(\hat{k}(R), \hat{l}(R)) = \arg\max_{k \in [K], l \in [L]} d_{k,l}(R)$; if there are multiple maximizers, select the one with the smallest $p$-value $\hat{p}_{k,l}(R)$. \\

    \textbf{Query  \(S\subseteq[n]\):} compute $d(S)$ and $\hat{p}(S)$ with Algorithm~\ref{alg:conformalized-closed-testing} based on $\mathcal{A}^{(\hat{k}(R))}$ and $\phi^{(\hat{l}(R))}$, using $\mathcal{D}^{\text{train}} \cup \mathcal{D}^{\text{tune}}$ for training, $\mathcal{D}^{\text{cal}}$ for calibration, and $\mathcal{D}^{\text{test}}$ as test set.

\SetKwInOut{Output}{Output}
\Output{A $(1-\alpha)$ lower bound $d(S)$ for the number of outliers in  $S$, and a $p$-value $\hat p(S)$ for testing $H_S$, for any queried subset $S \subseteq [n]$.}

\caption{Automatic Conformal Outlier Detection and Enumeration} \label{alg:acode_oracle}
\end{algorithm}

The tuning module of Algorithm~\ref{alg:acode_oracle} sees the calibration and test data only through the lenses of the unordered collection $\{Z_1, \ldots, Z_{m}\} \cup \{ Z_{m+1}, \ldots, Z_{m+n} \}$.
Consequently, when Algorithm~\ref{alg:conformalized-closed-testing} is applied again in the inference step of Algorithm~\ref{alg:acode_oracle}, using $(Z_1, \ldots, Z_{m})$ as a calibration set and $(Z_{m+1}, \ldots, Z_{m+n})$ as a test set, the inlier scores within $(Z_1, \ldots, Z_{m+n})$ are still exchangeable conditional on the selected classifier and testing procedure.

There is an interesting distinction between exchangeability and independence in the context of Algorithm~\ref{alg:acode_oracle}.
Recall from Section~\ref{sec:conf-scores} that the model in~\eqref{eq:model} and the approach described in~\eqref{eq:scores} lead to i.i.d.~scores for the inliers in the case of one-class classification, and only exchangeable scores in the case of binary classification \citep{marandon2022machine}.

Algorithm~\ref{alg:acode_oracle} changes this picture a little.
As we condition on the selected classifier and closed testing procedure, the scores of the inliers among $(Z_1, \ldots, Z_{m+n})$, although still exchangeable, can no longer be independent, irrespective of whether they were initially obtained through one-class or binary classification.

Relaxing the i.i.d.\ assumption on the inliers to exchangeability does not affect inference.
Theorems~\ref{thm:asymp_distr_LMPI-exch} and~\ref{thm:asymp_distr_adaptive-exch} show that the large-sample approximations of the null distributions of Shiraishi’s statistic and its adaptive version remain valid under exchangeability. Moreover, for the other local testing procedures considered in this paper, namely Fisher’s and Simes’ combinations, the joint distribution of null conformal $p$-values is unchanged whether the inliers are i.i.d.\ or exchangeable, as shown in \citet[Proposition~2.2]{gazin2024transductive}.
Formally, Algorithm~\ref{alg:acode_oracle} outputs a simultaneous lower bound 
$d(S)$ that satisfies \eqref{eq-TDguarantee}, provided that exchangeability among the inliers holds.
Thus, Algorithm~\ref{alg:acode_oracle} can often achieve high power without selection bias, as previewed in Figure~\ref{fig:exp-lhco} and confirmed in Section~\ref{sec:numerical}.

Finally, Algorithm~\ref{alg:acode_oracle} requires specifying a subset of interest $R$, with default $R=[n]$. ACODE then selects the pair of classifier and local test yielding the largest lower bound $d(R)$ on the number of outliers in 
$R$. Multiple choices of $R$ may be considered, and the final selection of the (classifier, local test) pair may be made post hoc by comparing their results, as long as this choice is made before entering line~7 of the Algorithm. 
If the goal is to optimize rejection of the global null hypothesis (detection), rather than to maximize the lower bound (enumeration), the algorithm should select the (classifier, local test) pair that yields the smallest $p$-value; i.e., 
$(\hat{k}(R), \hat{l}(R)) = \arg\min_{k \in [K], l \in [L]} \hat p_{k,l}(R)$ with $R=[n]$. 

\section{Empirical Demonstrations}\label{sec:numerical}

Section~\ref{sec:numerical-synthetic} applies ACODE to synthetic data, while Section~\ref{sec:numerical-lhco} considers the 2020 LHCO dataset \citep{kasieczka2021lhc}. Additional experiments are reported in Appendix~\ref*{app:numerical}. In particular, Appendix~\ref*{app:numerical-real} presents results on datasets previously studied by \citet{bates2021testing} and \citet{marandon2022machine}. Consistent with \citet{marandon2022machine}, binary classifiers achieve higher power on those relatively low-dimensional data. Yet, as known from \citet{liang2022integrative} and confirmed in this section, one-class classifiers can yield higher power on different data sets, highlighting the practical benefit of ACODE's adaptability.

\subsection{Experiments with Synthetic Data} \label{sec:numerical-synthetic}

\subsubsection{Setup}

We apply ACODE on synthetic data from a distribution inspired by~\citet{liang2022integrative} and \citet{bates2021testing}.
Each observation $Z_i \in \mathbb{R}^{1000}$ is sampled from a multivariate Gaussian mixture $P_Z^{a}$, such that $Z_i = \sqrt{a} \, V_i + W_i$, for some constant $a\geq 0$ and vectors $V_i,W_i \in \mathbb{R}^{1000}$.
The inliers correspond to $a=1$ and the outliers to $a = 0.7$.
The elements of $V_i$ are standard Gaussian, while each element of $W_i$ is independent and uniformly distributed on a discrete set $\mathcal{W} \subseteq \mathbb{R}^{1000}$ with $|\mathcal{W}|=1000$.
The vectors in $\mathcal{W}$ are independently sampled from uniform($[-3,3]^{1000}$) prior to the first experiment.
The numbers of inliers in the training, calibration, and tuning sets are by default 1000, 750, and 250, respectively.
As shown in Figure~\ref{fig:exp-synthetic-1-tuning} (Appendix~\ref{app:numerical-synthetic}), ACODE’s performance is quite stable across a wide range of tuning splits, degrading only when tuning data are extremely scarce.

ACODE is applied using 6 classification algorithms and 5 local testing procedures.
The algorithms include 3 one-class classifiers (isolation forest, support vector machine, and ``local outlier factor'' nearest neighbors) and 3 binary classifiers (deep neural network, random forest, and AdaBoost), all implemented in the Python package {\em scikit-learn}.
The one-class classifiers compute out-of-sample scores, e.g., as in \citet{bates2021testing}, while the binary classifiers compute in-sample scores with the approach of \citet{marandon2022machine}.

The testing procedures leveraged by ACODE are: Simes test with and without Storey's correction, Fisher's combination method, the WMW test, and Shiraishi's LMPR test from Section~\ref{sec:new_local-g-wmw}, applied using $G(F) = F^3$. 
This choice of $G$ is optimal against Lehmann's alternative with $k=3$, which may or may not fit our data well.
An implementation of the Shiraishi test based on a data-driven estimate of $G$ will be considered later.

\subsubsection{Global Outlier Enumeration}

We begin by constructing 90\% lower confidence bounds for the total number of outliers in a test set of size 1000, varying the proportion of outliers as a control parameter.
Figure~\ref{fig:exp-synthetic-1} summarizes median lower bounds produced by ACODE over 100 independent experiments, separately for different classifiers and local testing procedures.
The left and center panels compare the performance of ACODE applied using only the 3 one-class or binary classifiers, respectively, while the right panel corresponds to ACODE automatically selecting a classifier from the full suite of all 6 options.

\begin{figure}[!htb]
\centering
\includegraphics[width=\linewidth]{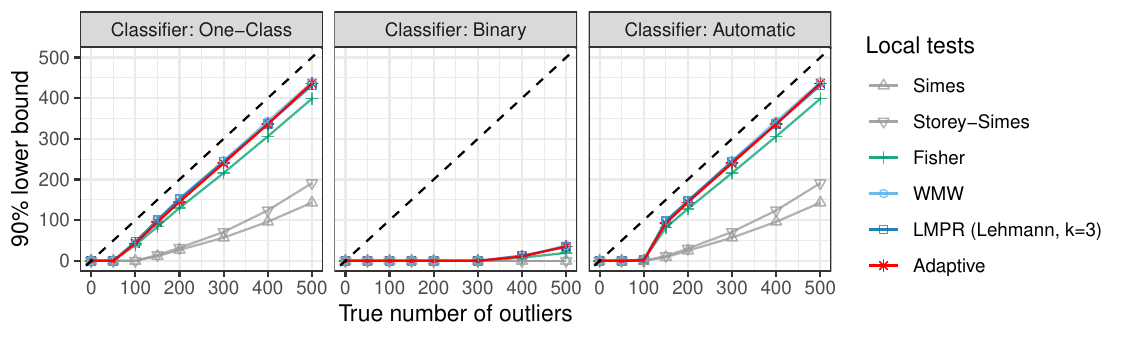}
\caption{Median values for a 90\% lower confidence bound on the number of outliers in a test set, computed by ACODE on synthetic data based on different classifiers and local testing procedures. The results are shown as a function of the true number of outliers within a test set of size 1000. The most adaptive version of ACODE can automatically select an effective classifier and local testing procedure in a data-driven way.
}
\label{fig:exp-synthetic-1}
\end{figure}

The findings highlight that one-class classifiers generally outperform binary classifiers on these data \citep{liang2022integrative}, although binary classifiers can be advantageous in other contexts, as shown in Appendix~\ref*{app:numerical-synthetic}.
Figure~\ref{fig:exp-synthetic-1} distinguishes between the performances of the 5 local testing procedures (shown in different colors), and the scenario in which ACODE selects one procedure adaptively.
For these data, the WMW test yields the most informative (highest) lower confidence bounds, and ACODE's performance closely approximates that of an {\em ideal oracle} that knows the optimal combination of classifier and testing procedure.

\subsubsection{Selective Outlier Enumeration} \label{sec:exp-selective}

Figure~\ref{fig:exp-synthetic-1-sel-q0.5} reports on related experiments in which the goal is to construct a 90\% lower confidence bounds for the number of outliers within a data-driven subset of test points, selected as those with the largest conformity scores.
To facilitate the interpretation of these experiments, ACODE is applied using a fixed classification algorithm, a one-class support vector machine.
This ensures the test subsets are always selected based on comparable conformity scores in all repetitions of the experiments.
We do not apply ACODE with the Shiraishi local test here due to its higher computational cost when the selected subset differs from the full test set; see Table~\ref*{tab:comp-costs}.
These results further demonstrate how the efficacy of different local testing procedures can vary in different situations.
Simes' method performs better with very small selected sets, while the WMW and Fisher's test excel in cases with moderately large selected sets.
Once more, ACODE can approximately maximize power by autonomously identifying the most effective local testing procedure for each scenario.

\begin{figure}[!htb]
\centering
\includegraphics[width=0.9\linewidth]{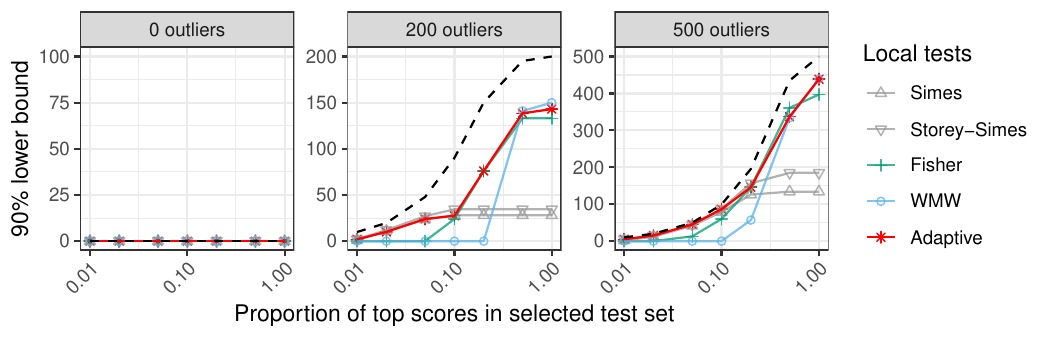}
\caption{Median values for a 90\% lower confidence bound on the number of outliers within an adaptively selected subset of 1000 test points, in experiments similar to those of Figure~\ref{fig:exp-synthetic-1}.
The results are shown as a function of the proportion of selected test points and of the total number of outliers in the test set.
The dashed curve corresponds to the true number of outliers in this selected set.
In these experiments, ACODE is applied using a one-class support vector classifier to compute the conformity scores.
}
\label{fig:exp-synthetic-1-sel-q0.5}
\end{figure}

\subsubsection{Hunting for Adversarial Anomalies} \label{sec:exp-adversarial}

We now apply ACODE leveraging an additional local testing procedure, utilizing a more flexible implementation of the LMPR approach outlined in Section~\ref{sec:new_local-g-wmw}. This approach enables our method to even detect elusive anomalies hidden by an adversary.

We consider inlier data sampled from a standard multivariate normal distribution with 100 independent components. The outliers are generated by an ``adversary'' agent through the following process. The adversary first trains a one-class support vector machine on a separate dataset consisting of 1,000 inlier points. 
Then, to generate each outlier, it randomly produces three independent inliers and selects the one closest to the decision boundary of the support vector machine. This approach results in outliers that are difficult to detect on an individual basis but can be identified collectively, as their conformity scores tend to be under-dispersed compared to those of the inliers.

Without prior knowledge of the distribution of outlier scores, ACODE is applied as in previous experiments, using the same toolbox of six machine learning classifiers, but including a sixth local testing procedure.
This procedure is a practical approximation of the Shiraishi test described in Section~\ref{sec:new_local-g-wmw}, using an empirical estimate $\hat{g}$ of $g$. This estimate is obtained by fitting a mixture of Beta distributions, while preserving the exchangeability between calibration and test scores, as detailed in Appendix~\ref*{app:outlier-distr}. 

The results reported in Figure~\ref{fig:exp-synthetic-adversary} demonstrate that our method effectively detects these elusive adversarial outliers, even though the first five local testing procedures considered become ineffective in this scenario.
For additional experiments that offer a more detailed view of the performance of our empirical approximation of the Shiraishi local testing procedure in different settings, please refer to Figures~\ref{fig:exp-G-hat-power}--\ref{fig:exp-G-hat-lb} in Appendix~\ref*{app:numerical-synthetic}.

\begin{figure}[!htb]
\centering
\includegraphics[width=0.85\linewidth]{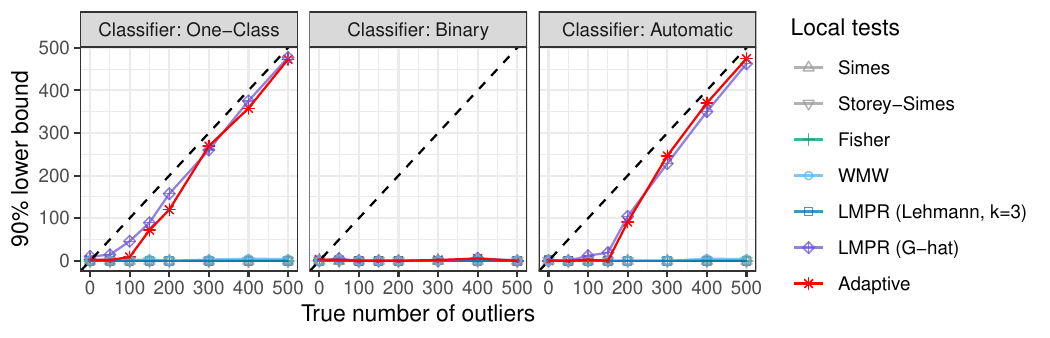}
\caption{
Empirical 90th percentile of the 90\% lower confidence bound on the number of outliers in the test set, computed by ACODE on synthetic data with adversarially hidden outliers exhibiting underdispersed conformity scores.
Most local testing procedures cannot detect these outliers, but a data-driven approximation of the Shiraishi test enables our method to achieve high power.
Other details are as in Figure~\ref{fig:exp-synthetic-1}. }
\label{fig:exp-synthetic-adversary}
\end{figure}

\subsection{Experiments with High-Energy Particle Collision Data} \label{sec:numerical-lhco}

We now apply our method on data from the 2020 LHCO contest \citep{kasieczka2021lhc}, revisiting the example previewed in Figure~\ref{fig:exp-lhco}.
Participants in this contest were asked to estimate the number of interesting ``signal'' particle collision events within a given data set by contrasting it with a ``reference''  data exclusively comprised of ``background'' events  generated by a physics simulation software.
While the LHCO challenge is closely related to the problem considered in this paper, the analysis in this section overlooks some of the physics intricacies and emphasizes instead the statistical aspects of our method.

We focus on a subset of the LHCO data containing 1,100,000 observations of 14 high-level features constructed based on subject domain knowledge.
Each observation was labeled as either a background (inlier) or a {\em 2-prong signal} (outlier) event, where the latter refers to hadron collision resulting in the emission of two charged particles.
The goals are to detect and enumerate the outliers within a randomly chosen unlabeled test set, having access to a reference set of independent and identically distributed inliers.

In each experiment, we randomly select a test set of 10,000 events, among which the proportion of outliers is a control parameter varied between 0 and 0.15.
We also randomly sample three disjoint subsets of inliers, to be used for training (cardinality 10,000), calibration (cardinality 1000), and tuning (cardinality 1000), respectively.
ACODE is then applied as described in Section~\ref{sec:numerical-synthetic}, leveraging the same suites of 6 classification algorithms and 5 local testing procedures.
The additional local testing procedure described in Section~\ref{sec:exp-adversarial} (the data-driven approximation of the Shiraishi oracle) is omitted here, for simplicity. This omission is due to its somewhat higher computational cost of its closed testing shortcut and the lack of clear advantages observed with these data.

Figure~\ref{fig:exp-lhco} compares the performance of different implementations of our method, in terms of power and the total number of detected outliers, as a function of the true number of outliers.
We refer to Figure~\ref*{fig:exp-lhco-full} and Table~\ref*{tab:lhco-lb} in Appendix~\ref*{app:numerical-synthetic} for a more comprehensive view of these results. The latter detail the performance of ACODE separately applied based on each of the 5 alternative local testing procedures considered here--- to prevent overcrowding, only a subset of these local procedures were previously displayed in Figure~\ref{fig:exp-lhco}.

Note that these results also include comparisons of ACODE's performance to that of a naive ``cherry-picking'' heuristic version of our method, which does not provide valid inferences, as well as to that of the BH procedure applied to conformal $p$-values.
The latter is an effective approach for individual outlier {\em identification} under FDR control \citep{bates2021testing,marandon2022machine}, but it is not designed to {\em estimate} the number of outliers.

Overall, ACODE yields more informative inferences compared to individual-level outlier identification, especially when leveraging the WMW test or Fisher's combination method.
Further, ACODE achieves near-oracle performance with respect to the selections of classifier and local testing procedure, without incurring in selection bias.
Additional results, with qualitatively similar conclusions, are presented in Appendix~\ref*{app:numerical-lhco}.

\section{Discussion} \label{sec:discussion}

This work opens several avenues for future research. For example, while our method employs sample splitting, one could explore extensions to more sophisticated conformal inference frameworks that use cross-validation, which tends to be more computationally expensive but also more powerful for small sample sizes \citep{barber2019predictive}. Additionally, this paper assumes an ideal reference data set of inliers; future work could examine relaxations of this assumption, accounting for possible contamination of the calibration data \citep{barber2022conformal}, or other distribution shifts \citep{tibshirani2019conformal} following a weighting approach similar to that of \citet{hu2023two}.

In the future, one could investigate whether faster closed-testing shortcuts can be developed for the Shirashi local test, which tends to be powerful but is relatively computationally expensive as currently implemented, especially in the adaptive selection regime.
A further extension is to integrate alternative procedures for estimating the number of outliers, such as the approach of \citet{gazin2024transductive}, which can be interpreted as approximate shortcuts to closed testing and could be incorporated into ACODE, as discussed in Appendix~\ref*{app:bounds_DKW}.
Lastly, an intriguing question is whether the optimality of local tests extends to closed testing procedures, a concept intertwined with the notion of admissibility. Although this issue remains unexplored, a starting point may be found in \cite{goeman2021onlyclosed}, which provides a general discussion on admissibility and its implications for closed testing.


%% file: appendix.tex
\renewcommand{\thesection}{A\arabic{section}}
\renewcommand{\theequation}{A\arabic{equation}}
\renewcommand{\thealgocf}{A\arabic{algocf}}
\renewcommand{\thetheorem}{A\arabic{theorem}}
\renewcommand{\thecorollary}{A\arabic{corollary}}
\renewcommand{\theproposition}{A\arabic{proposition}}
\renewcommand{\thelemma}{A\arabic{lemma}}
\renewcommand{\thetable}{A\arabic{table}}
\renewcommand{\thefigure}{A\arabic{figure}}
\setcounter{figure}{0}
\setcounter{table}{0}
\setcounter{proposition}{0}
\setcounter{theorem}{0}
\setcounter{lemma}{0}
\setcounter{algocf}{0} 

\newtheorem*{theorem*}{Theorem}

\section{Background on Local Testing} \label{app:local-testing}

\begin{table}[!htb]
\centering
\resizebox{\textwidth}{!}{\small
\begin{tabular}{p{2cm}p{3cm}p{4cm}p{2.2cm}p{4cm}}
\toprule
\textbf{Method} & \textbf{Validity} & \textbf{Precision} & \textbf{Power} & \textbf{References} \\
\midrule
Simes$^\dagger$ & PRDS & Exact if $\alpha(m+1)/|S|$ is an integer & Rare/Strong effects & \cite{simes1986improved}, \cite{sarkar2008simes} \\
\midrule
Fisher & Asymptotic$^\mathsection$ (Permutations) & Conservative (Exact) & Dense/Weak effects & \cite{fisher1970statistical}, \cite{bates2021testing} \\
\midrule
WMW & Permutations$^\ddagger$ & Exact$^\mathparagraph$ & Moderately Sparse/Weak effects & \cite{wilcoxon1945individual}, \cite{mannwhitney1947} \\
\bottomrule
\end{tabular}
}
\flushleft{
\footnotesize{$^\dagger$ A variant of the Simes' test incorporating Storey’s estimator for the proportion of nulls has also been proven to be valid and may lead to improvement if the proportion of true nulls is not close to 1. See Appendix~\ref{app:local-simes} for more details and the proof of the exactness of the Simes' test when $\alpha(m+1)/|S|$ is an integer. \\
$^\mathsection$ The asymptotic approximation proposed by \citet{bates2021testing} works well even with moderate values of $m$ and $|S|$. Our empirical observations suggest that permutations can sometimes enhance power, as the asymptotic approximation tends to be somewhat conservative when both $m$ and $|S|$ are small. Therefore,
for very small samples, one may want to embed Fisher's combination statistic within a permutation test, which increases the computational cost but ensures exact finite-sample validity under minimal exchangeability assumptions. We refer to Section~\ref{app:local-fisher} for a more detailed review.
\\
$^\ddagger$ Permutation-based methods offer exact finite-sample validity as long as the inlier scores in $(X_1,\ldots,X_m, (Y_j)_{j \in I_0})$ are exchangeable. \\
$^\mathparagraph$ The critical value for the WMW test may be calculated either exactly, in small samples, or through accurate asymptotic approximations.
See Appendix~\ref{app:local-wmw} for further details.
}}
\caption{Standard approaches for testing $H_S$. }
\label{tab:standard_approaches}
\end{table}

\subsection{The Simes Test for Conformal $p$-values} \label{app:local-simes}

The Simes combining function is a test statistic based the ordered $p$-values:
\begin{align}\label{eq-Simesstat}
    T^{\mathrm{Simes}}_S &= \min_{k \in \{1,\ldots,|S|\}} |S| \cdot \frac{p_{(k:S)}}{k},
\end{align}
where $p_{(k:S)}$ denotes the $k$-th smallest value in the multiset $\{p_j: j \in S \}$. The Simes combining function is symmetric because it remains unchanged under any permutation of $(p_j)_{j \in S}$.

The Simes test \citep{simes1986improved} rejects $H_S$ at level $\alpha \in (0,1)$ if and only if $\phi^{\mathrm{Simes}}_S =1$, where
\begin{align} \label{eq-Simeslocaltest}
   \phi^{\mathrm{Simes}}_S &=  \mathds{1}\left\{  T^{\mathrm{Simes}}_S \leq \alpha \right\}.
\end{align}
The PRDS property of the conformal $p$-values in~\eqref{eq-permp} implies that the Simes test applied to the conformal $p$-values $(p_j)_{j \in S}$ is a valid level-$\alpha$ test for the intersection null hypothesis $H_S$ \citep{sarkar2008simes}: for any given $S \subseteq [n]$, $\mathbb{P}(\phi^{\mathrm{Simes}}_S = 1|H_S) \leq \alpha$.

We say that a level-$\alpha$ test $\phi_S$ \emph{dominates} another level-$\alpha$ test $\psi_S$ if $\phi_S \geq \psi_S$. The
domination is strict if, in addition, $\phi_S > \psi_S$ for at least one $(p_j)_{j \in S}$. We say that a test
is \emph{admissible} if it is not strictly dominated by any test. 

Theorem 3.1 in \cite{vovk2022admissible} demonstrates that the Simes test dominates all tests valid under arbitrary dependence of the $p$-values and relying on symmetric combining functions. 

\subsection{The Permutation-Based Simes Test} \label{app:local-psimes}

The Simes test for conformal $p$-values is itself inadmissible, as it is dominated by the permutation-based Simes method.

The theory of permutation tests and related references can be found in the books \cite{cox1979theoretical}, 
 \cite{lehmann05},  \cite{pesarin2010permutation}, as well as in the works of \cite{hemerik2018exact, hemerik2021another}.

Write $(W_1,\ldots,W_{m+|S|})$ for the combined sample $(X_1,\ldots,X_m,(Y_j)_{j\in S})$. Consider the permutation group $\Pi$ whose elements are permutations of the set of integers $[m+|S|]$. Choose a test statistic $T=T(W_1,\ldots,W_{m+|S|})$ for which small values are to be
regarded as evidence against $H_S$. For any $\pi \in \Pi$, 
let $T_\pi$ denote the value of $T$ applied to the permuted vector $(W_{\pi(1)},\ldots,W_{\pi(m+|S|)}))$.
For $\alpha \in (0,1)$, the permutation critical value is defined by:
\begin{align} \label{eq-permcrit}
  \alpha^{\mathrm{Perm}} =  \max\left(a \in \{0,(T_\pi)_{\pi \in \Pi}\}:  \sum_{\pi \in \Pi} \mathds{1}\{ T_{\pi} \leq a \} \leq \alpha |\Pi| \right).
\end{align}
Then, the permutation-based Simes test  with $T=T_S^{\mathrm{Simes}}$ in~\eqref{eq-Simesstat}  rejects $H_S$ at level $\alpha$ if and only if $\phi^{\mathrm{SimesPerm}}_S=1$, where
\begin{align} \label{eq-SimesPerm}
   \phi^{\mathrm{SimesPerm}}_S &=  \mathds{1}\left\{ T^{\mathrm{Simes}}_S  \leq \alpha^{\mathrm{Perm}} \right\}.
\end{align}

A permutation test is \emph{exactly valid}, i.e. $\mathbb{P}[\phi^{\mathrm{Perm}}_S = 1 \; ;\; H_S] = \alpha$, for any choice of test statistic under the sole assumption that the scores $(X_1,\ldots,X_m, (Y_i)_{i \in I_0})$ are exchangeable. However, achieving exact validity for all values of $\alpha$ is generally not feasible due to the discreteness of the permutation distribution. Nevertheless, it is possible to ensure exactness through a mathematical artifice using a randomized critical region.

The computational cost of a permutation test may be $\mathcal{O}((m+|S|)!)$, which is prohibitive if $m$ or $|S|$ are even moderately large.
The typical solutions involve restricting the set of permutations, either by utilizing a fixed subgroup or a random subset of $\Pi$ (chosen independently of the conformity scores) with an added trivial identity permutation; e.g., see Theorems 1 and 2 in \cite{hemerik2018exact}.

The precise computation of $\alpha^{\mathrm{Perm}}$ for the Simes statistic requires ${n+|S| \choose |S|}$ permutations, representing the number of ways, disregarding order, that $|S|$ test units can be chosen from among $n+|S|$ units. However, this $\alpha^{\mathrm{Perm}}$ does not depend on the scores $(W_1,\ldots,W_{m+|S|})$ but only on $m$, $|S|$ and $\alpha$, making it possible to tabulate.
This is because the permutation-based Simes test is a \emph{rank} test. 

To see this, consider the conformal $p$-value corresponding to the permuted vector $(W_{\pi(1)},\ldots,W_{\pi(m+|S|)}))$, that is:
\begin{eqnarray*}\label{eq-permpW}
    p^\pi_j &=& \frac{1}{(m+1)}\left(1+ \sum_{k = 1}^{m} \mathds{1}\{ W_{\pi(k)} \geq W_{\pi(m+j)} \}\right).
\end{eqnarray*}
This yields $p_j$ in~\eqref{eq-permp} for the identity permutation. Note that $p^\pi_j$ does not depend on the scores $(W_1,\ldots,W_{m+|S|})$, but only on their ranks.

The next proposition demonstrates that the Simes test based on conformal 
$p$-values is inadmissible, as it is strictly dominated by the permutation-based Simes test. Moreover, the two tests coincide when $\alpha(m+1)/|S|$ is an integer.
\begin{proposition}\label{prop:Simes} %
The $\alpha$-level test $\phi^{\mathrm{SimesPerm}}_S$ in~\eqref{eq-SimesPerm} dominates the $\alpha$-level test
$\phi^{\mathrm{Simes}}_S$ in~\eqref{eq-Simeslocaltest}. The domination is strict for at least one combination of $\alpha$, $m$, $n$ and $S\subseteq [n]$.
Moreover, it holds that 
\begin{eqnarray}\label{eq-inequalities}
 \frac{|S|}{m+1} \left \lfloor \alpha \frac{m+1}{|S|} \right \rfloor \leq \mathbb{P}[ \phi^{\mathrm{Simes}}_S = 1 \; ;\; H_{S}] \leq \mathbb{P}[ \phi^{\mathrm{SimesPerm}}_S = 1 \; ;\; H_{S}] \leq \alpha, 
\end{eqnarray} 
and if $\alpha(m+1)/|S|$ is an integer, then all inequalities in~\eqref{eq-inequalities}  becomes equalities, i.e. the Simes test based on conformal $p$-values is of exact size $\alpha$.
\end{proposition}

\begin{proof}
If $\alpha^{\mathrm{Perm}}  > \alpha$, then $\phi^{\mathrm{SimesPerm}}_S  \geq  \phi^{\mathrm{Simes}}_S$. \\
If $\alpha^{\mathrm{Perm}}  < \alpha$, then $\phi^{\mathrm{SimesPerm}}_S  \leq  \phi^{\mathrm{Simes}}_S$. We need to show that if $\alpha^{\mathrm{Perm}}  < \alpha$, then $\phi^{\mathrm{SimesPerm}}_S  =  \phi^{\mathrm{Perm}}_S$ for all $(p_j)_{j \in S}$. We will derive a contradiction. Suppose $\phi^{\mathrm{Simes}}_S  >  \phi^{\mathrm{PermSimes}}_S$ for some $(\tilde{p}_j)_{j \in S}$. The corresponding test statistic $\tilde{T}^{\mathrm{Simes}}_S$ must be $\alpha^{\mathrm{Perm}}< \tilde{T}^{\mathrm{Simes}}_S \leq \alpha$. Since $\tilde{T}^{\mathrm{Simes}}_S=T_{\tilde{\pi}}$ for some $\tilde{\pi} \in \Pi$, it follows that $|\Pi|^{-1}\sum_{\pi \in \Pi} \mathds{1}\{T_\pi \leq \alpha \}  > \alpha$ and we have a contradiction because $\mathbb{P}[ \phi^{\mathrm{Simes}}_S = 1 \; ;\; H_{S}]  \leq \alpha$. \\
If $\alpha^{\mathrm{Perm}}  = \alpha$, then $\phi^{\mathrm{Simes}}_S  =  \phi^{\mathrm{PermSimes}}_S $. By Corollary 3.5 in \cite{marandon2022machine}, if $\alpha(m+1)/|S|$ is an integer, then all inequalities in (1) becomes equalities. 
\end{proof}

Table~\ref{tab:sizeSimes} below presents the size of the tests $\phi^{\mathrm{Simes}}$ and $\phi^{\mathrm{SimesPerm}}$, along with the critical value $\alpha^{\mathrm{Perm}}$, as a function of $m$, with $\alpha=0.1$ and $n=3$. 

\begin{table}[!htb]
\resizebox{\textwidth}{!}{
\begin{tabular}{r|rrrrrrrrrr}
\toprule
$m$  & 9 & 14 & 19 & 24 & 29 & 34 & 39 & 44 & 49 & 54 \\ 
\midrule
   $\mathbb{P}(\phi_S^{\mathrm{Simes}}=1|H_S)$ &  0.000 & 0.022 & 0.013 & 0.009 & 0.100 & 0.086 & 0.075 & 0.072 & 0.064 & 0.058  \\ 
    $\mathbb{P}(\phi_S^{\mathrm{SimesPerm}}=1|H_S)$  & 0.055 & 0.025 & 0.014 & 0.009 & 0.100 & 0.097 & 0.083 & 0.072 & 0.065 & 0.058  \\ 
  $\alpha^{\mathrm{Perm}}$  & 0.200 & 0.133 & 0.100 & 0.080 & 0.100 & 0.171 & 0.150 & 0.111 & 0.100 & 0.091  \\ 
   \hline
\end{tabular}
}
\caption{Size of the tests $\phi^{\mathrm{Simes}}$ and $\phi^{\mathrm{SimesPerm}}$, along with the critical value $\alpha^{\mathrm{Perm}}$, as a function of $m$, with $\alpha=0.1$ and $n=3$.}
\label{tab:sizeSimes}
\end{table}

\subsection{The Fisher Combination Test for Conformal $p$-values} \label{app:local-fisher}

An alternative classical approach for testing $H_S$ based on conformal $p$-values $(p_j)_{j \in S}$ is provided by the Fisher combination method \citep{fisher1970statistical}.
This method originally assumed independence but was later refined by \citet{bates2021testing} to accommodate the (weak) positive dependencies exhibited by conformal $p$-values.
The idea is to evaluate the statistic $T^{\mathrm{Fisher}}_S = - 2 \sum_{j \in S} \log(p_j)$ 
and then reject $H_S$ at level $\alpha$ if and only if $T^{\mathrm{Fisher}}_S > c^{\chi^2}_{\alpha}(|S|) \sqrt{1+|S|/m} - 2 |S| ( \sqrt{1+|S|/m} - 1 )$, where $c^{\chi^2}_{\alpha}(|S|)$ is the $(1-\alpha)$-quantile of the $\chi^2$ distribution with $2|S|$ degrees of freedom. The indicator of this rejection event can thus be written as:
\begin{align}\label{fisher}
    \phi^{\mathrm{Fisher}}_S &= \mathds{1}\left\{ - 2 \sum_{j \in S} \log(p_j) > c^{\chi^2}_{\alpha}(|S|) \sqrt{1+|S|/m} - 2 |S| \left( \sqrt{1+|S|/m} - 1 \right) \right\}.
\end{align}
\citet{bates2021testing} proved that this test is asymptotically valid in the limit of large $|S|$ and $m$, with both sample sizes growing at the same rate.
This differs from the classical Fisher test \citep{fisher1970statistical}, which assumes the $p$-values are mutually independent and can be recovered from \eqref{fisher} by letting $m \to \infty$ while holding $|S|$ fixed.

\subsection{The Wilcoxon-Mann-Whitney Rank Test} \label{app:local-wmw}

An important class of permutation tests is defined in terms of sample ranks, known as \emph{rank tests}. These tests rely on the ranks of the data, making them invariant to monotone transformations, as these transformations do not alter the relative ranking of observations.

For any $S \subseteq [n]$ with cardinality $s$, let $R_1,\ldots,R_{m+s}$ denote the ranks of $(X_1,\ldots,X_m,(Y_j)_{j\in S})$. 
The dependence of $R_1,\ldots,R_{m+s}$ on $S$ is left implicit to simplify the notation, since this does not create ambiguity.

This classical two-sample Wilcoxon \citep{wilcoxon1945individual} or Mann-Whitney \citep{mannwhitney1947} test rejects $H_S$ for large values of the sum of ranks in the test sub-sample $S$
\begin{align}\label{eq-Wstat}
    T^{\mathrm{W}}_S &= \sum_{j=m+1}^{m+s} R_j,
\end{align}
or for large values of the U-statistic 
\begin{align}\label{eq-MWstat}
    T^{\mathrm{MW}}_S &=  \sum_{j \in S}\sum_{i = 1}^{m}  \mathds{1}\{ X_{i} < Y_{j}\}.
\end{align}
The Mann Whitney U-statistic can be expressed as $T^{\mathrm{MW}}_S = T^{\mathrm{W}}_S + |S|(|S|+1)/2$, differing only by a constant term. The two formulations result in the same test, which we will call the Wilcoxon-Mann-Whitney (WMW) test.
The WMW level-$\alpha$ test is given by
\begin{align}\label{eq-MWtest}
    \phi^{\mathrm{WMW}}_S &= \mathds{1}\left\{ T^{\mathrm{W}}_S  \geq c^{\text{W}}_{\alpha}(m,|S|) \right\}=\mathds{1}\left\{ T^{\mathrm{MW}}_S  \geq c^{\text{MW}}_{\alpha}(m,|S|) \right\}.
\end{align}
The critical values $c^{\text{W}}_{\alpha}(m,|S|)$ and $c^{\text{MW}}_{\alpha}(m,|S|)$ are the $(1-\alpha)$ quantiles of the permutation distribution of $T^{\mathrm{W}}_S$ and $T^{\mathrm{MW}}_S$, respectively. Since the two test statistics result in the same test, we will denote them generally by $T^{\mathrm{MWW}}_S$ and the critical value by $c^{\text{WMW}}_{\alpha}(m,|S|)$. 

For small samples, the permutation null distribution of $T^{\text{WMW}}_S$ can be found either via recursion \citep{mannwhitney1947} or direct permutations. 
The test $\phi^{\mathrm{WMW}}_S$ is of exact size $\alpha$, i.e. $\mathbb{P}(\phi^{\mathrm{WMW}}_S = 1 | H_S)= \alpha$, for $\alpha \in \Lambda = \{a_r, r \in [m|S|] \}$, where $a_r=\mathbb{P}(T^{\mathrm{WMW}}_S \geq r | H_S)$; for $\alpha \in (0,1) \setminus  \Lambda$, it is conservative, i.e. $\mathbb{P}(\phi^{\mathrm{WMW}}_S = 1 | H_S) < \alpha$.

A well-known approach when $m$ and $|S|$ are both large is that of \cite{hoeffding1948ustatistics}, which is based on an application of the Central Limit Theorem for U-statistics.   
Under the null hypothesis $H_S$, for large $m$ and $|S|$, the Mann Whitney statistic $T^{\text{MW}}_S$ is approximately normally distributed:
\begin{align}\label{eq-WMWnormalapprox}
  T^{\text{MW}}_S & \approx N\left(\frac{|S|m}{2}, \frac{m|S|(m + |S|+ 1)}{12}\right).
\end{align}
Other asymptotic approximations include the Edgeworth expansion of \citet{fix1955EdgApproximation} and the uniform approximation of \citet{buckle1969UnifApproximation}.

It is also interesting to note that the Mann-Whitney statistic $T^{\text{MW}}_S$ reduces to $T^{\text{MW}}_{j} = \sum_{i=1}^m\mathds{1}\{X_i < Y_j \}$ if $S=\{j\}$, for any $j \in [n]$. The statistic $T^{\text{MW}}_{j}$ can be equivalently expressed as a function of the conformal $p$-value $p_j$ in ~\eqref{eq-permp}:
\begin{align}\label{eq-MWj}
    T^{\mathrm{MW}}_j &=   \sum_{i = 1}^{m} \mathds{1}\{ X_{i} < Y_{j}\}=   m - \sum_{i = 1}^{m} \mathds{1}\{ X_{i} \geq Y_{j}\} =  (m+1)(1-p_j) 
\end{align}
When $S=\{j\}$, the WMW two-sample test, with the 1st sample being $X_1,\ldots,X_n$ and the 2nd ``sample'' being just $Y_j$, rejects $H_j$ if the rank of $Y_j$ in the sequence $(X_1,\ldots,X_n,Y_j)$ is large or, equivalently, if the conformal $p$-value $p_j$ in~\eqref{eq-permp} is small. 
Therefore, the WMW test provides a particularly intuitive bridge between the modern framework of conformal inference and the classical world of two-sample rank tests \citep{kuchibhotla2020exchangeability}.

Another interesting connection, formally stated in the next proposition, is that the Mann-Whitney test for $H_S$
  can be defined by the average of conformal $p$-values test statistic, also known as Edgington's method \citep{edgington1972additive}.

\begin{proposition}
    The Mann-Whitney test statistic $T^{\text{MW}}_S$ can be equivalently expressed as a function of the conformal $p$-values $p_j$ in ~\eqref{eq-permp}:
    \begin{align*}
T^{\text{MW}}_S & = (m+1) \sum_{j \in S}( 1- p_j).
\end{align*}
Then the WMW test can be defined by the average of conformal $p$-values test statistic, also known as Edgington's method  \citep{edgington1972additive}:
   \begin{align*}
\bar{p}_S & = \frac{1}{|S|} \sum_{j \in S}p_j.
\end{align*}
\end{proposition}
\begin{proof}
Note that The Mann-Whitney test statistic $T^{\text{MW}}_S$ is the sum of the individual contributions $T^{\text{MW}}_j$ with $j \in S$, i.e. $T^{\text{MW}}_S = \sum_{j \in S} T^{\text{MW}}_j$. Then 
    $T^{\text{MW}}_S  = (m+1) \sum_{j \in S}( 1- p_j)$ follows from~\eqref{eq-MWj}. Since $T^{\text{MW}}_S$ is a monotonic function of $\bar{p}_S = |S|^{-1}\sum_{j \in S}p_j$, the test statistic can be simplified to the average of conformal $p$-values $\bar{p}_S$, leading to rejection of $H_S$ for small values of $\bar{p}_S$.
\end{proof}

The Edgington's method has been studied by \cite{ruschendorf1982random} and \cite{meng1994posterior}. They show that twice the average of the $p$-values is a valid $p$-value for arbitrary dependence. Recently, \cite{choi2023averaging} demonstrated that the ``twice the average'' rule cannot be improved even under the assumption of exchangeability of the $p$-values.
In contrast, the WMW test, when utilizing the average of the conformal $p$-values as the test statistic, rejects $H_S$ if $\bar{p}_S \leq \alpha^{\mathrm{Perm}}$, where the critical value $\alpha^{\mathrm{Perm}}$ is determined by taking the $\alpha$ quantile of the permutation distribution of $\bar{p}_S$, ensuring exact validity. This critical value $\alpha^{\mathrm{Perm}}$ may be calculated either exactly in small samples or through asymptotic approximations. A well-known asymptotic approach when $m$ and $|S|$ are both large is that of \cite{hoeffding1948ustatistics}, which leads to a much more powerful test compared to the combination approach for arbitrary dependence. 
In particular, the WMW test rejects $H_S$ when
\begin{align*}
\bar{p}_S \leq  & \, \frac{m+2}{2(m+1)} - z_{1-\alpha}  \sqrt{\frac{m(|S|+m+1)}{12 |S|(m+1)^2 }},
\end{align*}
where $z_{1-\alpha}$ is the $(1-\alpha)$ quantile of the standard normal distribution. For $\alpha=0.1$ and $m=|S|=100$, this translates to rejecting $H_S$ if $\bar{p}_S \leq 0.453$, which is a much better critical value than the $\alpha/2=0.05$ required for arbitrary dependence.

\cite{lehmann53power} investigated the power of rank tests in the case when the distribution of the scores assigned to outliers in the test set is equal to $G(F)=F^{k}$ for any integer $k\geq 2$, which is referred to as \textit{Lehmann's alternative} \citep{balakrishnan2021my}.
Here, $F^{k}$ represents the distribution of the maximum of $k$ independent random variables with distribution $F$. If $X \sim F$ and $Y \sim F^{k}$, then $\mathbb{P}(X<Y)=k/(k+1)$. It is worth noting that the distance of $F^{k}$ from $F$ remains the same for any specification of $F$.

Under Lehmann's alternative $G(F) = F^k$, for integers $k\geq2$, the associated density function is $g(u)= ku^{k-1}\I{u\in[0,1]}$ and the test statistics in~\eqref{eq:def-TG} reduces to:
\begin{align}\label{eq:def-Tk}
  T^{\text{WMW, k}} & = \frac{k}{(N+1)\cdot\ldots\cdot(N+k-1)}\sum_{j=m+1}^N\prod_{l=0}^{k-2} (R_j+l),
\end{align}
which are the locally most powerful rank tests against the class of Lehmann's alternatives. 
For $k=2$, the test statistic in~\eqref{eq:def-Tk} differs from the WMW statistic only by a multiplicative constant and therefore yields an equivalent test.

Rank tests against Lehmann’s alternative were also studied in 
\cite{conover1988locally} and \cite{rosenbaum2007confidence}.

\subsection{Shirashi’s Locally Most Powerful Rank Test} \label{app:local-shirashi}

\begin{theorem}[\cite{shiraishi1985local}]\label{thm:LMPI}
    Assume $X_1,\ldots,X_m$ i.i.d.~with continuous c.d.f.~$F$, and $Y_1,\ldots,Y_n$ i.i.d. with c.d.f. $(1-\theta)F + \theta G(F)$ for some $\theta \in [0,1]$ and some continuous c.d.f. $G$ on $[0,1]$  with density $g$.
    Then, the LMPR test rejects  $H_0: \theta = 0$ for large values of:
\begin{eqnarray*}
    T^{\text{g}}  = \sum_{j = m+1}^{N} \E{g(U_{N}^{(R_{j})})}.
\end{eqnarray*}
\end{theorem}

Under the setup of Section~\ref*{sec:new_local-g-wmw}, we now write the local asymptotic power function of Shiraishi's LMPR test and compare its relative efficiency with that of the Neyman–Pearson optimal test. Let $\pi^{\phi}_N(\theta) = \mathbb{E}_\theta( \phi )$ denote the power function of a test $\phi$ of size $\alpha$ for  $\theta \in [0,1]$ and overall sample size $N$, and assume that $n/N \rightarrow \lambda \in (0,1)$. Consider a sequence of local alternatives $\theta_N = h/\sqrt{N}$ converging to the null hypothesis for some $h>0$.
The \emph{local asymptotic power function} of Shiraishi's LMPR test $\phi^g$
 is given by \citet[Theorem 1]{shiraishi1985local} as
$$\pi^{\phi^{\mathrm{g}}}(h) := \lim_{N\rightarrow \infty}\pi^{\phi^{\mathrm{g}}}_N\Big(\frac{h}{\sqrt{N}}\Big) = \Phi\Big(z_{\alpha} + h \sqrt{I_g} \Big),\qquad I_g = \lambda(1-\lambda)\mathbb{V}\mathrm{ar}(g(U)),$$
where $I_g$ is the local asymptotic information. 
The local asymptotic power is maximized at $\lambda=1/2$, so the Shiraishi's test is most powerful for balanced calibration and test sample sizes. This follows from the fact that the number of cross-group comparisons is of order 
$mn = \lambda(1-\lambda)N^2$, which is maximized at $\lambda=1/2$. This observation provides a concrete and practically relevant guideline for test design.

By contrast, the local asymptotic power function of the optimal test $\phi^{opt}$, which is the likelihood ratio test (LRT) by the Neyman-Pearson lemma and assumes oracle knowledge of $F$, $G$ and $\theta$ (see also Appendix~\ref{app:optimal_scoring}), is given by
$$\pi^{\phi^{opt}}(h) := \lim_{N\rightarrow \infty}\pi^{\phi^{opt}}_N\Big(\frac{h}{\sqrt{N}}\Big) = \Phi\Big(z_{\alpha} + h \sqrt{I^{opt}_g} \Big), \qquad I^{opt}_g = \lambda\mathbb{V}\mathrm{ar}(g(U)).$$
The power of the LRT increases with $\lambda$ because the test statistic depends only on the test sample. Since the inlier distribution $F$ is assumed known, the calibration sample is not used. It follows that the asymptotic relative efficiency (ARE), or Pitman efficiency, of Shiraishi's LMPR test relative to the LRT optimal test equals $1-\lambda$ \citep[Corollary 1]{shiraishi1985local}. 
This shortfall in asymptotic efficiency is not a limitation of the LMPR test per se, but rather a consequence of the two-sample setting. In the one-sample setting, the LMPR test may achieve optimal power, as shown in \cite{shiraishi1986optimum}.

The dense and sparse regimes $\sqrt{N}\theta_N \rightarrow \infty$ and $\sqrt{N}\theta_N \rightarrow 0$ are studied nonparametrically in \citet{huang2023detecting} under a location-shift model for the outliers.

\section{Estimating the Outlier Distribution} \label{app:outlier-distr}

This section details how to empirically estimate the outlier distribution $G$ and its density function $g$, which is essential for implementing a data-driven version of Shirashi's local testing procedure as described in Section~\ref{sec:new_local-g-wmw}.
Recall from Section~\ref{sec:new_local-g-wmw} that Shirashi's optimality result \citep{shiraishi1985local} is based on the following mixture model:
\begin{align} \label{eq:lmpi-mixture-model}
\begin{split}
  X_1,\ldots,X_m & \overset{\text{i.i.d.}}{\sim} F, \\
  Y_1,\ldots,Y_n & \overset{\text{i.i.d.}}{\sim} (1-\theta)F + \theta G(F),
\end{split}
\end{align}
where $\theta \in [0,1]$ represents the proportion of outliers, and $G$ is a distribution function on $[0,1]$ with density $g$. 
In the following, we describe a practical and effective approach for computing an estimate $\hat{g}$ of $g$ using the data in $(X_1,\ldots,X_m,Y_1,\ldots,Y_n)$, while ensuring that an approximately valid and approximately LMPR test can still be performed using Shirashi's method conditional on $\hat{g}$.

\subsection{Diluting the Test Set for Conditional Exchangeability}

The key idea, inspired by approaches proposed in different contexts by \citet{liang2022integrative} and \citet{marandon2022machine}, is to randomly split the calibration sample $(X_1,\ldots,X_m)$ into two disjoint subsets, namely $(X_1,\ldots,X_{m_1})$ and $(X_{m_1+1}, \ldots, X_{m})$, for some choice of $m_1$ such that $1 < m_1 < m$, like $m_1 = \lceil m/2 \rceil$.
Intuitively, $(X_1,\ldots,X_{m_1})$ is used as a smaller calibration set and $(X_{m_1+1}, \ldots, X_{m}, Y_1,\ldots,Y_n)$ as a ``diluted'' test set when fitting $g$.
More precisely, the estimation of 
$g$ is carried out by fitting the following mixture model:
\begin{align} \label{eq:lmpi-mixture-model-2}
\begin{split}
  X_1,\ldots,X_{m_1} & \overset{\text{i.i.d.}}{\sim} F, \\
  X_{m_1+1}, \ldots, X_{m}, Y_1,\ldots,Y_n & \overset{\text{i.i.d.}}{\sim} (1-\theta_1)F + \theta_1 G(F),
\end{split}
\end{align}
where $\theta_1 = (\theta n + m_1) / (n + m_1)$ and $G$ is the same outlier distribution function as in Equation~\eqref{eq:lmpi-mixture-model}.
Any approach may be applied to fit this mixture model, as long as it is invariant to the ordering of the data in $(X_{m_1+1}, \ldots, X_{m}, Y_1,\ldots,Y_n)$. Otherwise, we simply need to permute the scores of the diluted test set before estimating $g$. 
The motivation for this method is that it ensures the inlier scores among $(Y_1,\ldots,Y_n)$ remain exchangeable with the smaller calibration set $(X_{m_1+1}, \ldots, X_{m})$ conditional on the estimated $\hat{g}$. 
Therefore, any downstream conformal inferences that utilizes $(X_{m_1+1}, \ldots, X_{m})$ as a reference (or calibration) set are still valid.
This is akin to the idea underlying AdaDetect, the individual outlier identification method  proposed by \citet{marandon2022machine}.


A subtle point worth highlighting is that conditional on $\hat{g}$ the inliers among $(Y_1,\ldots,Y_n)$ are only exchangeable with $(X_{m_1+1}, \ldots, X_{m})$ but not independent. 

Nevertheless, Shiraishi’s adaptive rank test, which replaces the unknown outlier density $g$ with its estimate $\hat{g}$, retains finite-sample validity, and the asymptotic approximation of its null distribution remains valid, as shown in Appendix~\ref{app:adaptive_validity}.

\subsection{An Intuitive Mixture Modeling Approach}

Although there is an extensive literature on fitting mixture models like the one in~\eqref{eq:lmpi-mixture-model-2}, covering a wide range of parametric and nonparametric methods, a full review is beyond our scope. Instead, we present a simple and intuitive approach that has proven effective for our purposes.
\begin{enumerate}

\item Rescale the data, $X_1,\ldots,X_{m},Y_{1},\ldots,Y_{n}$, to be between 0 and 1, if they are not already scaled as such (many standard classifiers output scores in this range by default).

\item Fit a $\text{Beta}(b_1,b_2)$ distribution to $X_1,\ldots,X_{m_1}$ via maximum likelihood. To avoid numerical instabilities from values too close to 0 or 1, threshold all data points between 0.001 and 0.999. Let $\hat{b}_1$ and $\hat{b}_2$ be the estimated parameters.

\item Apply an inverse CDF transform based on the fitted $\text{Beta}(\hat{b}_1,\hat{b}_2)$ distribution to $(X_{m_1+1}, \ldots, X_{m}, Y_1,\ldots,Y_n)$, making the inlier distribution approximately uniform on $[0,1]$. Let $(\tilde{X}_{m_1+1}, \ldots, \tilde{X}_{m}, \tilde{Y}_1,\ldots,\tilde{Y}_n)$ denote the transformed test set.

\item Fit the following mixture model via maximum likelihood:
\begin{align} \label{eq:lmpi-mixture-model-3}
\tilde{X}_{m_1+1}, \ldots, \tilde{X}_{m}, \tilde{Y}_1,\ldots,\tilde{Y}_n & \overset{\text{i.i.d.}}{\sim} (1-\theta_1) \text{Uniform}(0,1) + \theta_1 G,
\end{align}
where $G$ is approximated by a $\text{Beta}(b'_1,b'_2)$ distribution. The estimated parameters are denoted as $\hat{b}'_1$ and $\hat{b}'_2$.

\item Output the estimated outlier distribution: $\hat{G} \approx \text{Beta}(\hat{b}'_1,\hat{b}'_2)$.

\end{enumerate}

While this approach may seem somewhat heuristic, as it is not aimed at making rigorous inferences about $g$, it is sufficient for our purposes. The invariance of this estimation procedure to the ordering of $(\tilde{X}_{m_1+1}, \ldots, \tilde{X}_{m}, \tilde{Y}_1,\ldots,\tilde{Y}_n)$ ensures the validity of our test regardless of how closely $\hat{g}$ approximates the true $g$. Moreover, as demonstrated in our numerical experiments, this estimation method performs well enough in practice to provide our testing method with a noticeable advantage over standard approaches, such as the WMW test, particularly in scenarios where the oracle Shiraishi test has a significant edge.

\subsection{Enabling Closed Testing Shortcuts for Monotone Statistics}

When applied within a closed testing framework, Shirashi's testing procedure can become computationally expensive without an appropriate shortcut. As discussed in Appendix~\ref{app:shortcuts},  a fast and exact shortcut exists, but it relies on the monotonicity of the outlier probability density.

Although more sophisticated approaches have been studied to address this issue \citep{patra2016estimation}, a simple post-hoc solution that works well for our purposes is to approximate the probability density function of the $\text{Beta}(\hat{b}'_1,\hat{b}'_2)$ distribution, obtained as described in the previous section, with a monotone increasing function using isotonic regression on a finite grid of values. Since it may not be clear in advance whether the best approximation of the outlier density is increasing or decreasing, the direction of monotonicity is adaptively chosen by comparing the residual sum of squares for the two corresponding isotonic regression models.

While this method is somewhat heuristic, we have found it to be more reliable in practice than other, more sophisticated approaches we tried. Nonetheless, it is likely that this approach could be further refined with some additional effort.








\subsection{Shiraishi’s Adaptive Rank Test: Finite-Sample Validity and Asymptotic Null Distribution}\label{app:adaptive_validity}

Let $\hat{F}$ denote any estimator of $F$ based on $(X_1,\ldots,X_{m_1})$. Define the probability integral-transformed variables 
$\tilde{X}_{j} = \hat{F}(X_{j})$, $j =m_1+1,\ldots,m$, and $\tilde{Y}_{j} = \hat{F}(Y_{j})$, $j \in [n]$, and estimate $g$ by fitting the mixture model
$$\tilde{X}_{m_1+1}, \ldots, \tilde{X}_{m}, \tilde{Y}_1,\ldots,\tilde{Y}_n  \overset{\text{i.i.d.}}{\sim} (1-\theta_1) \mathrm{Uniform}(0,1) + \theta_1 G,$$
where $G$ has density $g$ on $[0,1]$. 
Let $\hat{g}$  be the resulting estimator, and assume that it is invariant under permutations of
$(\tilde W_1,\ldots,\tilde W_{\tilde N}) = (\tilde{X}_{m_1+1}, \ldots, \tilde{X}_{m}, \tilde{Y}_1,\ldots,\tilde{Y}_n)$, where $\tilde N = m-m_1 + n$. 
No further assumptions on  $\hat g$ are imposed beyond the condition
(\ref{eq:regulation_condition_adaptive}) discussed below, which ensures the validity of the asymptotic approximation. In particular, the i.i.d. assumption underlying the fitted mixture model need not hold. 

Let $\tilde R_1,\ldots, \tilde R_{\tilde N}$ denote the ranks of $\tilde W_1,\ldots,\tilde W_{\tilde N}$.
The Shiraishi adaptive rank test statistic is defined as
\begin{align} \label{eq:adaptive_Shiraishi_stat}
  T^{\hat{\text{g}}} & = \sum_{j = m-m_1+1}^{ \tilde N} \E{\hat{g}(U_{\tilde N}^{(\tilde R_{j})})},
\end{align}
and the corresponding test by 
\begin{align} \label{eq:adaptive_Shiraishi_test}
\phi^{\hat{g}} = \mathds{1}\left\{ T^{\hat{g}}  > c_\alpha^{\hat{g}} \right\}.
\end{align}
Here $c_\alpha^{\hat{g}}$ denotes the permutation critical value, defined as the $k$-th ordered statistic of $\pi_1 T^{\hat{g}}, \ldots, \pi_{\tilde N !} T^{\hat{g}}$
where each $\pi$ is a permutation of $[\tilde N]$, and $\pi T^{\hat{g}}$ is the test statistic $T^{\hat{g}}$ computed on  $(\tilde W_{\pi(1)},\ldots, \tilde W_{\pi(\tilde N)})$, that is,
$$\pi T^{\hat{g}} = \sum_{j=m-m_1+1}^{\tilde N} \E{\hat{g}(U_{\tilde N}^{(\tilde R_{\pi(j)})})}.$$ 
The index $k$ is given by $k = \lceil (1-\alpha) N! \rceil$.

Finite-sample validity of the adaptive Shiraishi rank test follows from the exchangeability of the transformed vector
$(\E{\hat{g}(U_{ \tilde N}^{(\tilde R_1)})}, \ldots, \E{\hat{g}(U_{ \tilde N}^{(\tilde R_{\tilde N})})})$, as shown formally in the following proposition. 

\begin{proposition}\label{prop:adaptive_exchangeability}
If the vector $(\tilde X_{m_1+1},\ldots, \tilde X_n, \tilde Y_1,\ldots, \tilde Y_n)$ is exchangeable, then the transformed vector
$$(\E{\hat{g}(U_{ \tilde N}^{(\tilde R_1)})}, \ldots, \E{\hat{g}(U_{ \tilde N}^{(\tilde R_{\tilde N})})})$$ is also exchangeable. 

\end{proposition}

The proof relies on the following lemma, which states that a transformation of exchangeable random variables is itself exchangeable whenever applying a permutation after the transformation is equivalent to applying the transformation after permuting the original variables.

\begin{lemma}[\cite{dean1990linear}]
Suppose $(W_1,\ldots,W_k) \in \mathcal{W}^k$ is a vector of exchangeable random variables. Fix a transformation $A:  \mathcal{W}^k \rightarrow  (\mathcal{W}')^q$. If for each permutation $\pi_1: [q] \rightarrow [q]$ there exists a permutation $\pi_2: [k] \rightarrow [k]$ such that 
$$\pi_1 A(w) = A(\pi_2 w) \quad \mathrm{for\,\,all\,\,} w \in \mathcal{W}^k,$$
then $A(\cdot)$ preserves exchangeability. 
\end{lemma}

\begin{proof}[Proof of Proposition~\ref{prop:adaptive_exchangeability}]
Let $A: \mathbb{R}^{\tilde N} \rightarrow \mathbb{R}^{\tilde N}$ denote the transformation that maps the vector $(\tilde X_{m_1+1},\ldots, \tilde X_m, \tilde Y_1,\ldots, \tilde Y_n)$ to the vector $(\E{\hat{g}(U_{ \tilde N}^{(\tilde R_1)})}, \ldots, \E{\hat{g}(U_{ \tilde N}^{(\tilde R_{\tilde N})})})$. 
Let $\pi:[ \tilde N] \rightarrow [ \tilde N]$ any permutation of the indices. Then for all $( w_1,\ldots, w_{\tilde N} ) \in \mathbb{R}^{ \tilde N}$
$$\pi A( w_1,\ldots, w_{\tilde N} ) = A( \pi( w_1,\ldots, w_{\tilde N}) ).$$
Thus, $A$ preserves exchangeability. 
\end{proof}

We now state a parallel result to Theorem~\ref{thm:asymp_distr_adaptive-exch} for the asymptotic approximation of the null distribution of Shiraishi’s adaptive test.

\begin{theorem}
\label{thm:asymp_distr_adaptive-exch}
    Suppose that
\begin{eqnarray}\label{eq:regulation_condition_adaptive}
        \frac{\max_{1\leq r \leq \tilde N}(\E{\hat g(U_{\tilde N}^{(r)})} -\hat \mu_{\tilde N})^2}{\min(n,m-m_1) \hat \sigma^2_{\tilde N}} \xrightarrow{\mathbb{P}} 0,
    \end{eqnarray}
    where $$\hat \mu_{\tilde N} = \frac1{\tilde N}\sum_{r\in[\tilde N]}\E{\hat g(U_N^{(r)})},\qquad  
    \hat \sigma^2_{\tilde N} = \frac1{\tilde N-1} \sum_{r \in [\tilde N]}( \E{\hat g(U_{\tilde N}^{(r)})} -\hat \mu_{\tilde N} )^2,$$
    and where $\hat g$ is invariant under permutations of the pooled sample $(\tilde X_{m_1+1},\ldots, \tilde X_m, \tilde Y_1,\ldots,\tilde Y_n)$. 
    Then, under the null hypothesis 
    $\tilde H^*_0: \mathrm{the\,\,vector\,\,}(\tilde X_{m_1+1},\ldots, \tilde X_m, \tilde Y_1,\ldots,\tilde Y_n)\,\,\mathrm{is\,\,exchangeable},$ the standardized Shiraishi adaptive statistic converges in distribution to a standard normal,
    \begin{eqnarray}\label{eq:rescaled-LMPI}
        \frac{T^{\hat{\mathrm{g}}} - n\hat\mu_N}{\sqrt{\frac{(m-m_1) n}{\tilde N} \hat{\sigma}^2_{\tilde N}}} \xrightarrow{d} N(0,1).
    \end{eqnarray}
\end{theorem}

\begin{proof}
Conditional on the pooled data
\((\tilde X_{m_1+1},\ldots,\tilde X_m,\tilde Y_1,\ldots,\tilde Y_n)\),
the estimator \(\hat g\) is fixed, and the quantities
\(\E{\hat{g}(U_{ \tilde N}^{(1)})}, \ldots, \E{\hat{g}(U_{ \tilde N}^{(\tilde N)})}\)
are deterministic values.
Since condition~\eqref{eq:regulation_condition_adaptive} holds in probability,
it holds for all realizations of the pooled data except on a set of probability
tending to zero. We restrict attention to such realizations in what follows.

The Shiraishi adaptive statistic can be written as
    \begin{align*} 
  T^{\hat{\text{g}}} & = \sum_{r \in [\tilde N]} \E{\hat{g}(U_{\tilde N}^{(r)})} C_r,
\end{align*}
where $C_r = 1$ if the unit with rank $r$ is assigned to the test sample, and $C_r = 0$ otherwise.
Exchangeability under \(\tilde H_0^*\) implies that, conditional on the pooled data,
\((C_1,\ldots,C_{\tilde N})\) is uniformly distributed over all vectors with exactly
\(n\) ones. Thus, the null distribution of \(T^{\hat g}\) is induced entirely by
permutations of \((C_1,\ldots,C_{\tilde N})\), which is equivalent to sampling
\(n\) units without replacement from the finite population
$
\{\E{\hat{g}(U_{ \tilde N}^{(1)})}, \ldots, \E{\hat{g}(U_{ \tilde N}^{(\tilde N)})}\}$.
Consequently, \(T^{\hat g}/n\) is the sample mean of this finite population, and the
result follows from Theorem~1 of \citet{li2017general}.
\end{proof}

Note that if $\hat g$  is bounded in probability, then  condition~\eqref{eq:regulation_condition_adaptive} holds. This property is satisfied, for example, by the nonparametric estimator $\hat g$ of \citet{patra2016estimation}, and may alternatively be ensured via truncation of $\hat g$.

\subsection{Shiraishi’s Adaptive Rank Test: Local Asymptotic Power}\label{app:adaptive_localpower}

We now show that, under suitable assumptions, the estimator $\hat{g}$ of \citet{patra2016estimation} satisfies condition \eqref{eq-L2consistency}. 
Consequently, the adaptive rank test based on the Patra-Sen estimator has the same local asymptotic distribution as the oracle LMPR test and therefore achieves the same local asymptotic power. Moreover, the closed-testing shortcut described in Appendix A3.4 applies, giving an adaptive test that is directly usable in practice.

In the procedure described in Appendix~\ref{app:outlier-distr}, the calibration sample $(X_1,\ldots,X_m)$ is split into two parts. 
We estimate \(F\) by the empirical distribution function \(\hat F\) based on \(X_1,\ldots,X_{m_1}\). Define the probability integral-transformed variables 
$\tilde{X}_{j} = \hat{F}(X_{j})$, $j =m_1+1,\ldots,m$, and $\tilde{Y}_{j} = \hat{F}(Y_{j})$, $j \in [n]$, so that the inlier distribution is discrete Uniform on $[0,1]$.
The distribution of the inlier can be made exactly Uniform on $[0,1]$  by applying the randomized probability integral-transform
$$\tilde W = \hat{F}(W - \epsilon) + U[\hat{F}(W) - \hat{F}(W-\epsilon)]$$
for a sufficiently small $\epsilon>0$, where $U\sim \mathrm{Uniform}(0,1)$. 

Let  $\tilde W_1,\ldots, \tilde W_{\tilde N}$ denote the pooled transformed sample, where $\tilde N=(m-m_1)+n$. 
Conditional on $(X_1,\ldots,X_{m_1})$, the pooled sample satisfies
$$\tilde W_1,\ldots, \tilde W_{\tilde N} \stackrel{\mathrm{i.i.d.}}{\sim} (1-\theta_1) \mathrm{Uniform}(0,1) +  \theta_1 G.$$
We further assume that $G$ has a continuous, bounded, and non-increasing density $g$, and the mixture model is identifiable, for example by the condition $g(1)=0$, placing us in the setting required for the estimator $\hat{g}_{\tilde N}$ of \citet{patra2016estimation}. 

The Grenander-type estimator $\hat g_{\tilde N}$ proposed by
\citet{patra2016estimation} satisfies, for any $p\in[1,2)$,
\[
\int_0^1 \bigl|\hat g_{\tilde N}(u)-g(u)\bigr|^p\,du
= O_{\mathbb P}\!\left(\tilde N^{-p/3}\right),
\]
provided that assumptions (A1), (A2), and (A2$'$) in
\citet{durot2007l_p} are satisfied; see Corollary~1 therein.

The estimation error $(\hat g_{\tilde N}-g)$ is uniformly bounded in probability
because both terms are uniformly bounded. 
This boundedness, together with the above $L^p$ convergence,
implies convergence in $L^2$: for any $p\in[1,2)$,
\[
\int_0^1 \bigl|\hat g_{\tilde N}(u)-g(u)\bigr|^2\,du
\le
\biggl(\sup_{u\in[0,1]} \bigl|\hat g_{\tilde N}(u)-g(u)\bigr|\biggr)^{2-p}
\int_0^1 \bigl|\hat g_{\tilde N}(u)-g(u)\bigr|^p\,du
\xrightarrow{\mathbb P} 0.
\]

Let $u_r=r/(\tilde N+1)$ and set $\delta_{\tilde N}(u)=\hat g_{\tilde N}(u)-g(u)$. Then
\[
\frac{1}{\tilde N}\sum_{r=1}^{\tilde N} (\hat b_r-b_r)^2
= \int_0^1  \delta_{\tilde N}(u)^2\,du + E_{\tilde N},
\]
where
\[
E_{\tilde N}
:= \frac{1}{\tilde N}\sum_{r=1}^{\tilde N}  \delta_{\tilde N}(u_r)^2
   - \int_0^1  \delta_{\tilde N}(u)^2\,du.
\]
Since $g$ and $\hat g_{\tilde N}$ are non-increasing and uniformly bounded,
they have bounded total variation, denoted by $\mathrm{TV}(g)$ and
$\mathrm{TV}(\hat g_{\tilde N})$. Consequently, $\delta_{\tilde N}$ also has bounded
total variation, and
\[
\mathrm{TV}\!\left( \delta_{\tilde N}^2\right)
\le 2\| \delta_{\tilde N}\|_\infty\,\mathrm{TV}( \delta_{\tilde N})
= O_{\mathbb P}(1).
\]
It follows that
\[
|E_{\tilde N}|
\le \frac{\mathrm{TV}( \delta_{\tilde N}^2)}{\tilde N}
= o_{\mathbb P}(1),
\]
and therefore
\[
\frac{1}{\tilde N}\sum_{r=1}^{\tilde N}(\hat b_r-b_r)^2
\xrightarrow{\mathbb P} 0.
\]

Other nonparametric estimators for the unknown outlier density are studied in \cite{matias2014nonparametric} and \cite{liang2025nonparametric}.

\section{Computational Shortcuts for Closed Testing} \label{app:shortcuts}

In general, computing $d(S)$ in~\eqref{eq-dS} involves the evaluation of exponentially many tests, which hinders its practical application. For specific local tests, however, there exist polynomial-time \emph{shortcuts} \citep{goeman2019simultaneous, goeman2021onlyclosed, tian2023Large-scale}. 

Shortcuts fall into two categories: exact shortcuts that yield identical rejections
to the full closed testing procedure, and approximate or conservative shortcuts that may produce fewer rejections than the full procedure. The shortcuts presented in this section are exact, with the exception of an approximate shortcut discussed at the end of the section, based on the method of \cite{gazin2024transductive}. 

The table below summarizes the computational complexity of these shortcuts for computing the lower bound $d$ on the total number of outliers, as well as the lower bound  $d(S)$ for a subset $S$ of the test sample. In all exact shortcuts considered in this paper, computation of $d(S)$ requires first computing $d$. Since $d$ does not depend on $S$, it needs to be computed only once. 

The reported computational complexity includes a preparatory step
in which the 
$m$ calibration and $n$ test scores are sorted, the ranks of the pooled scores are computed, and the conformal $p$-values are obtained from these ranks, and are therefore already ordered.
This requires \(\mathcal{O}(n \log n)\) time when $m$ is of the same order as $n$.
Existing shortcuts for Simes, WMW, and Fisher local tests are implemented efficiently, with computational complexity linearithmic in
$n$ for both $d$ and $d(S)$, matching the cost of sorting $n$ values.
The quadratic and cubic complexities for computing $d$ and $d(S)$ for Storey-Simes and Shiraishi local tests arise from naive implementations of Algorithms~\ref{alg-ASimesCT} and~\ref{alg:ShiraishiCT}, and could likely be further refined with additional effort.

\begin{table}[H]
    \centering
    \begin{tabular}{c|c|c}
    \toprule
        Local test & \multicolumn{2}{c}{Computational complexity for computing the lower bound}  \\
                & for the total number of outliers & for the number of outliers in a subset  \\
          \midrule
        Simes & $\mathcal{O}(n\log n)$ & $\mathcal{O}(n\log n)$  \\
        Storey-Simes & $\mathcal{O}(n^2)$ & $\mathcal{O}(n^3)$  \\
        Shiraishi & $\mathcal{O}(n^2)$ & $\mathcal{O}(n^3)$  \\
        WMW / Fisher & $\mathcal{O}(n\log n)$ & $\mathcal{O}(n\log n)$  \\
        \hline
    \end{tabular}
    \caption{Computational cost of the shortcut procedures for Algorithm~\ref{alg:closed-testing} for computing the lower bound $d$ for the total number of outliers and the lower bound $d(S)$ on the number of outliers in a subset $S$ of the test set,  as a function of the chosen local test. 
    }
    \label{tab:comp-costs}
    \end{table}

\subsection{Shortcuts for Simes Local Tests} \label{app:shortcuts-simes}

With Simes' local test, the shortcut described in~\cite{goeman2019simultaneous} and given in Algorithm~\ref{alg-SimesCT} is exact and allows calculating $d_{\mathrm{Simes}}(S)$ in linear time, after an initial preparatory step that takes $\mathcal{O}(n \log n)$ time.

\begin{algorithm}[!htb]
\SetKwInOut{Input}{Input}
\Input{$p$-values $p_1,\ldots,p_n$;  significance level $\alpha \in (0,1)$.}

Sort the $p$-values $p_{(1)} \leq \ldots \leq p_{(n)}$;

Compute $h=\max\{0\leq k \leq n: p_{(n-k+j)} > j\alpha/k \,\,\mathrm{for\,\,} j=1,\ldots,k\}$;

Compute the lower bound on the total number of outliers $d_{\mathrm{Simes}} = d_{\mathrm{Simes}}([n]) = n-h$;

\textbf{Query \(S\subset[n]\):} compute $d_{\mathrm{Simes}}(S) = \min\{0\leq k \leq |S|: p_{(k+j:S)} > j\alpha/h \,\,\mathrm{for\,\,} j=1,\ldots,|S|-k\}$;

\SetKwInOut{Output}{Output}
\Output{A simultaneous $(1-\alpha)$-confidence lower bound $d_{\mathrm{Simes}}(S)$ for the number of outliers in $S$, for $S=[n]$ and for any queried subset $S\subset[n]$.}

\caption{Shortcut for closed testing with Simes' test~\citep{goeman2019simultaneous}}\label{alg-SimesCT}
\end{algorithm} 

Line 2 of Algorithm~\ref{alg-SimesCT} gives $h$, which is a $(1-\alpha)$-confidence upper bound for the number of true inliers $|I_0|$ in the overall test set. Then  $d_{\mathrm{Simes}} = n - h$ is a $(1-\alpha)$-confidence lower bound for the number of false hypotheses (cfr.~number of true outliers in the test sample) $|I_1|$.

The index set of the discoveries (cfr. localized outliers) is given by $D_{\mathrm{Simes}} =  \{ j \in [n]: d_{{\mathrm{Simes}}}(\{j\}) = 1\}
    = \{ j \in [n]: p_{j} \leq \alpha/h \}.$
All hypotheses $H_j$ with index $j \in D_{\mathrm{Simes}}$ are rejected by closed testing while controlling the familywise error rate (FWER) at level-$\alpha$ \citep{hommel1988stagewise}: $\mathbb{P}( |D_{\mathrm{Simes}} \cap I_0 |>0 ) \leq \alpha$.

\subsection{Closed Testing with Simes Local Tests vs BH} \label{app:closed-testing-simes-bh}

The Simes' test and the Benjamini-Hochberg (BH) procedure share the same critical values, and we can expect some relationships between the two methods.

Let $\phi$ denote a level-$\alpha$  test for the global null hypothesis $H = H_{[n]}$. Closed testing provides the following quantities: 
\begin{itemize}
    \item $\mathds{1}\{\phi=1\}$ to indicate outlier detection; 
    \item $d$ to provide outlier enumeration;
    \item $|D|$ to quantify outlier discovery.
\end{itemize}
Outlier enumeration involves counting the number of true outliers in the test sample without necessarily pinpointing their exact locations. Outlier enumeration is implied by outlier discovery (by using $|D|$) and, in turn, implies outlier detection (by using $\mathds{1}\{d>0\}$).
These quantities can be applied to any subset $S$ of the test sample, obtaining $\mathds{1}\{\phi_S=1\}$, $d(S)$, and $|D \cap S |$. 
Simultaneous confidence bounds $d(S)$ obtained by closed testing satisfy the following relations:
\begin{eqnarray}\label{CTrelations}
\mathds{1}\{\phi_S=1\} \geq \mathds{1}\{d(S)>0\}, \qquad d(S) \geq |D \cap S |.
\end{eqnarray}

The first relation,
    $1\{\phi_S=1\} \geq 1\{d(S)>0\}$, follows from the identity $1\{d(S)>0\} = 1\{\phi_K=1, \forall \,K \supseteq S\}$. It highlights that a procedure which fixes $S$ a priori and tests only $H_S$ is at least as powerful as testing the same hypothesis with correction for multiple testing.
In the relevant case $S=[n]$, however, the two approaches coincide. The second relation emphasizes that the number of outliers obtained by enumeration is at least as large as the number of outliers identified.

The BH algorithm applied to $p_1,\ldots,p_n$ at level $\alpha$ returns the index set of the discoveries:
\begin{eqnarray}\label{eq-DBH}
D_{\mathrm{BH}} &=&  \{j \in [n]:  p_j \leq \alpha d_{\mathrm{BH}}/n\}
\end{eqnarray}
where
\begin{eqnarray}\label{eq-dBH}
d_{\mathrm{BH}} &=&  \max\left\{k \in \{0,\ldots,n\}: \sum_{j=1}^{n} \mathds{1}\{p_j \leq \alpha k/n \} \geq k \right\}.
\end{eqnarray}
is the number of discoveries made by the BH procedure. To distinguish between discoveries with FWER and FDR guarantees, we will refer to discoveries made by closed testing as FWER discoveries, while those made by an FDR controlling procedure will be called FDR discoveries.

The Simes' test for $H$ and the BH procedure share the same critical values, and we can expect some relationships between the two methods.
A simple observation is that the event where the CT$_{\mathrm{Simes}}$ and the BH procedure make at least one non-trivial statement is identical: 
\begin{eqnarray}\label{eq-weakFWER}
\mathds{1}\{d_{\mathrm{Simes}} > 0\}=\mathds{1}\{d_{\mathrm{BH}} > 0\}.
\end{eqnarray}
Consequently, the two methods have the same weak FWER control and the same power for outlier detection. 

This is illustrated in the right panel of Figure~\ref{fig:exp-lhco}, where the solid gray line (Simes) and the dotted gray line (BH) coincide, indicating that both procedures have identical power for rejecting the global null.
Furthermore, \cite{goeman2019simultaneous} showed that the confidence bound $d_{\mathrm{Simes}}$ is between the number of Hommel-FWER discoveries and the number of BH-FDR discoveries, i.e.
\begin{eqnarray}\label{eq-dSimesdBH}
    | D_{\mathrm{Simes}} | \leq d_{\mathrm{Simes}} \leq | D_{\mathrm{BH}} |.
\end{eqnarray}
Finally, note that the BH discoveries do
not provide the true discovery guarantee, i.e.,
\begin{eqnarray}\label{eq-BHnoguarantee}
\mathbb{P}(| D_{\mathrm{BH}} \cap S| \leq |I_1 \cap S| \mathrm{\,\,for\,\,all\,\,}S  )
\end{eqnarray}
may be less than $1-\alpha$.

\subsection{Shortcuts for Storey-Simes Local Tests} \label{app:shortcuts-adaptive-simes}

The BH procedure guarantees that expected proportion of inliers among the discoveries is bounded by $\alpha\pi_0$, where $\pi_0=|I_0|/n$ is the unknown proportion of inliers in the test sample. When $\pi_0$ is expected to be not close to 1, $\alpha\pi_0$ falls below the target level $\alpha$, making the BH procedure slightly too conservative. This motivates adjusting the level into $\alpha/\hat{\pi}_0$ with $\hat{\pi}_0$ an estimate of $\pi_0$, resulting in a $\pi_0$-adaptive version of the BH algorithm.

The idea of the adaptive BH procedure extends to the Simes test, leading to a $\pi_0$-adaptive version of Simes' test for $H_S$ that incorporates the proportion of true null hypotheses in $S$. The test is formulated as:
\begin{equation}\label{eq-ASimeslocaltest}
    \phi_S^\mathrm{Storey-Simes} = \mathds{1}\left\{\min_{k\in \{1,\ldots,|S|\}} \{|S| p_{(k:S)}/k\}\leq \alpha/ \hat\pi_0^S \right\},
\end{equation}
where 
\begin{eqnarray}\label{eq:storey-estimator}
\hat{\pi}^{S}_0 &=& \frac{1+\sum_{j\in S} \mathds{1}\{p_j > \lambda\} }{|I|(1-\lambda)}   
\end{eqnarray} 
and $\lambda = h/(|S|+1)$ for any pre-specified integer $h$. This estimator is known as the Schweder-Spjøtvoll or Storey's estimator \citep{schweder1982plots, storey2002direct, storey2004strong}.
The test $\phi_I^\mathrm{Storey-Simes}$ was considered in \cite{bogomolov2023testing} and \cite{heller2023simultaneous}.

The closed testing with adaptive Simes' test ensures the true discovery guarantee by observing that the adaptive Simes' test rejects $H_{I_0}$ if and only if the adaptive Benjamini-Hochberg procedure applied to conformal $p$-values $(p_j, j\in I_0)$ rejects at least one hypothesis. This rejection occurs with a probability bounded by $\alpha$, as shown in Corollary 2.5 and Theorem 2.6 in \cite{bates2021testing}, and Theorem 3.1 in \cite{mary2022semisup}.

The shortcut given in Algorithm~\ref{alg-ASimesCT} allows calculation of $d_\mathrm{Storey-Simes}(S)$ in cubic time.
Algorithm~\ref{alg-ASimesCT} is based on a worst-case argument. For any $L \subseteq [n]$ with $|L|=l$,  the Storey-Simes test rejects $H_L$ if
$$  \displaystyle  \min_{  j \in L } \Big(  l \frac{p_{(j:L)}}{j} \Big)  \leq \alpha \Big(\frac{l (1-\lambda)}{1 + \sum_{j \in [L]} \mathds{1}\{p_{j} > \lambda\}} \Big).$$
Among all subsets $L$ of size $l$, the least favorable (hardest to reject) choice is the set corresponding to the largest $l$ $p$-values $p_{(n-l+1)},\ldots,p_{(n)}$. 
Consequently, for a fixed $l$, it suffices to evaluate the Storey–Simes test on this worst-case subset. Algorithm~\ref{alg-ASimesCT} exploits this fact to implement closed testing exactly, while avoiding explicit enumeration of all subsets.

\begin{algorithm}[H]\label{alg-ASimesCT}
\SetKwInOut{KwInit}{Initialize}
\SetKwInOut{Input}{Input}
\SetKwInOut{Output}{Output}
\Input{$p$-values $p_1,\ldots,p_n$;  significance level $\alpha \in (0,1)$; tuning parameter $\lambda \in (0,1)$.}

Sort the $p$-values $p_{(1)} \leq \ldots \leq p_{(n)}$;

$h \gets n$;

\While{$h \ge 1$ \ \textbf{and}\ 
$  \displaystyle  \min_{  j \in [h] } \Big(  h \frac{p_{(n-h+j)}}{j} \Big)  \leq \alpha \Big(\frac{h (1-\lambda)}{1 + \sum_{j \in [h]} \mathds{1}\{p_{(n-h+j)} > \lambda\}} \Big)$}{
    $h \gets h - 1$;
}

Compute $d_\mathrm{Storey-Simes} = d_\mathrm{Storey-Simes}([n]) = n-h$;

\textbf{Query  \(S\subset [n]\):} $d_\mathrm{Storey-Simes}(S) \gets 0$;

Sort the $p$-values in $S$, $p_{(1:S)}\leq \ldots \leq p_{(|S|:S)}$, and in $S^c$, $p_{(1:S^c)}\leq \ldots \leq p_{(n-|S|:S^c)}$;

\For{each $i\in\{1, \ldots, |S|\}$}{
\For{each $j\in\{0, \ldots, \max(h-|S|+i-1,0)\}$}{

  $l \gets |S|-i+j+1$;

Merge $(p_{(i:S)},\ldots,p_{(|S|:S)})$ and $(p_{(n-|S|-j+1:S^c)},\ldots,p_{(n-|S|:S^c)})$ and obtain the sorted values $q_{(1)} \leq \ldots \leq q_{(l)}$;

  \If{
  $\displaystyle  \min_{j \in [l]} \Big(  l \frac{q_{(j)}}{j} \Big)  > \alpha \Big(\frac{l (1-\lambda)}{1 + \sum_{j \in [l]} \mathds{1}\{q_j > \lambda\}} \Big)$,
  }{\Return{$d_\mathrm{Storey-Simes}(S)$}.}
  
}

$d_\mathrm{Storey-Simes}(S) = d_\mathrm{Storey-Simes}(S)+1$;
}

\Output{$(1-\alpha)$-confidence lower bound $d_{\mathrm{Storey-Simes}}(S)$ for the number of true discoveries in $S$, for $S=[n]$ and any queried subset $S \subset [n]$.}
\caption{Shortcut for closed testing with adaptive Simes local tests using the Storey estimator.}
\end{algorithm}

\subsection{Shortcut for Shiraishi Local Tests} \label{app:shortcuts-G}

The shortcut provided in Algorithm~\ref{alg:ShiraishiCT} is exact and enables the calculation of $d_\mathrm{Shiraishi}(S)$ in cubic time for a given monotone function $g$. The same shortcut applies when $g$ is replaced by a monotone estimate $\hat{g}$. In the following, we assume that $g$ is non-decreasing.

Let $a^l = (a^l_1,\ldots,a^l_{m+l})$ be a vector where the $r$th element is $a^l_r = \mathbb{E}[g(U^{(r)}_{m+l})]$. Since $g$ is non-decreasing, the elements of the vector are ordered such that $a_1^l \leq \ldots \leq a^l_{m+l}$. The score vector $a^l$ can be estimated with the desired accuracy using Monte Carlo simulation. 

Alternatively, one may use the  approximation $b^l_r = g(r/(m+l+1))$, which is computationally efficient.

The shortcut given in Algorithm~\ref{alg:ShiraishiCT} allows calculation of $d_\mathrm{Shiraishi}(S)$ in cubic time when $a^l$ can be evaluated in linear time, as is the case when using the approximation $b^l$.
In a similar spirit to Algorithm~\ref{alg-ASimesCT}, Algorithm~\ref{alg:ShiraishiCT} is based on a worst-case argument and is therefore exact. For any $L \subseteq [n]$ with $|L|=l$,  the Shiraishi test statistic for testing $H_L$ is
$$
T_L^g =  \sum_{j \in L} a^{l}_{R_{m+j}}.
$$
Among all subsets $L$ of size $l$, the least favorable (hardest to reject) choice is the subset corresponding to the smallest
$l$ scores among $Y_1,\ldots,Y_n$. 

Shiraishi's test for $H_L$ is defined as $$\phi_{L}^{\mathrm{Shiraishi}} = \mathds{1}\{ T_L^g > c_\alpha^g(m,l) \}.$$
Using the asymptotic approximation, the critical value $c^g_\alpha(m,l)$ corresponds to the $(1-\alpha)$-quantile of the normal distribution with the following mean and variance:
\begin{equation*}
    l \mu_{m+l} = \frac{l}{m+l}\sum_{k\in[m+l]} a^l_k,\quad
    \frac{m l}{m+l} \sigma^2_{m+l} = \frac{ml}{(m+l)(m+l-1)}\sum_{k\in[m+l]} \Big(a^l_k-\frac{1}{m+l}\sum_{k\in[m+l]} a^l_k\Big)^2.
\end{equation*}

\begin{algorithm}[!htb]
\SetKwInOut{Input}{Input}
\Input{Calibration scores $(X_1,\ldots,X_m)$; test scores $(Y_1,\ldots,Y_n)$; a non-decreasing density function $g$;  significance level $\alpha \in (0,1)$.}

Sort the calibration score $X_{(1)}\leq \ldots \leq X_{(m)}$ and the test scores $Y_{(1)}\leq \ldots\leq Y_{(n)}$;

Compute the ranks $R_1,\ldots,R_N$ of the pooled sample  $(X_{(1)},\ldots,X_{(m)},Y_{(1)},\ldots,Y_{(n)})$

$h \gets n$;

\While{$h \ge 1$ \ \textbf{and}\ 
$\sum_{j=m+1}^{m+h} \mathbb{E}[g(U_{m+h}^{(R_j)})] > c^g_\alpha(m,h)$}{
    $h \gets h - 1$;
}

Compute $d_\mathrm{Shiraishi} = d_\mathrm{Shiraishi}([n]) = n-h$;

\textbf{Query  \(S\subset [n]\):} $d_\mathrm{Shiraishi}(S) \gets 0$;

Sort the test scores in $S$, $Y_{(1:S)}\leq \ldots \leq Y_{(|S|:S)}$, and in $S^c$, $Y_{(1:S^c)}\leq \ldots \leq Y_{(n-|S|:S^c)}$;

\For{$i = 1, \ldots, |S|$}{
\For{$j =  0,\ldots, \max(h-|S|+i-1,0)$}{

  $l \gets |S|-i+j+1$;

  Merge $(Y_{(1:S)},\ldots,Y_{(|S|-i+1:S)})$ and $(Y_{(1:S^c)}, \ldots, Y_{(j:S^c)})$ to obtain the sorted values $W_{(1)}\leq \ldots \leq W_{(l)}$;
  
  Compute the ranks $R_1,\ldots,R_{m+l}$ of  $(X_{(1)},\ldots,X_{(m)},W_{(1)},\ldots,W_{(l)})$;

  \If{
  $\sum_{j=m+1}^{m+l} \mathbb{E}[g(U_{m+l}^{(R_j)})] \leq c^g_\alpha(m,l)$,
  }{\Return{$d_\mathrm{Shiraishi}(S)$.}}

$d_\mathrm{Shiraishi}(S) = d_\mathrm{Shiraishi}(S)+1$;
  
}
}

\SetKwInOut{Output}{Output}

\Output{A $(1-\alpha)$-lower bound $d_\mathrm{Shiraishi}(S)$ for the number of outliers in $S$, for $S=[n]$ and any queried subset $S \subset [n]$.}

\caption{Shortcut for closed testing with Shiraishi local tests.} \label{alg:ShiraishiCT}
\end{algorithm}

\subsection{Shortcut for Monotone, Symmetric, and Separable Tests} \label{app:shortcuts-tian}

\cite{tian2023Large-scale} introduced an exact shortcut for the general case of closed testing with local tests satisfying three properties: \emph{monotonicity}, \emph{symmetry} and \emph{separability}. These properties are met by classic Mann-Whitney test and Fisher's combination test. For precise definitions of these properties, please refer to Appendix B in~\cite{tian2023Large-scale}.

Algorithm 1 in \cite{tian2023Large-scale} enables a linear-time computation of the $(1-\alpha)$-confidence lower bound $d(S)$ for any $S\subseteq[n]$, after a preparatory step involving the computation and sorting of the test statistics.

\subsection{Approximate Shortcuts Based on DKW Concentration}\label{app:bounds_DKW}

\cite{gazin2024transductive}  investigated the joint distribution of conformal $p$-values and derived a finite-sample Dvoretzky–Kiefer–Wolfowitz (DKW) type concentration inequality for their empirical distribution function. Although we have not done so in this paper, the approach of \citet{gazin2024transductive} could also be integrated within ACODE (specifically Algorithm~\ref*{alg:acode_oracle}) to complement the local testing procedures currently implemented.

For nested sets $R_i = \{j \in [n]:p_j \leq p_{(i)} \}$, $i \in [n]$, \cite{gazin2024transductive} proposed a simultaneous false discovery proportion (FDP) bound 
\begin{eqnarray*}
\overline{\mathrm{FDP}}_{\mathrm{DKW}}(R_i) = \frac{n_0(I_m(t) + \lambda_{\alpha,m,n_0})}{\max(1,|R_i|)}
\end{eqnarray*}
where $n_0= |I_0|$, $I_m(t) = \lfloor (m+1)t \rfloor / (m+1)$,
\begin{eqnarray*}
    \lambda_{\alpha,m,n} &=& \Psi^{(r)}(1),\\
    \Psi(x) &=& \left(\frac{\log(1/\alpha) + \log(1+\sqrt{2\pi} \frac{2 \tau_{m,n} x}{(m+n)^{1/2}} )}{2 \tau_{m,n}}\right)^{1/2} \wedge 1,
\end{eqnarray*}
$\Psi^{(r)}$ denotes the function $\Psi$ iterated $r$ times and $\tau_{m,n} = m n / (m+n)$. In practice the unknown $n_0$ should be replaced $n$ or by a suitable estimator $\hat n_0$; see 
\citet[Corollary 4.1]{gazin2024transductive}.
Since the finite-sample joint distribution of the conformal $p$-values can be simulated exactly, the constant $\lambda_{\alpha,m,n_0}$ can be calibrated numerically; see 
\citet[Remark 2.6]{gazin2024transductive}.

The above FDP bound is explicit and elegant but, as pointed out by \citet[Section B]{gazin2024transductive}, it can be conservative in some cases. \cite{gazin2024transductive} proposed an improved and more general numerical approach based on the notion of templates introduced by \cite{blanchard2020post}. Here we show that an improved version, extending to all sets, can also be obtained via \emph{interpolation} \citep{blanchard2020post,goeman2021onlyclosed}.
Specifically, we first translate the false discovery bound into the corresponding lower bound on the number of outliers in $S$: 
\begin{eqnarray}\label{eq:DKW}
    d_{\mathrm{DKW}}(S) &=& \left\{ 
    \begin{array}{ll}
      |R_i| - \bar{V}(R_i)  & \mathrm{if\,\,}S = R_i\mathrm{\,\,for\,\,some\,\,}1\leq i\leq n, \\
       0  & \mathrm{otherwise,}
    \end{array}
 \right.
\end{eqnarray}
where $$\bar{V}(R_i) = \min(\lfloor n_0(I_m(t) + \lambda_{\alpha,m,n_0})\rfloor , |R_i|)$$ is the upper bound on the number of false discoveries in $R_i$.
Here, the original upper bound on the number of false discoveries in $R_i$ is rounded down to an integer, and the expression is truncated at zero to ensure
$\bar{V}(R_i) \leq |R_i|$. 

Then, by Lemma C.2 in \cite{li2022simultaneous}, the interpolated version of 
$d_{\mathrm{DKW}}(S)$ is
\begin{eqnarray}\label{eq:DKW_improved}
    \bar{d}_{\mathrm{DKW}}(S) &=& \max_{i \in [n]} \{ |S \cap R_i| - \bar{V}(R_i) \} \vee 0.
\end{eqnarray}
This bound dominates the original one in the sense that
$$\bar{d}_{\mathrm{DKW}}(S) \geq d_{\mathrm{DKW}}(S),\qquad \forall\,S \subseteq [n].$$
The improvement is typically strict, since $d_{\mathrm{DKW}}(R_i)$ is not necessarily monotone in $R_i$; in particular, it may occur that $d_{\mathrm{DKW}}(R_i) > d_{\mathrm{DKW}}(R_k)$ for $R_i \subset R_k$. In contrast, the interpolated bound $\bar{d}_{\mathrm{DKW}}(S)$ enforces monotonicity and raises
$d_{\mathrm{DKW}}(R_k)$ to at least $d_{\mathrm{DKW}}(R_i)$, implying that $\bar{d}_{\mathrm{DKW}}(R_k) > d_{\mathrm{DKW}}(R_k)$. 

Moreover, by Lemma C.3 in \cite{li2022simultaneous}, the bound $\bar{d}_{\mathrm{DKW}}(S)$ in Equation (\ref{eq:DKW_improved}) is \emph{coherent}, in the sense that it cannot be further improved by interpolation \citep{goeman2021onlyclosed}.  
One can also apply Theorem 1 in \cite{goeman2021onlyclosed}  to connect the simultaneous bounds $\bar{d}_{\mathrm{DKW}}(S)$ to closed testing. The closed testing procedure implied by this theorem yields $d^{\mathrm{CT}}_{\mathrm{DKW}}(S)$ (see Appendix C.4 in \cite{li2022simultaneous} for an example of such a construction), which is either equal to or strictly larger than $\bar{d}_{\mathrm{DKW}}(S)$. Consequently, the original bound $d_{\mathrm{DKW}}(S)$ in Equation (\ref{eq:DKW}) can be viewed as an  approximate shortcut to the closed testing procedure.

\subsection{Empirical Comparison of Computation Times}\label{app:timings}

The computational complexities summarized in Table~\ref{tab:comp-costs} describe the asymptotic scaling of the shortcut procedures, but do not by themselves quantify their practical cost. To address this point, we complement Table~\ref{tab:comp-costs} with an empirical comparison of computation times for the proposed lower bounds $d$ across the local tests considered in the paper.

We compare Simes, Fisher, and WMW (which have linearithmic complexity and are implemented in \textsf{C++}), as well as Storey--Simes and Shiraishi (which have quadratic complexity and are implemented in \textsf{R}). To ensure a fair comparison, we include a common preparatory step in the timing for \emph{every} method: we sort the $m$ calibration scores and $n$ test scores, compute the ranks of the pooled scores, and form the corresponding conformal $p$-values (which are therefore already ordered). When $m$ is of the same order as $n$, this preprocessing requires $\mathcal{O}(N\log N)$ time with $N=m+n$, and it is included in the reported runtime of each method.

\paragraph{Simulation setup.}
We generate balanced calibration and test samples with $m=n$. Inliers are drawn from a standard normal distribution, and outliers are drawn from a location-shifted distribution with shift $\mu=\sqrt{2\log n}$. The outlier proportion is set to $\theta=0.5$. We vary the sample size from $n=5000$ to $n=25000$.

\paragraph{Reported metric.}
For each method and each $n$, we report the median computation time (in milliseconds) over 3 repetitions.

\paragraph{Results.}
Figure~\ref{fig:timings} displays the results. All methods exhibit the expected asymptotic behavior: the \textsf{C++} implementations of Simes, Fisher, and WMW scale close to the initial sorting cost, whereas Storey--Simes and Shiraishi show the anticipated super-linear growth. The Shiraishi method in our implementation relies on the linear-time approximation $g(r/(N+1))$ for the score vector; using $\mathbb{E}\!\left[g(U_N^{(r)})\right]$ instead would be more computationally demanding. Nevertheless, for $n=5000$ and $n=15000$, computing $d$ takes on the order of 1 second and 10 seconds, respectively, which appears reasonable for the applications considered in this paper.

\begin{figure}[!htb]
    \centering
    \includegraphics[scale=.7]{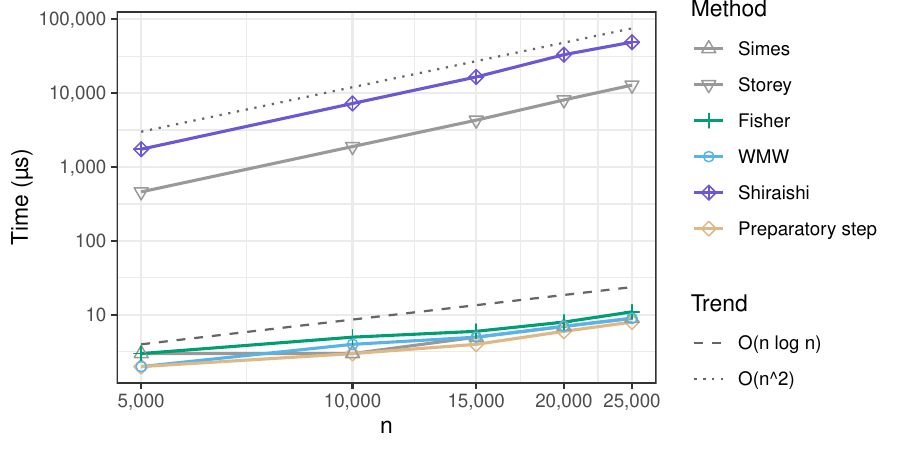}
    \caption{Median computation time (milliseconds) for computing the overall lower bound $d$ across local tests, including a common preprocessing step (sorting pooled scores, computing ranks, and forming ordered conformal $p$-values). Times are medians over 10 repetitions.}
    \label{fig:timings}
\end{figure}

\section{Optimal Score Function for Outlier Detection}\label{app:optimal_scoring}

In this section, we show that the likelihood ratio function $r^{opt}$ is the optimal scoring function for constructing conformity scores that are \emph{sufficient} for the true proportion of outliers in the test sample. In particular, sufficiency of the conformity scores implies that no power is lost for testing the global null hypothesis of no outliers when inference is based on the conformity scores. Specifically, we show that the optimal test based on the original observations $Z_{m+1},\ldots,Z_{m+m}$ coincides with a two-step procedure that (i) computes the optimal conformity scores $Y_1^{opt}=r^{opt}(Z_{m+1}),\ldots,Y^{opt}_n=r^{opt}(Z_{m+n})$ and (ii) applies the optimal likelihood ratio test to these scores. 

We assume that the test observations $Z_{m+1}, \ldots, Z_{m+n} \in \mathbb{R}^d$ are i.i.d. from the mixture model:
\begin{equation}\label{eq:model-mixture}
   Z_{m+1}, \ldots, Z_{m+n} \overset{\text{i.i.d.}}{\sim} (1-\pi)P_0 + \pi \bar P_1,
\end{equation}
where $\pi \in [0,1]$ denotes the expected outlier proportion, and $P_0$ is the inlier distribution generating the calibration sample $Z_1, \ldots, Z_m$. 
The mixture model (\ref{eq:model-mixture}) can be reconciled with the setup in Section~\ref{sec:setup} by expressing the outlier component $\bar P_1$ as the average of the individual outlier distributions $P_j$ in Equation~\eqref{eq:model}: $\bar P_1 := n^{-1}\sum_{j \in [n]} P_j$. 
With this choice, the model in~\eqref{eq:model-mixture} closely resembles that in~\eqref{eq:model}, although the two are not exactly equivalent. 
We further assume that $P_0$ and each $P_j$ are continuous distributions with densities $p_0$ and $p_j$, for all $j \in [n]$. Consequently, the outlier component $\bar P_1$ is also a continuous distribution with density $\bar p_1= n^{-1} \sum_{j \in [n]} p_j$. 


Under this setup, the likelihood for $\pi$ of the test sample $Z_{m+1},\ldots,Z_{m+n}$ is
\begin{align*}
       L(\pi; Z_{m+1},\ldots,Z_{m+n}) &= \nonumber
       \prod_{j\in[n]} \left[(1- \pi) p_0(Z_{m+j}) + \pi \bar p_1(Z_{m+j})\right]
       =\\&= \nonumber
       \prod_{j\in[n]}\left[(1-\pi) + \pi r^{opt}(Z_{m+j}) \right]\cdot \prod_{j\in[n]} p_0(Z_{m+j}),
    \end{align*}
    where
\begin{align*}
r^{opt}(z) = \frac{\bar{p}_1(z)}{p_0(z)}.
    \end{align*}
This factorization shows that the likelihood depends on $\pi$ only through  $r^{opt}(Z_{m+1}),\ldots, r^{opt}(Z_{m+n})$, implying by  that these conformity scores are sufficient for $\pi$. 

To test the null hypothesis
$H: \pi = 0$ against the alternative  
$K: \pi = \pi_1$ for some $\pi_1 \in (0,1]$, the Neyman–Pearson lemma yields the likelihood ratio test (LRT) with test statistic
\begin{align}\label{eq:optimal_LRT}
       \Lambda &= 
       \prod_{j\in[n]} \left[  (1- \pi_1) + \pi_1 r^{opt}(Z_{m+j})\right].
\end{align}
Since the conformity scores $Y_1^{opt}=r^{opt}(Z_{m+1}),\ldots, Y_n^{opt}=r^{opt}(Z_{m+n})$ are sufficient for $\pi$, the LRT based on the original observations coincides with the LRT based on these scores. This justifies the use of conformity scores derived from 
$r^{opt}$ as a dimension-reducing transformation that is optimal for testing the global null hypothesis of no outliers.

This suggests that inference should be based on the optimal calibration scores $X^{opt}_1 = r^{opt}(Z_1), \ldots, X^{opt}_m = r^{opt}(Z_m)$ and the optimal test scores $Y_1^{opt}= r^{opt}(Z_{m+1}), \ldots, Y_n^{opt}=r^{opt}(Z_{m+n})$. Local testing methods based on conformal $p$-values or ranks are invariant under monotone transformations of the conformity scores $(X^{opt}_1,\ldots,X^{opt}_m,Y^{opt}_1,\ldots,Y^{opt}_n)$.
Consequently, using any scoring function of the form
$r^* = \Psi \circ  r^{opt}$ with $\Psi$ increasing and continuous (possibly depending on unknown parameters), yields 
the same inference \citep{marandon2022machine}. This holds because the ranks $R^{opt}_1,\ldots,R^{opt}_N$ of the optimal conformity scores $(X^{opt}_1,\ldots,X^{opt}_m,Y^{opt}_1,\ldots,Y^{opt}_n)$ coincide with the ranks of the conformity scores $(X^{*}_1,\ldots,X^{*}_m,Y^{*}_1,\ldots,Y^{*}_n)$
constructed from any scoring function of the form $r^*$.

\citet{marandon2022machine} derived the scoring function $(1-\pi)r^{opt}$, which is of the form $r^*$, for outlier identification with marginal false discovery rate (mFDR) control. They showed that, among all procedures that reject hypotheses by thresholding conformity scores while controlling mFDR, the score function $(1-\pi)r^{opt}$ maximizes the true discovery rate.

\citet{marandon2022machine} considered applying the Benjamini–Hochberg procedure to the oracle conformal $p$-values $p^{opt}_1,\ldots,p^{opt}_n$ obtained from $(X^{opt}_1,\ldots,X^{opt}_m,Y^{opt}_1,\ldots,Y^{opt}_n)$, or equivalently from conformity scores constructed using any scoring function of the form $r^*$. 
They propose to estimate $r^*$ via binary classification, with the aim of approximating the oracle conformal $p$-values $p^{opt}_1,\ldots,p^{opt}_n$, or equivalently, the oracle ranks $R^{opt}_1,\ldots,R^{opt}_N$. This supports the use of binary classification for computing conformity scores when using testing methods based on conformal $p$-values or ranks. 

In this section, we consider optimality for outlier detection, that is, for the global test of the null hypothesis of no outliers. However, optimality for outlier enumeration does not, in general, follow from optimality of the individual local tests \citep{heller2023optimal}, although admissibility of the local tests is necessary \citep{goeman2021onlyclosed}.

\section{Proofs}\label{app:proofs}

\subsection{Theorem~\ref*{thm:asymp_distr_LMPI-exch}}

\begin{proof}[Proof of Theorem~\ref*{thm:asymp_distr_LMPI-exch}]
This result follows from a general formulation of the finite-population central limit theorem \citep{hajek1960limiting, li2017general}. 
The Shiraishi statistic can be written as 
  $$T^{\text{g}}  = \sum_{r \in [N]} \E{g(U_{N}^{(r)})} C_r$$
where $C_r = 1$ if the unit with rank $r$ is assigned to the test sample, and $C_r = 0$ otherwise. Exchangeability under $H_0^*$ implies that $(C_1,\ldots,C_{N})$ are uniformly distributed over all vectors with exactly $n$ ones. Thus, the null distribution of $T^{\text{g}}$ is induced by permutations of $(C_1,\ldots,C_{N})$, which is equivalent to sampling $n$ units without replacement from the finite population $\{ \E{g(U_{  N}^{(1)})}, \ldots, \E{g(U_{  N}^{(N)})} \}$. Consequently, $T^{\text{g}}/n$ is the sample mean of this finite population, and the result follows by applying Theorem~1 of \cite{ li2017general}.
\end{proof}

\subsection{Theorem~\ref{thm:asymp_equivalence} and Corollary~\ref{cor:asymp_equivalence-power}}

\begin{proof}[Proof of Theorem~\ref{thm:asymp_equivalence}]
Under $H_0:\theta=0$, we condition on the pooled sample and define
\[
\mathbb P_* := \mathbb P_0(\,\cdot \mid \{W_1,\ldots,W_N\}) .
\]
Conditioning on the pooled sample fixes $\hat g_N$.  

Let $R\subset[N]$ denote the set of ranks assigned to the test sample, with $|R|=n$. Under $\mathbb P_*$, $R$ is uniformly distributed over all subsets of $[N]$ of cardinality $n$.

Define the score errors 
$$\delta_r=\hat g_N\!\left(\frac{r}{N+1}\right)
-
g\!\left(\frac{r}{N+1}\right), \qquad r \in [N]$$
and their finite-population mean
\[
\bar\delta = \frac{1}{N}\sum_{r=1}^N \delta_r .
\]
The difference between the centered statistics $T^{\hat{\text{g}}}$ and $T^{\text{g}}$ is
\[
(T^{\hat{\text{g}}} - n\hat{\mu}_N) - (T^{\text{g}} - n\mu_N)
=
\sum_{r=1}^N \delta_r\,\mathds 1\{r\in R\} - n\bar\delta
=
\sum_{r=1}^N (\delta_r-\bar\delta)\,\mathds 1\{r\in R\}.
\]
Hence, conditional on the pooled sample, $(T^{\hat{\text{g}}} - n\hat{\mu}_N) - (T^{\text{g}} - n\mu_N)$ is the sum of $n$ draws  obtained by simple random sampling without replacement from the finite population
$\{\delta_r-\bar\delta:r=1,\ldots,N\}$.
Since $\mathbb E_*[\mathds 1\{r\in R\}]=n/N$ for all $r$, it follows that
\[
\mathbb E_*[(T^{\hat{\text{g}}} - n\hat{\mu}_N) - (T^{\text{g}} - n\mu_N)]=0 .
\]
For simple random sampling without replacement, the finite-population variance
formula gives
\[
\mathbb{V}\mathrm{ar}_*[(T^{\hat{\text{g}}} - n\hat{\mu}_N) - (T^{\text{g}} - n\mu_N)]
=
\frac{n(N-n)}{N-1}\cdot
\frac{1}{N}\sum_{r=1}^N(\delta_r-\bar\delta)^2 .
\]
Since $n/N\to\lambda\in(0,1)$, the multiplicative term is of order $N$, and therefore there
exists a constant $C>0$ such that, for all sufficiently large $N$,
\[
\mathbb{V}\mathrm{ar}_*[(T^{\hat{\text{g}}} - n\hat{\mu}_N) - (T^{\text{g}} - n\mu_N)]
\le
C\,N\cdot\frac{1}{N}\sum_{r=1}^N\delta_r^2 .
\]
By Chebyshev’s inequality, for any $\varepsilon>0$,
\[
\mathbb P_*\!\left(\bigl|(T^{\hat{\text{g}}} - n\hat{\mu}_N) - (T^{\text{g}} - n\mu_N)\bigr|>\varepsilon\sqrt N\right)
\le
\frac{C}{\varepsilon^2}\cdot\frac{1}{N}\sum_{r=1}^N\delta_r^2 .
\]
By (\ref{eq-L2consistency}), the right-hand side converges to zero in probability.
Hence,
\[
\mathbb P_*\!\left(\bigl|(T^{\hat{\text{g}}} - n\hat{\mu}_N) - (T^{\text{g}} - n\mu_N)\bigr|>\varepsilon\sqrt N\right)
\xrightarrow{\mathbb P} 0,
\]
which shows that $$\frac{(T^{\hat{\text{g}}} - n\hat{\mu}_N) - (T^{\text{g}} - n\mu_N)}{\sqrt N}\xrightarrow{\mathbb P_*}0.$$
Finally, taking expectations with respect to the pooled sample yields
$$\frac{(T^{\hat{\text{g}}} - n\hat{\mu}_N) - (T^{\text{g}} - n\mu_N)}{\sqrt N}\xrightarrow{\mathbb P_0}0.$$
\end{proof}

\begin{proof}[Proof of Corollary~\ref{cor:asymp_equivalence-power}]
Consider the sequence of probability measures $\{\mathbb{P}^N_h\}$ corresponding to the local alternatives
$\theta_N = h/\sqrt{N}$ for some $h>0$.
By \citet[Theorem~1]{shiraishi1985local}, the sequence
$\{\mathbb P_h^{N}\}$ is contiguous with respect to the sequence
$\{\mathbb P_0^{N}\}$.
Therefore, by Le Cam’s first lemma
\citep[Lemma~6.4]{van2000asymptotic},
\[
\frac{(T^{\hat{\text{g}}} - n\hat{\mu}_N) - (T^{\text{g}} - n\mu_N)}{\sqrt N}\xrightarrow{\mathbb P_0}0
\quad\Longrightarrow\quad
\frac{(T^{\hat{\text{g}}} - n\hat{\mu}_N) - (T^{\text{g}} - n\mu_N)}{\sqrt N}\xrightarrow{\mathbb P_h}0.
\]
Since the oracle LMPR statistic satisfies
\[
\frac{T^{\text{g}} - n\mu_N}{\sqrt N}
\xrightarrow{d}
N \!\bigl(h I_g,\, I_g\bigr)
\quad\text{under } \mathbb P_h, \qquad I_g = \lambda(1-\lambda)\mathbb{V}\mathrm{ar}(g(U)),
\]
it follows by Slutsky’s theorem that
\[
\frac{T^{\hat{\text{g}}} - n\hat{\mu}_N}{\sqrt N}
=
\frac{T^{\text{g}} - n\mu_N}{\sqrt N}
+
\frac{(T^{\hat{\text{g}}} - n\hat{\mu}_N) - (T^{\text{g}} - n\mu_N)}{\sqrt N}
\xrightarrow{d}
 N\!\bigl(h I_g,\, I_g\bigr)
\quad\text{under } \mathbb P_h .
\]
Consequently, the adaptive rank statistic and the oracle LMPR statistic
have the same limiting distribution under local alternatives and therefore
the same local asymptotic power function.

\end{proof}

\section{Additional Empirical Results} \label{app:numerical}

\subsection{Numerical Experiments with Synthetic Data} \label{app:numerical-synthetic}

Figure~\ref{fig:exp-synthetic-1-tuning} illustrates the effect of the proportion of calibration data allocated to tuning on ACODE’s performance. Specifically, we repeat the experiments summarized in Figure~\ref*{fig:exp-synthetic-1} while varying the fraction of calibration points used for tuning. The total calibration sample size is fixed at 1000, of which 250 points (25\%) were used for tuning in Figure~\ref*{fig:exp-synthetic-1}; here we vary this proportion from 0.01 to 0.99. The results show that ACODE's performance remains quite stable over a wide range of tuning fractions, but degrades when too few points are used for tuning (rendering tuning ineffective) or when too many are allocated to tuning (leaving insufficient data for computing significance tests). This motivates our choice of 25\% as a reasonable default.

\begin{figure}[!htb]
\centering 
\includegraphics[width=\linewidth]{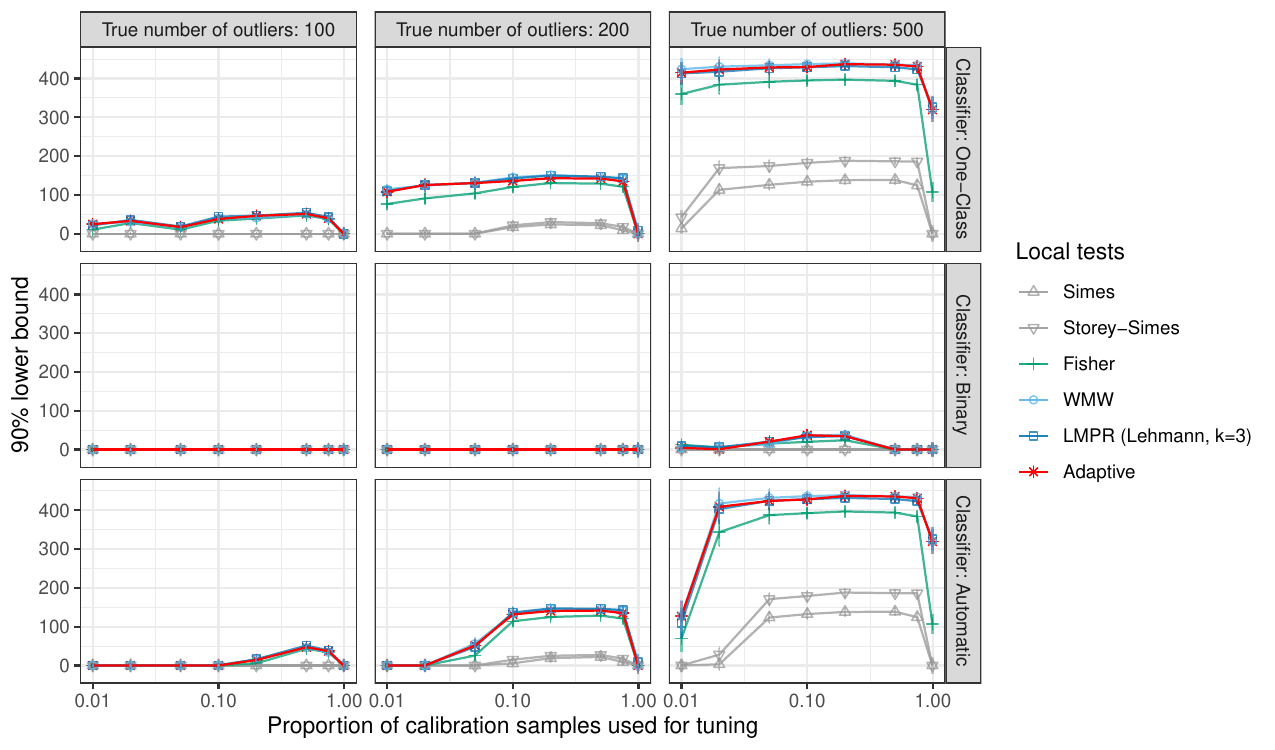}
\caption{Median values for a 90\% lower confidence bound on the number of outliers in a test set, computed by ACODE on synthetic data based on different classifiers and local testing procedures. The results are shown as a function of the proportion of calibration points (1000 observations in total) used for tuning (ref.~Section~\ref*{sec:tuning_new}), for different true numbers of outliers within a test set of size 1000. Other details are as in Figure~\ref*{fig:exp-synthetic-1}, where the proportion of calibration data points used for tuning is fixed at 0.25.
}

\label{fig:exp-synthetic-1-tuning}
\end{figure}

Figure~\ref{fig:exp-synthetic-1-sel-q0.9} provides additional details on the experiments of Figure~\ref*{fig:exp-synthetic-1-sel-q0.5}, reporting the empirical 90-th quantiles for the 90\% lower confidence bounds, instead of the median. 

\begin{figure}[!htb]
\centering
\includegraphics[width=0.9\linewidth]{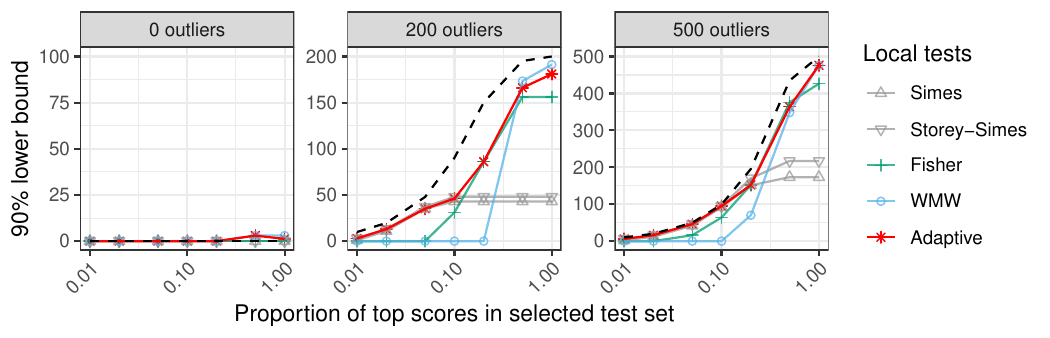}
\caption{Empirical 90-th quantile for a 90\% lower confidence bound on the number of outliers within an adaptively selected subset of 1000 test points, in experiments similar to those of Figure~\ref{fig:exp-synthetic-1}. Other details are as in Figure~\ref*{fig:exp-synthetic-1-sel-q0.5}.
}
\label{fig:exp-synthetic-1-sel-q0.9}
\end{figure}

To assess the difficulty of individual outlier detection in this synthetic setting, we additionally compare ACODE with the Benjamini-Hochberg (BH) \citep{benjamini1995controlling} procedure applied to individual conformal $p$-values at nominal FDR level 10\%, following the approach of \citet{bates2021testing}. As shown in Figure~\ref{fig:exp-synthetic-1-bh}, individual discovery is feasible in this setting, due to the relatively strong signals, but consistently less powerful than collective outlier enumeration with ACODE.

\begin{figure}[!htb]
\centering
\includegraphics[width=\linewidth]{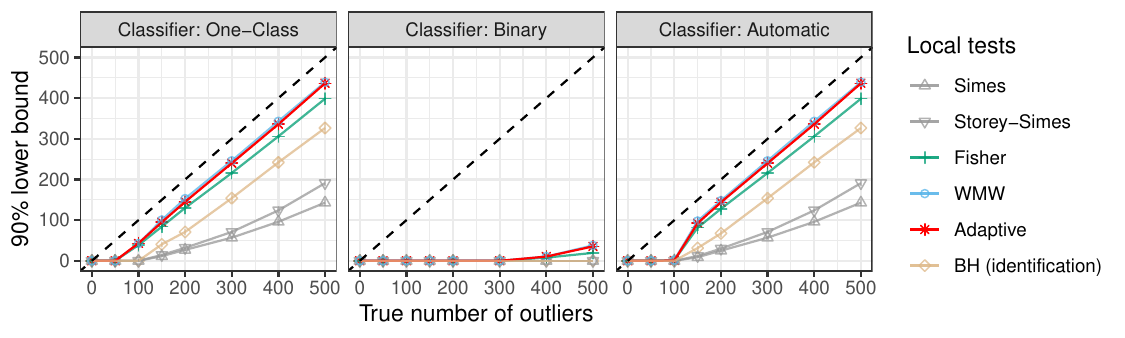}
\caption{Median values for a 90\% lower confidence bound on the number of outliers in a test set, computed by ACODE on synthetic data based on different classifiers and local testing procedures, as in Figure~\ref{fig:exp-synthetic-1}. Here, we introduce an additional benchmark: the number of discoveries made by BH procedure for individual outlier identification with FDR control, applied at nominal FDR level 10\%. This shows our ACODE method has higher power, except when applied using Simes' test as local testing procedure.
}
\label{fig:exp-synthetic-1-bh}
\end{figure}

\FloatBarrier

Figures~\ref{fig:exp-synthetic-2} and~\ref{fig:exp-synthetic-2-sel-q0.5} describe findings similar to those in Figures~\ref*{fig:exp-synthetic-1} and~\ref*{fig:exp-synthetic-1-sel-q0.5}, respectively. The distinction is that now the data are simulated from a binomial model borrowed from \citet{liang2022integrative}, for which binary classifiers are more powerful than one-class classifiers.
Figures~\ref{fig:exp-synthetic-1-q0.9}, \ref{fig:exp-synthetic-2-sel-q0.9}, and~\ref{fig:exp-synthetic-2-q0.9} summarize the 90-th empirical quantiles of the lower confidence bounds for the numbers of outliers presented in Figures~\ref{fig:exp-synthetic-1}, \ref{fig:exp-synthetic-2}, and \ref{fig:exp-synthetic-2-sel-q0.5}, respectively.
Figure~\ref{fig:exp-synthetic-3} delves into experiments related to those depicted in Figures~\ref{fig:exp-synthetic-2} and~\ref{fig:exp-synthetic-2-sel-q0.5}, but utilizing synthetic data from a mixture model that bridges between the distributions considered above.
These results highlight ACODE's flexibility, which selects an effective classifier and local testing procedure for each specific case.

Figures~\ref{fig:exp-G-hat-power} and~\ref{fig:exp-G-hat-lb} offer additional insights into the performance of the Shiraishi local testing procedure, outlined in Section~\ref{sec:new_local-g-wmw}, when applied in different settings. These results highlight its practical advantages in situations where the outlier scores are either underdispersed or overdispersed relative to the inlier scores.

\begin{figure}[!htb]
\centering 
\includegraphics[width=0.9\linewidth]{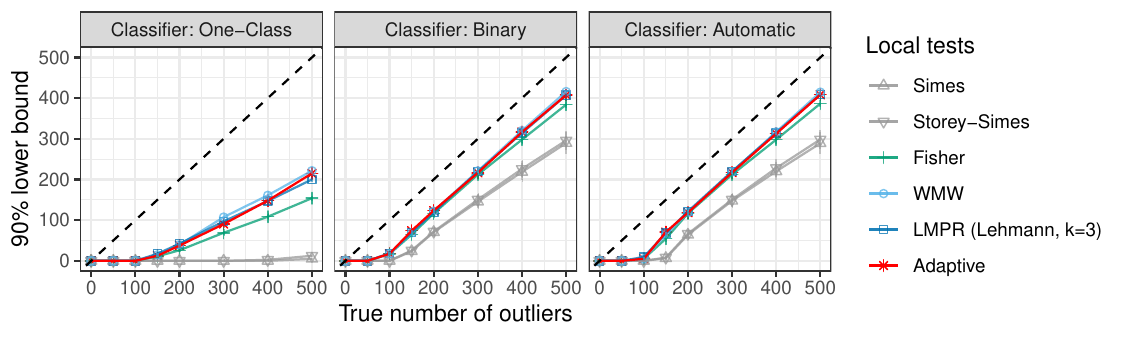}
\caption{Median values for a 90\% lower confidence bound on the total number of outliers in a test set, computed by ACODE on synthetic data based on different classifiers and local testing procedures. The results are shown as a function of the true number of outliers within a test set of size 1000. 
In these experiments, the synthetic data are generated from a binomial model borrowed from \citet{liang2022integrative}.}
\label{fig:exp-synthetic-2}
\end{figure}

\begin{figure}[!htb]
\centering 
\includegraphics[width=0.9\linewidth]{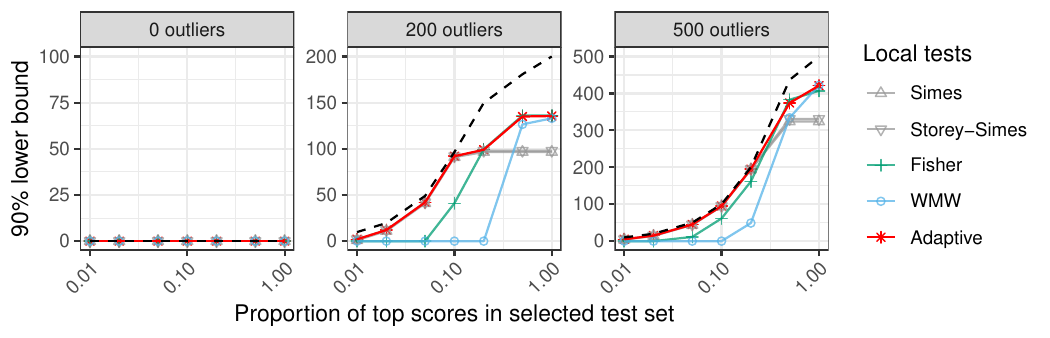}
\caption{Median lower confidence bounds for the number of outliers within an adaptively selected subset of 1000 test points, in numerical experiments otherwise similar to those of Figure~\ref{fig:exp-synthetic-1}.
The results are shown as a function of the proportion of selected test points and of the total number of outliers in the test set.
The dashed curve indicates the true number of outliers in the selected subset, averaged over 100 independent experiments (the true count varies across runs, though with relatively low variance).
In these experiments, ACODE is applied using a one-class support vector classifier to compute the conformity scores.
}
\label{fig:exp-synthetic-2-sel-q0.5}
\end{figure}

\begin{figure}[!htb]
\centering 
\includegraphics[width=0.9\linewidth]{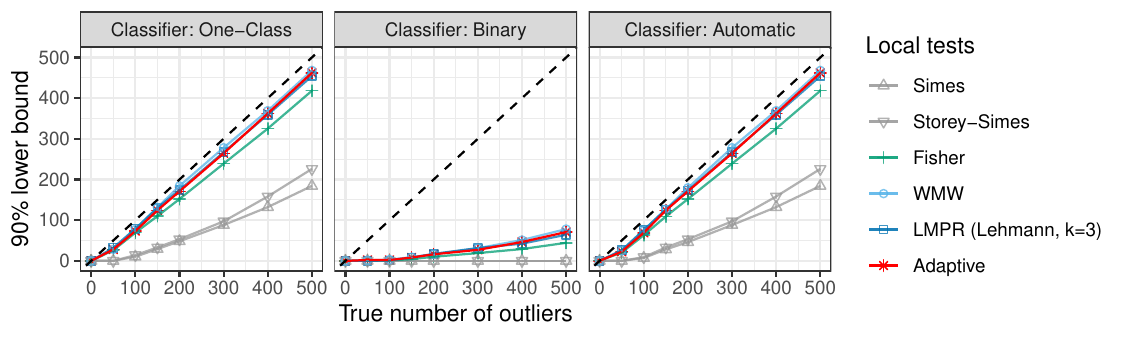}
\caption{Empirical 90-th quantile for a 90\% lower confidence bounds on the total number of outliers in a test set, computed by ACODE on synthetic data based on different classifiers and local testing procedures. Other details are as in Figure~\ref{fig:exp-synthetic-1}.
 }
\label{fig:exp-synthetic-1-q0.9}
\end{figure}

\begin{figure}[!htb]
\centering 
\includegraphics[width=0.9\linewidth]{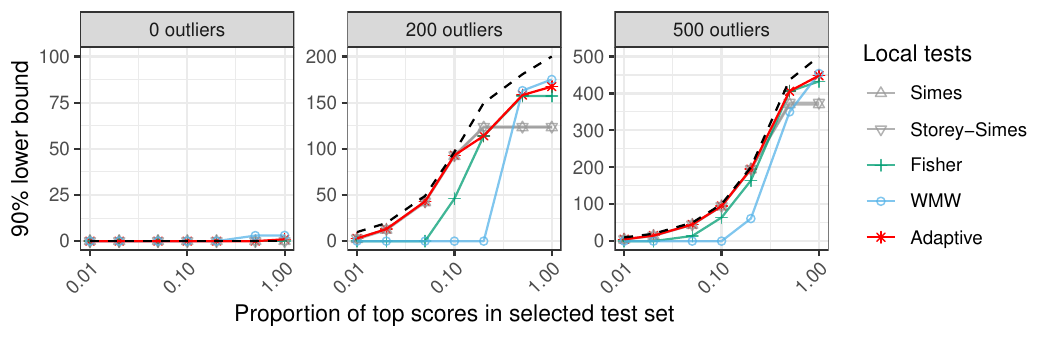}
\caption{Empirical 90-th quantile for a 90\% lower confidence bound for the number of outliers within an adaptively selected subset of 1000 test points, in numerical experiments otherwise similar to those of Figure~\ref{fig:exp-synthetic-1}. Other details are as in Figure~\ref{fig:exp-synthetic-2-sel-q0.5}.
}
\label{fig:exp-synthetic-2-sel-q0.9}
\end{figure}

\begin{figure}[!htb]
\centering 
\includegraphics[width=0.9\linewidth]{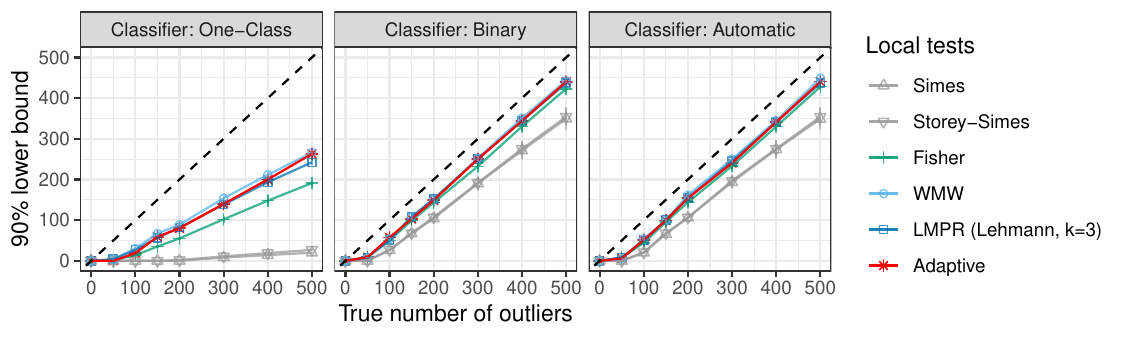}
\caption{Empirical 90-th quantile for a 90\% lower confidence bounds on the total number of outliers in a test set, computed by ACODE on synthetic data based on different classifiers and local testing procedures. Other details are as in Figure~\ref{fig:exp-synthetic-2}.
 }
\label{fig:exp-synthetic-2-q0.9}
\end{figure}

\begin{figure}[!htb]
\centering 
\includegraphics[width=0.9\linewidth]{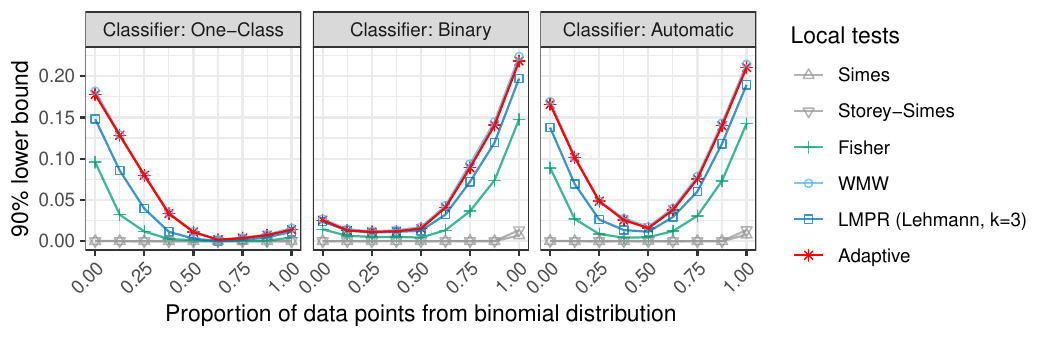}
\caption{Lower confidence bounds for the total number of outliers in a test set, computed by ACODE on synthetic data based on different classifiers and local testing procedures. The results are shown as a function of the true number of outliers within a test set of size 1000. 
In these experiments, the synthetic data are generated from a binomial model borrowed from \citet{liang2022integrative}.
Other details are as in Figure~\ref{fig:exp-synthetic-2}.
 }
\label{fig:exp-synthetic-3}
\end{figure}

\FloatBarrier

\begin{figure}[!htb]
\centering 
\includegraphics[width=\linewidth]{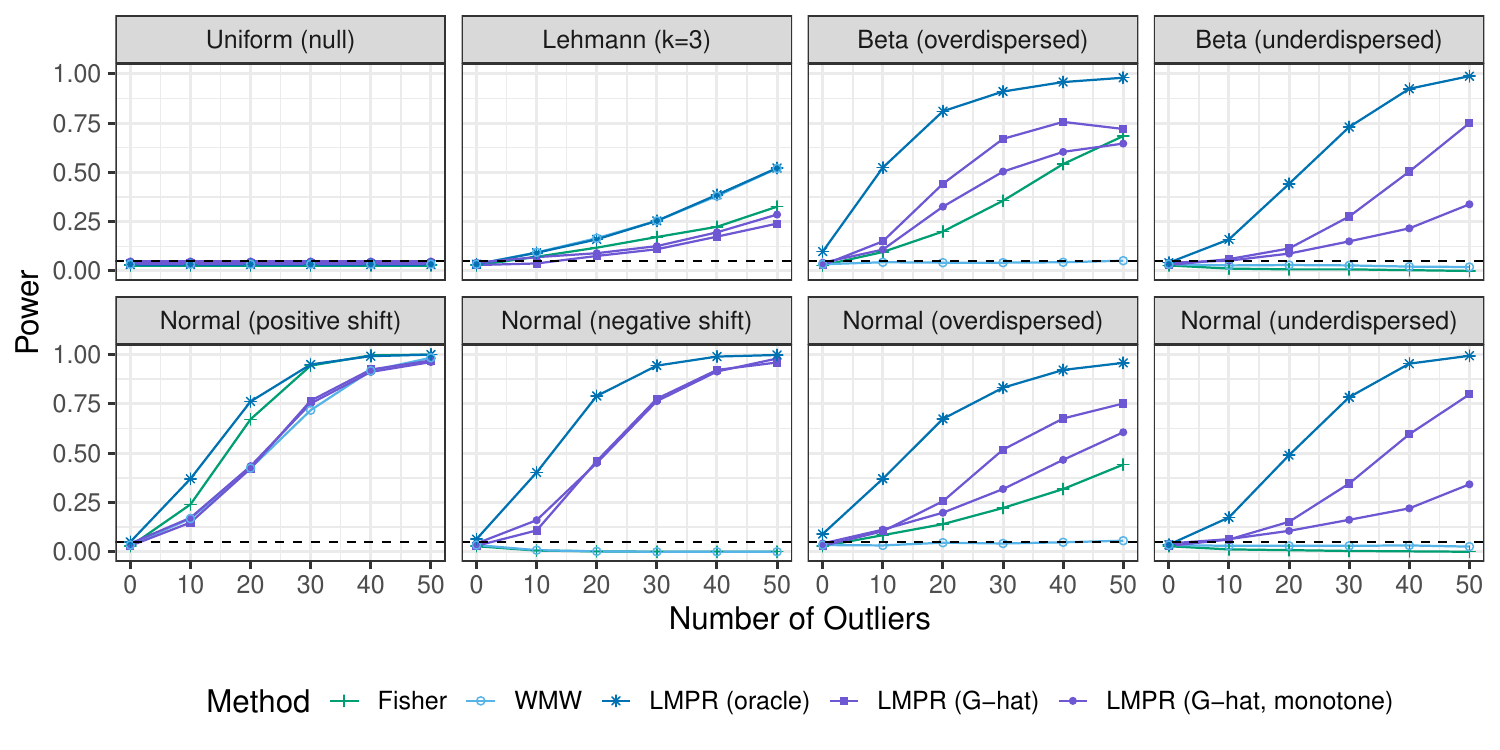}
\caption{Power of different rank tests for the global null hypothesis of no outliers in a test set containing 200 independent observations of a simulated univariate score, using a calibration set containing 500 independent inlier scores. The nominal level is $\alpha=0.05$ (horizontal dashed line).
The results are shown as a function of the true number of outliers in the test set. 
Top: the inliers are randomly sampled from a uniform distribution on $[0,1]$, while the outliers are sampled from the alternative distribution indicated in each panel.
Bottom: the inliers are standard normal, while the outlier scores follow a non-standard normal distribution, as indicated in each panel. Three versions of the Shiraishi test from Section~\ref{sec:new_local-g-wmw} are compared. The first version uses oracle knowledge of the true transformation $G$ linking the outlier distribution to the inlier distribution. The second version of the Shiraishi test uses an empirical estimate of $G$ obtained as described in Appendix~\ref{app:outlier-distr}.
The third version relies on a monotone approximation of the derivative $g = G'$, which (in Figure~\ref{fig:exp-G-hat-lb}) enables a computationally convenient shortcut in the context of closed testing, as also explained in Appendix~\ref{app:outlier-distr}.
 }
\label{fig:exp-G-hat-power}
\end{figure}

\begin{figure}[!htb]
\centering 
\includegraphics[width=\linewidth]{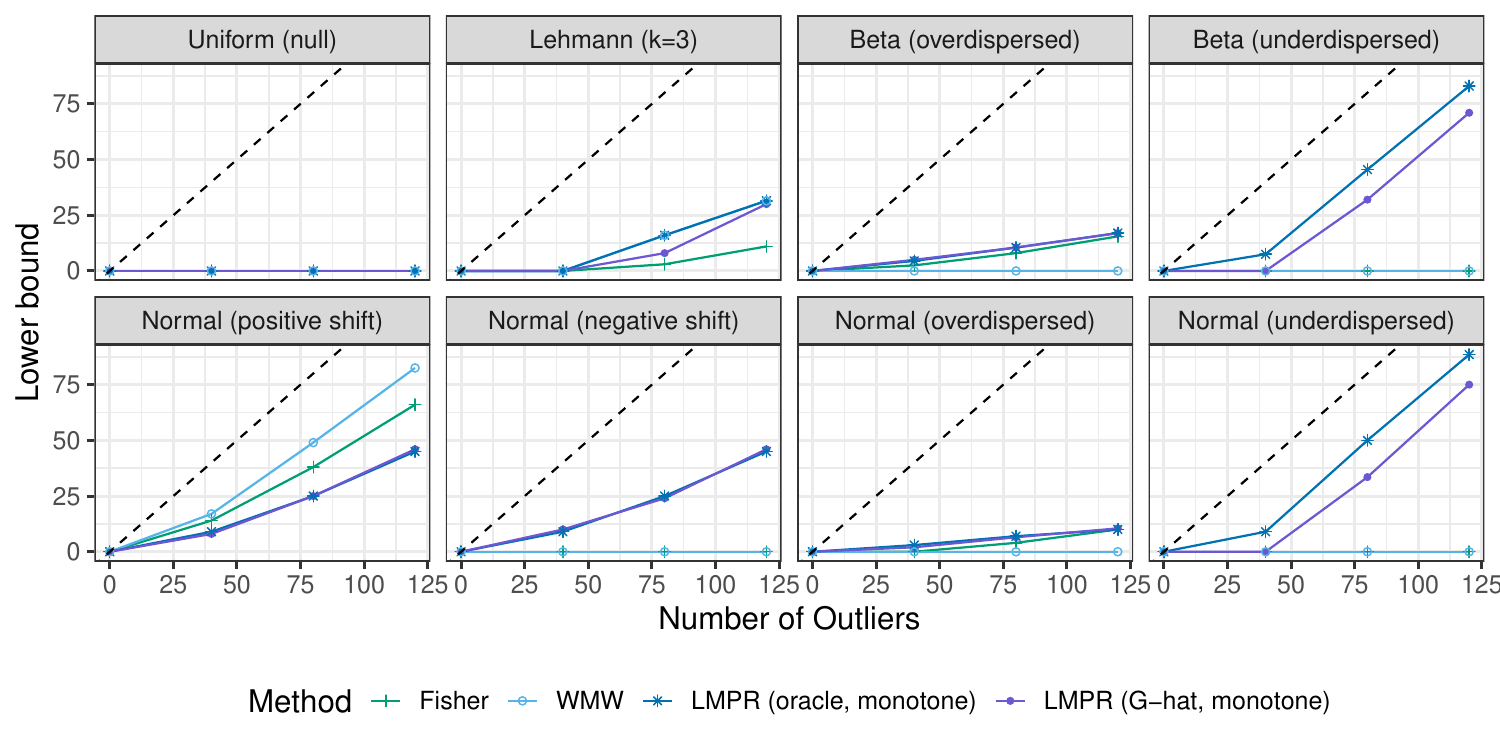}
\caption{
Median values for a 90\% lower confidence bound on the total number of outliers in the test set, in the same experiments of Figure~\ref{fig:exp-G-hat-power}.
The lower bound is calculated through closed testing, using different local testing procedures. For the Shiraishi approach, if the derivative $g$ of the oracle function $G$ is not monotone, it is in practice replaced by a monotone approximation that enables the application of a computationally efficient closed testing shortcut.
 }
\label{fig:exp-G-hat-lb}
\end{figure}

\FloatBarrier

\subsection{Numerical Experiments with Particle Collision Data} \label{app:numerical-lhco}

Figure~\ref{fig:exp-lhco-full} is a more detailed version of Figure~\ref{fig:exp-lhco}, including a broader range for the choice of the local test.
\begin{figure}[!htb]
\centering 
\includegraphics[width=\linewidth]{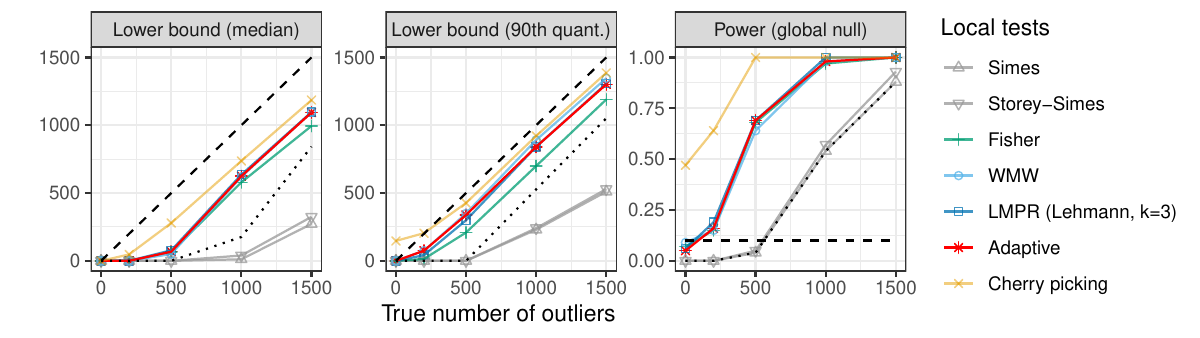}
\caption{Performance of ACODE for collective outlier detection with the LHCO data, in the experiments of Figure~\ref*{fig:exp-lhco}.
Our method utilizes a local testing procedure that may be adaptively selected (red curve) or fixed (other solid curves).
Compared to Figure~\ref*{fig:exp-lhco}, these results include a more detailed comparison of the performance of ACODE applied using different choices of local testing procedures.
}
\label{fig:exp-lhco-full}
\end{figure}

Figure~\ref{fig:exp-lhco-sel} presents results from experiments in which ACODE is used to construct 90\% lower confidence bounds for the number of outliers in a data-driven subset of test points, selected as those with the largest conformity scores. To reduce computational cost, we consider smaller test sets of 2000 data points with varying proportions of outliers, and we do not apply ACODE with the Shiraishi local test, whose implementation becomes relatively expensive when the selected set differs from the full test set. As in Figure~\ref{fig:exp-synthetic-1-sel-q0.5}, conformity scores are computed using a single classifier (AdaBoost) to facilitate interpretation. The results confirm that the adaptive version of ACODE again approximately maximizes power by automatically selecting the testing procedure.

\begin{figure}[!htb]
\centering
\includegraphics[width=0.9\linewidth]{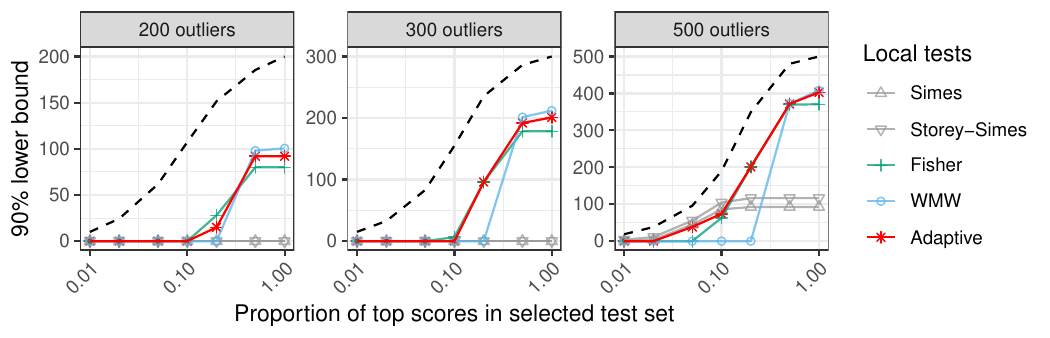}
\caption{Median values over repeated experiments for a 90\% lower confidence bound on the number of outliers within an adaptively selected subset of 1000 test points, in numerical experiments similar to those of Figure~\ref*{fig:exp-synthetic-1}.
The results are shown as a function of the proportion of selected test points and of the total number of outliers.
The dashed curve indicates the true number of outliers in the selected subset, averaged over 100 independent experiments (the true count varies across runs, though with relatively low variance).
}
\label{fig:exp-lhco-sel}
\end{figure}

Table~\ref{tab:lhco-lb} provides a more detailed view of the results in Figure~\ref{fig:exp-lhco-full}.

\begin{table}[!htb]
\centering
\resizebox{\textwidth}{!}{
\input{tables/tabA4_lhco_nt10000_nc2000.tex}
}
\caption{Performance of ACODE for collective outlier detection with the LHCO data, as a function of the true number of outliers. 
The performance is measured in terms of median and 90-th percentile values of a 90\% lower confidence bound for the number of outliers (top and center) and the power to reject the global null hypothesis of no outliers (bottom). The numbers in parenthesis are standard errors. Other details are as in Figures~\ref*{fig:exp-lhco} and~\ref{fig:exp-lhco-full}. } 
\label{tab:lhco-lb}
\end{table}

\FloatBarrier

\subsection{Additional Experiments with Real Data} \label{app:numerical-real}

In this section, we further investigate the performance of ACODE by applying it to 6 additional data sets previously utilized by in the related literature \citep{bates2021testing,liang2022integrative,marandon2022machine}.
These data sets are:
\begin{itemize}
  \item \texttt{ALOI} \citep{aloi}, containing 27 variables, 1508 outliers, 48026 inliers. Previously used by \citet{bates2021testing}.
  \item \texttt{Covertype} \citep{cover}, containing 10 variables, 2747 outliers, 286048 inliers. Previously used by \citet{bates2021testing}.
  \item \texttt{CreditCard} \citep{creditcard}, containing 30 variables,  492 outliers, 284315 inliers. Previously used by \citet{bates2021testing} and \citet{marandon2022machine}.
  \item \texttt{Mammography} \citep{mammography}, containing 6 variables, 260 outliers, 10923 inliers. Previously used by \citet{bates2021testing}, \citet{liang2022integrative} and \citet{marandon2022machine}.
  \item \texttt{Pendigits} \citep{pendigits}, containing 16 variables, 156 outliers, 6714 inliers. Previously used by \citet{bates2021testing}.
  \item \texttt{Shuttle} \citep{shuttle}, containing 9 variables, 3511 outliers, 45586 inliers. Previously used by \citet{bates2021testing} and \citet{marandon2022machine}.
\end{itemize}

These experiments are conducted following an approach similar to that described in Section~\ref{sec:numerical-synthetic}, applying our method based on the same suite of 6 classification algorithms and 6 local testing procedures.
Each experiment is independently repeated 100 times, each time randomly sampling disjoint training, calibration, tuning, and test sets of sizes 1000, 100, 100, and 100, respectively.
The training, calibration, and tuning sets include only inliers, while the proportion of outliers in the test set is varied as a control parameter between 0 and 1.

Figures~\ref{fig:exp-4-data-credit-card}--\ref{fig:exp-4-data-aloi} summarize, separately for each data set, the performances of the lower confidence bounds and global tests obtained with our method.
The results are reported as a function of the true number of outliers in the test set and are also stratified based on the type of classifier utilized by ACODE and on the underlying local testing procedure.
The findings overall indicate that binary classification models generally yield higher power in these experiments, which can be attributed to the relatively low dimensionality of these data sets.
Moreover, among the local testing methods considered, both Fisher's method and the WMW method demonstrate similarly high power, typically outperforming Simes' method.
Crucially, the fully automatic implementation of ACODE, which selects both the classification model and the local testing procedure in a data-adaptive way, leads to near-oracle performance in all cases.

Finally, Figures~\ref{fig:exp-data-sel} and~\ref{fig:exp-data-sel-q0.9} describe experiments aimed at constructing 90\% lower confidence bounds for the number of outliers within a data-dependent subset of test points, specifically those selected for their high conformity scores. For easier understanding of these experiments, ACODE is now applied with a fixed classification algorithm, the one-class isolation forest, rather than leveraging the full suite of 6 different algorithms.
The results show that the WMW test consistently outperforms the Fisher combination method in terms of power, aligning with findings previously documented in Sections~\ref{sec:numerical-synthetic} and~\ref{sec:numerical-lhco}.
Furthermore, the fully adaptive implementation of ACODE effectively maximizes power by automatically identifying the most powerful local testing procedure for each case, again achieving oracle-like performance.

\begin{figure}[!htb]
\centering 
\includegraphics[width=0.9\linewidth]{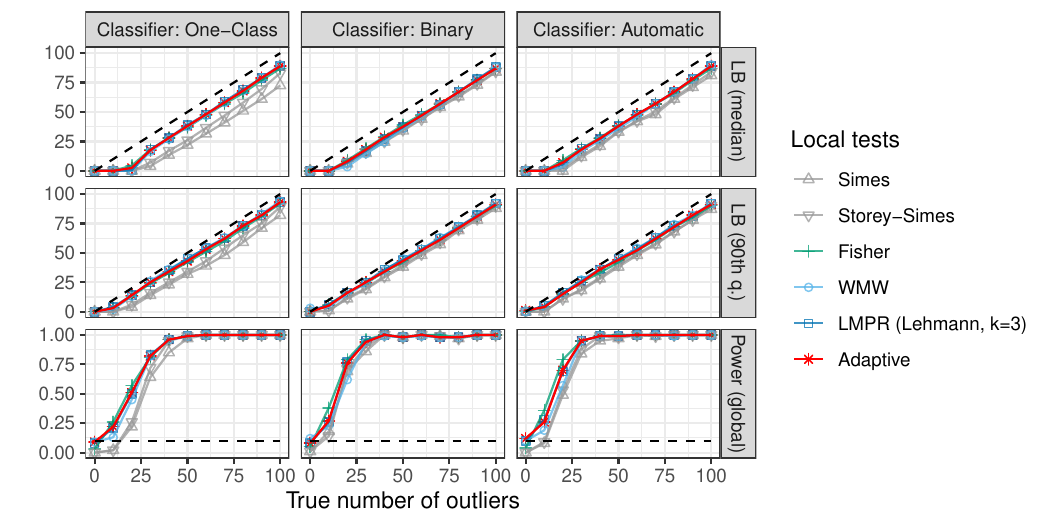}
\caption{Performance of ACODE for collective outlier detection with the \texttt{creditcard} data set, as a function of the true number of outliers in a test set of size 100.
  The performance is measured in terms of median and 90-th percentile values of a 90\% lower confidence bound for the number of outliers (top and center) and the power to reject the global null hypothesis of no outliers (bottom). Other details are as in Figure~\ref{fig:exp-synthetic-1}.
}
\label{fig:exp-4-data-credit-card}
\end{figure}

\begin{figure}[!htb]
\centering 
\includegraphics[width=0.9\linewidth]{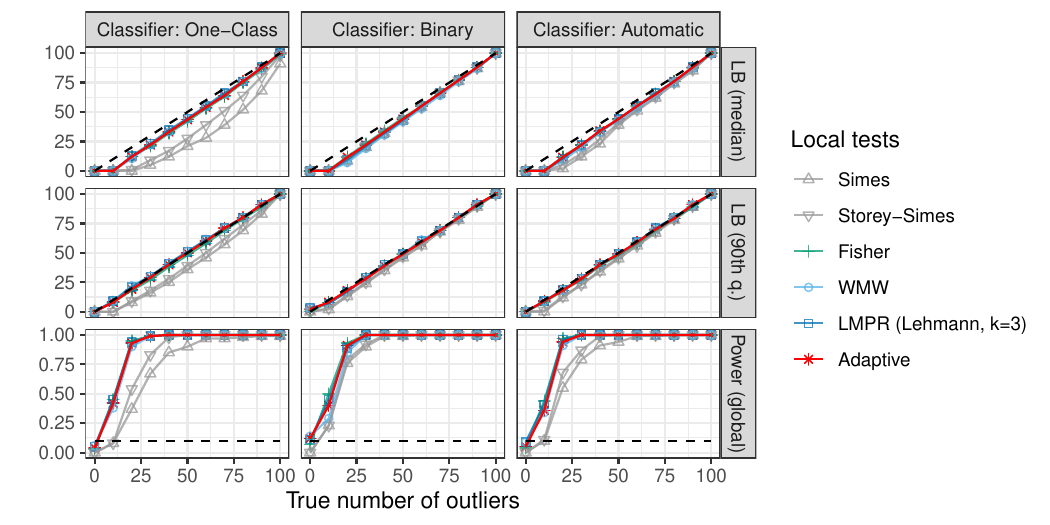}
\caption{Performance of ACODE for collective outlier detection with the \texttt{pendigits} data set, as a function of the true number of outliers in a test set of size 100.
  The performance is measured in terms of median and 90-th percentile values of a 90\% lower confidence bound for the number of outliers (top and center) and the power to reject the global null hypothesis of no outliers (bottom). Other details are as in Figure~\ref{fig:exp-4-data-credit-card}.
}
\label{fig:exp-4-data-pendigits}
\end{figure}

\begin{figure}[!htb]
\centering 
\includegraphics[width=0.9\linewidth]{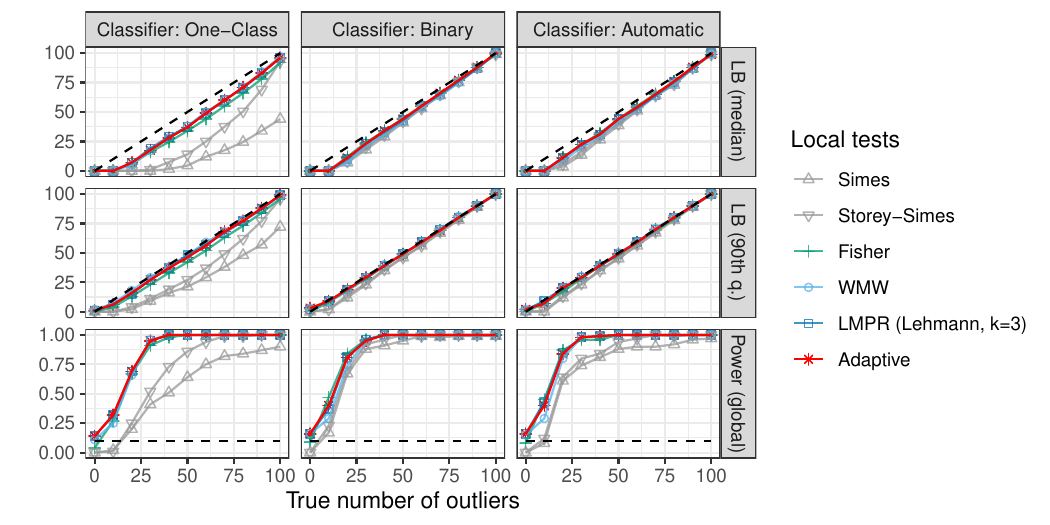}
\caption{Performance of ACODE for collective outlier detection with the \texttt{covertype} data set, as a function of the true number of outliers in a test set of size 100.
  The performance is measured in terms of median and 90-th percentile values of a 90\% lower confidence bound for the number of outliers (top and center) and the power to reject the global null hypothesis of no outliers (bottom). Other details are as in Figure~\ref{fig:exp-4-data-credit-card}.
}
\label{fig:exp-4-data-cover}
\end{figure}

\begin{figure}[!htb]
\centering 
\includegraphics[width=0.9\linewidth]{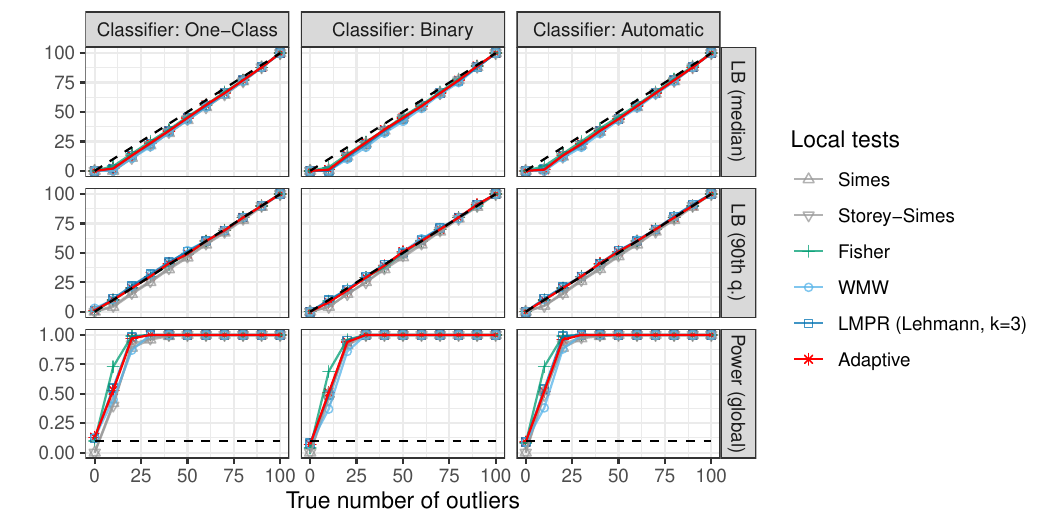}
\caption{Performance of ACODE for collective outlier detection with the \texttt{shuttle} data set, as a function of the true number of outliers in a test set of size 100.
  The performance is measured in terms of median and 90-th percentile values of a 90\% lower confidence bound for the number of outliers (top and center) and the power to reject the global null hypothesis of no outliers (bottom). Other details are as in Figure~\ref{fig:exp-4-data-credit-card}.
}
\label{fig:exp-4-data-shuttle}
\end{figure}

\begin{figure}[!htb]
\centering 
\includegraphics[width=0.9\linewidth]{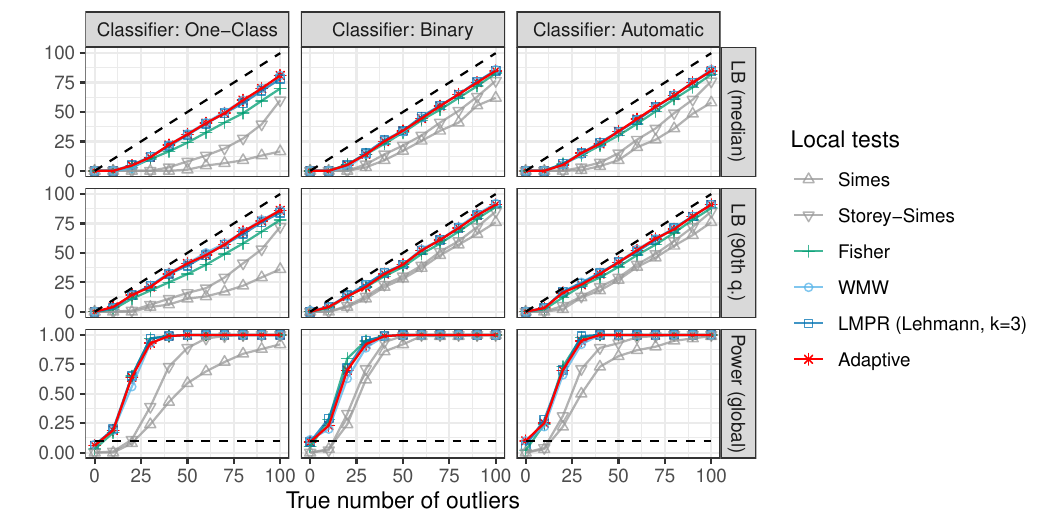}
\caption{Performance of ACODE for collective outlier detection with the \texttt{mammography} data set, as a function of the true number of outliers in a test set of size 100.
  The performance is measured in terms of median and 90-th percentile values of a 90\% lower confidence bound for the number of outliers (top and center) and the power to reject the global null hypothesis of no outliers (bottom). Other details are as in Figure~\ref{fig:exp-4-data-credit-card}.
}
\label{fig:exp-4-data-mammography}
\end{figure}

\begin{figure}[!htb]
\centering 
\includegraphics[width=0.9\linewidth]{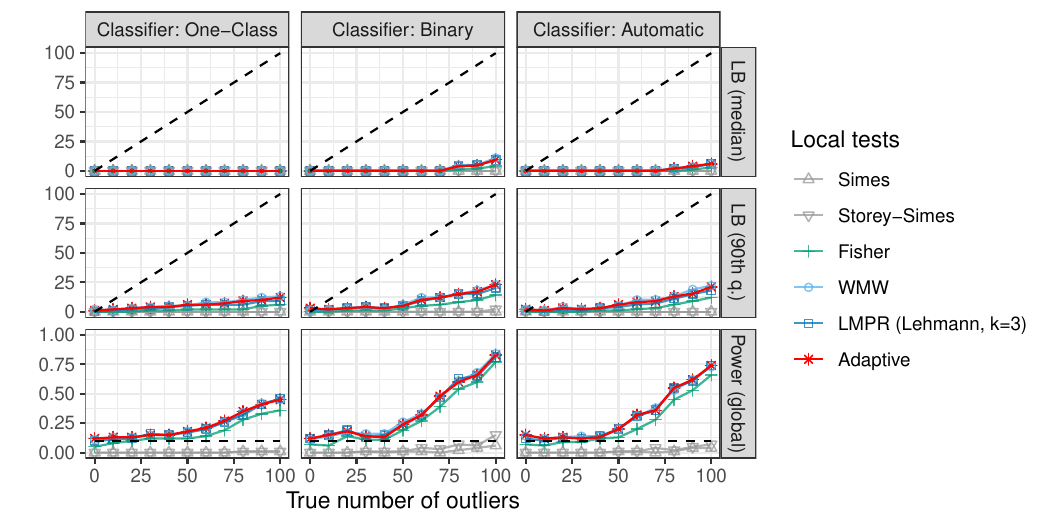}
\caption{Performance of ACODE for collective outlier detection with the \texttt{aloi} data set, as a function of the true number of outliers in a test set of size 100.
  The performance is measured in terms of median and 90-th percentile values of a 90\% lower confidence bound for the number of outliers (top and center) and the power to reject the global null hypothesis of no outliers (bottom). Other details are as in Figure~\ref{fig:exp-4-data-credit-card}.
}
\label{fig:exp-4-data-aloi}
\end{figure}

\begin{figure}[!htb]
\centering 
\includegraphics[width=0.9\linewidth]{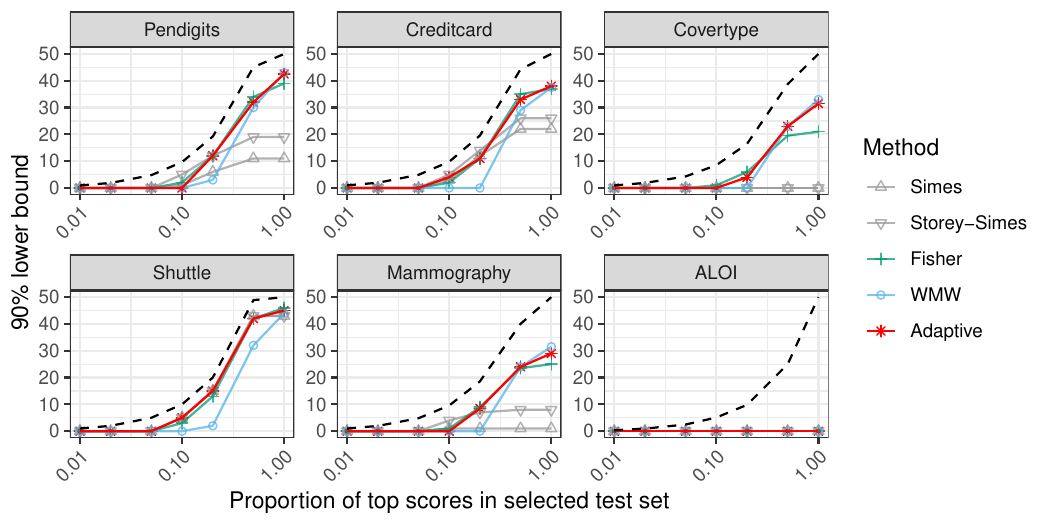}
\caption{Median values for a lower confidence bound on the number of outliers within an adaptively selected subset of 1000 test points, in numerical experiments with several different data sets.
The results are shown as a function of the proportion of selected test points.
In these experiments, ACODE is applied using a one-class isolation forest model to compute the conformity scores.
Other details are as in Figure~\ref{fig:exp-synthetic-1-sel-q0.5}.
}
\label{fig:exp-data-sel}
\end{figure}

\begin{figure}[!htb]
\centering 
\includegraphics[width=0.9\linewidth]{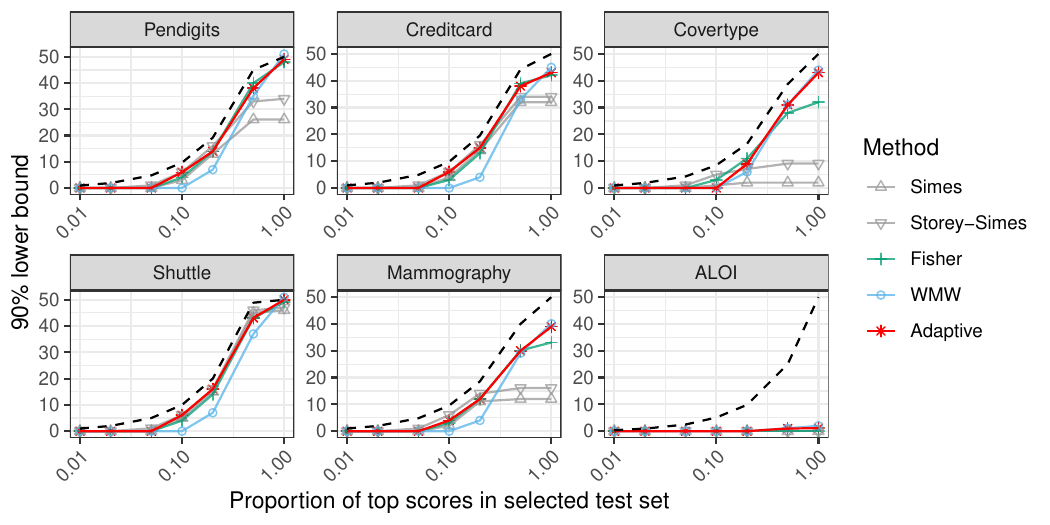}
\caption{Empirical 90-th quantile for a lower confidence bound on the number of outliers within an adaptively selected subset of 1000 test points, in numerical experiments with several different data sets.
Other details are as in Figure~\ref{fig:exp-data-sel}.
}
\label{fig:exp-data-sel-q0.9}
\end{figure}

\FloatBarrier

\subsection{\cite{ludbrook1998permutation} Example} \label{app:ludbrook-example}

Analogously to signal detection, where—depending on sparsity and signal strength—a signal may be estimable, detectable, or undetectable (see Fig. 1 of \citet{donoho2004higher}), our problem admits a similar hierarchy. Outliers may be individually identifiable, enumerable (and thus detectable in aggregate), or undetectable. These inferential tasks form a strict ordering: detection is easier than enumeration, which in turn is easier than identification. Less demanding tasks yield higher power but provide coarser information, reflecting an inherent trade-off between power and informativeness. 

Our closed testing–based approach navigates this trade-off: when effects are detectable, it enables enumeration; when they cannot be pinpointed, it allows them to be progressively bracketed within subsets of observations.

To illustrate this point, we consider a simple example based on the data of \citet{ludbrook1998permutation}. This example shows that even when individual-level detection fails, in particular, when no single 
$p$-value falls below the significance level 
$\alpha$, group-level inference can still reveal the presence of outliers. Specifically, the closed testing procedure yields a lower bound on the number of outliers and, , by examining subsets of observations, allows them, even if they cannot be pinpointed individually, to be bracketed within a smaller group of observations.

\citet[section 4.3]{ludbrook1998permutation} reports the calibration and test scores:
    \begin{eqnarray*}
    (x_1,x_2,x_3,x_4,x_5,x_6,x_7) &=& (5.42, 5.86, 6.16, 6.55, 6.8, 7, 7.11),\\
    (y_1,y_2,y_3,y_4,y_5) &=& (6.51, 7.56, 7.61, 7.84, 11.5).
    \end{eqnarray*}
    
    The resulting conformal $p$-values are:
    \[
    (p_1,p_2,p_3,p_4,p_5) = (0.625, 0.125, 0.125, 0.125, 0.125).
    \]
    
    At the nominal level $\alpha = 0.05$, all $p$-values exceed $\alpha$, and neither the Benjamini–Hochberg nor the Storey procedure rejects any hypotheses. In contrast, the WMW rank test rejects the global null hypothesis  that there are no outliers with a $p$-value of 0.03, as reported in \cite{ludbrook1998permutation}. This suggests that, although individual outliers may not be identified, a group of outliers may still be detected.
    
    Using the closed testing procedure based on local WMW tests, we obtain an overall bound of $d = 2$. This means that, with 95\% confidence, at least two of the five test observations are outliers, although we cannot determine exactly which ones.
    
    To narrow down location of the outliers, we can examine subsets of $[n]$.  For example, $d(S) = 2$ for $S = \{2, 3, 4, 5\}$, indicating that at least two outliers are among the four observations with the largest test scores. Zooming in further, for $S'=\{3,4,5\}$ and $S''=\{4,5\}$ we obtain $d(S')=1$ and $d(S'')=0$, indicating that while at least one outlier lies among the three highest-scoring test observations, the outliers cannot be bracketed further. 
    

    Note that if interest is restricted to an a priori defined subset $R$ of the test observations, then applying the closed testing procedure directly to $R$ is at least as powerful, and typically more powerful, than applying the procedure to the full set $[n]$ and subsequently focusing on $R$.
    For example, if the set $R = \{2, 3, 4, 5\}$ had been specified a priori as the hypotheses of interest, the resulting lower bound would improve from 2 to 3. 
    Choosing a priori a subset does not always lead to an improvement. For example, the set $\{1, 2, 3, 4\}$ would result in the same lower bound of 1.

\FloatBarrier


%% file: tables/tabA4_lhco_nt10000_nc2000.tex
\begin{tabular}{cccccccc}
\toprule
\multicolumn{1}{c}{ } & \multicolumn{7}{c}{Local testing procedure} \\
\cmidrule(l{3pt}r{3pt}){2-8}
Outliers & Simes & Storey-Simes & Fisher & WMW & LMPR (Lehmann, k=3) & Adaptive & Cherry picking\\
\midrule
\addlinespace[0.3em]
\multicolumn{8}{l}{\textbf{90\% Lower bound (median)}}\\
\hspace{1em}0 & 0 (0) & 0 (0) & 0 (2) & 0 (5) & 0 (4) & 0 (4) & 0 (8)\\
\hspace{1em}200 & 0 (0) & 0 (0) & 0 (2) & 0 (5) & 0 (4) & 0 (5) & 48 (9)\\
\hspace{1em}500 & 0 (1) & 0 (1) & 65 (9) & 50 (14) & 78 (13) & 68 (13) & 279 (12)\\
\hspace{1em}1000 & 12 (9) & 40 (10) & 579 (21) & 620 (23) & 636 (21) & 624 (21) & 737 (15)\\
\hspace{1em}1500 & 272 (18) & 324 (18) & 994 (20) & 1096 (22) & 1096 (20) & 1094 (20) & 1186 (15)\\
\addlinespace[0.3em]
\multicolumn{8}{l}{\textbf{90\% Lower bound (90-th quantile)}}\\
\hspace{1em}0 & 0 (0) & 0 (0) & 0 (2) & 0 (5) & 0 (4) & 0 (4) & 147 (8)\\
\hspace{1em}200 & 0 (0) & 0 (0) & 14 (2) & 75 (5) & 42 (4) & 75 (5) & 204 (9)\\
\hspace{1em}500 & 0 (1) & 0 (1) & 209 (9) & 350 (14) & 299 (13) & 338 (13) & 427 (12)\\
\hspace{1em}1000 & 230 (9) & 241 (10) & 700 (21) & 891 (23) & 841 (21) & 841 (21) & 921 (15)\\
\hspace{1em}1500 & 510 (18) & 529 (18) & 1190 (20) & 1346 (22) & 1307 (20) & 1301 (20) & 1386 (15)\\
\addlinespace[0.3em]
\multicolumn{8}{l}{\textbf{Power (global null)}}\\
\hspace{1em}0 & 0.00 (0.00) & 0.00 (0.00) & 0.07 (0.03) & 0.09 (0.03) & 0.06 (0.02) & 0.05 (0.02) & 0.47 (0.05)\\
\hspace{1em}200 & 0.00 (0.00) & 0.00 (0.00) & 0.15 (0.04) & 0.15 (0.04) & 0.19 (0.04) & 0.16 (0.04) & 0.64 (0.05)\\
\hspace{1em}500 & 0.04 (0.02) & 0.05 (0.02) & 0.68 (0.05) & 0.64 (0.05) & 0.69 (0.05) & 0.69 (0.05) & 1.00 (0.00)\\
\hspace{1em}1000 & 0.54 (0.05) & 0.57 (0.05) & 0.97 (0.02) & 0.98 (0.01) & 1.00 (0.00) & 0.98 (0.01) & 1.00 (0.00)\\
\hspace{1em}1500 & 0.88 (0.03) & 0.93 (0.03) & 1.00 (0.00) & 1.00 (0.00) & 1.00 (0.00) & 1.00 (0.00) & 1.00 (0.00)\\
\bottomrule
\end{tabular}